\documentclass[11pt]{article}
\usepackage{amsmath,amsthm,amssymb}
\usepackage{fullpage}
\usepackage{authblk}
\usepackage{hyperref}
\usepackage[usenames,dvipsnames,svgnames,table]{xcolor}
\usepackage{graphicx}
\usepackage{caption}
\usepackage{subcaption}

\newtheorem{theorem}{Theorem}
\newtheorem{problem}{Problem}
\newtheorem{corollary}[theorem]{Corollary}
\newtheorem{lemma}{Lemma}
\newtheorem{proposition}{Proposition}

\newsavebox{\mybox}
\newenvironment{boxedproblem}[2]
{\begin{lrbox}{\mybox}\begin{minipage}{0.95\textwidth} \begin{problem}[\textsc{#1}]\label{#2} \  }
{\end{problem} \end{minipage}\end{lrbox} \begin{center}\fbox{\usebox{\mybox}} \end{center}}

\newcommand{\qedclaim}{\hfill $\diamond$ \medskip}

\newenvironment{proofof}[1]{\medskip\noindent\emph{Proof of #1. }\ignorespaces}{\hfill\qed\medskip\par\noindent\ignorespacesafterend}

\theoremstyle{definition}
\newtheorem{definition}{Definition}

\newcounter{propertycounter}

\newcommand{\todo}[1]{}
\renewcommand{\todo}[1]{{\color{red} TODO: {#1}}}

\begin{document}

\title{Fully polynomial FPT algorithms for some classes of bounded clique-width graphs\thanks{This work
    has been partially supported by ANR project Stint under reference
    ANR-13-BS02-0007 and ANR program ``Investments for the Future'' under reference
    ANR-11-LABX-0031-01.}}
\date{}

\author[1]{David Coudert}
\author[1,2,3]{Guillaume Ducoffe}
\author[2,4]{Alexandru Popa}

\affil[1]{\small Universit\'e C\^ote d'Azur, Inria, CNRS, I3S, France}
\affil[2]{\small National Institute for Research and Development in Informatics, Romania}
\affil[3]{\small The Research Institute of the University of Bucharest ICUB, Romania}
\affil[4]{\small University of Bucharest, Faculty of Mathematics and Computer Science}

\maketitle

\begin{abstract}
Parameterized complexity theory has enabled a refined classification of the difficulty of NP-hard optimization problems on graphs with respect to key structural properties, and so to a better understanding of their true difficulties.  
More recently, hardness results for problems in P were achieved using reasonable complexity theoretic assumptions such as: Strong Exponential Time Hypothesis (SETH), 3SUM and All-Pairs Shortest-Paths (APSP). 
According to these assumptions, many graph theoretic problems do not admit truly subquadratic algorithms, nor even truly subcubic algorithms (Williams and Williams, FOCS 2010 and Abboud, Grandoni, Williams, SODA 2015).
%
A central technique used to tackle the difficulty of the above mentioned problems is fixed-parameter algorithms for polynomial-time problems with {\em polynomial dependency} in the fixed parameter (P-FPT). 
This technique was rigorously formalized by Giannopoulou et al. (IPEC 2015).
Following that, it was continued by Abboud, Williams and Wang in SODA 2016, by Husfeldt (IPEC 2016) and  Fomin et al. (SODA 2017), using the treewidth as a parameter.
Applying this technique to {\em clique-width}, another important graph parameter, remained to be done.

In this paper we study several graph theoretic problems for which hardness results exist such as {\em cycle problems} (triangle detection, triangle counting, girth, diameter), {\em distance problems} (diameter, eccentricities, Gromov hyperbolicity, betweenness centrality) and {\em maximum matching}. 
We provide hardness results and fully polynomial FPT algorithms, using clique-width and some of its upper-bounds as parameters (split-width, modular-width and $P_4$-sparseness). 
We believe that our most important result is an $\mathcal{O}(k^4 \cdot n + m)$-time algorithm for computing a maximum matching where $k$ is either the modular-width or the $P_4$-sparseness. 
The latter generalizes many algorithms that have been introduced so far for specific subclasses such as cographs, $P_4$-lite graphs, $P_4$-extendible graphs and $P_4$-tidy graphs.

Our algorithms are based on preprocessing methods using modular decomposition, split decomposition and primeval decomposition.
Thus they can also be generalized to some graph classes with unbounded clique-width.  

\end{abstract}

\section{Introduction}

The classification of problems according to their complexity is one of the main goals in computer science.
This goal was partly achieved by the theory of NP-completeness which helps to identify the problems that are unlikely to have polynomial-time algorithms.
However, there are still many problems in P for which it is not known if the running time of the best current algorithms can be improved.
Such problems arise in various domains such as computational geometry, string matching or graphs.
Here we focus on the existence and the design of {\em linear-time} algorithms, for solving several graph problems when restricted to classes of bounded {\em clique-width}.
The problems considered comprise the detection of short cycles ({\it e.g.}, {\sc Girth} and {\sc Triangle Counting}), some distance problems ({\it e.g.}, {\sc Diameter}, {\sc Hyperbolicity}, {\sc Betweenness Centrality}) and the computation of maximum matchings in graphs.
We refer to Sections~\ref{sec:pbs-cw},~\ref{sec:dist-pbs} and~\ref{sec:maxmatching}, respectively, for a recall of their definitions.

Clique-width is an important graph parameter in structural graph theory, that intuitively represents the closeness of a graph to a cograph --- {\it a.k.a.}, $P_4$-free graphs~\cite{CPS85,CoB00}.
Some classes of perfect graphs, including distance-hereditary graphs, and so, trees, have bounded clique-width~\cite{GoR00}.
Furthermore, clique-width has many algorithmic applications.
Many algorithmic schemes and metatheorems have been proposed for classes of bounded clique-width~\cite{CMR00,Cou12,EGW01}.
Perhaps the most famous one is Courcelle's theorem, that states that every graph problem expressible in Monadic Second Order logic ($MSO_1$) can be solved in $f(k) \cdot n$-time when restricted to graphs with clique-width at most $k$, for some computable function $f$ that only depends on $k$~\cite{CMR00}.
Some of the problems considered in this work can be expressed as an $MSO_1$ formula. 
However, the dependency on the clique-width in Courcelle's theorem is super-polynomial, that makes it less interesting for the study of graphs problems in P.
Our goal is to derive a {\em finer-grained} complexity of polynomial graph problems when restricted to classes of bounded clique-width, that requires different tools than Courcelle's theorem.

\medskip
Our starting point is the recent theory of ``Hardness in P'' that aims at better hierarchizing the complexity of polynomial-time solvable problems~\cite{VaW15}. 
This approach mimics the theory of NP-completeness.
Precisely, since it is difficult to obtain unconditional hardness results, it is natural to obtain hardness results assuming some complexity theoretic conjectures. 
In other words, there are key problems that are widely believed not to admit better algorithms such as 3-SAT (k-SAT), 3SUM and All-Pairs Shortest Paths (APSP).
Roughly, a problem in P is hard if the existence of a faster algorithm for this problem implies the existence of a faster algorithm for one of these fundamental problems mentioned above.
In their seminal work, Williams and Williams~\cite{VaW10} prove that many important problems in graph theory are all equivalent under subcubic reductions.
That is, if one of these problems admits a truly sub-cubic algorithms, then all of them do.
Their results have extended and formalized prior work from, {\it e.g.},~\cite{GaO95,KrS06}.
The list of such problems was further extended in~\cite{AGV15,BCH16}.

Besides purely negative results ({\it i.e.}, conditional lower-bounds) the theory of ``Hardness in P'' also comes with renewed algorithmic tools in order to leverage the existence, or the nonexistence, of improved algorithms for some graph classes.
The tools used to improve the running time of the above mentioned problems are similar to the ones used to tackle NP-hard problems, namely approximation and FPT algorithms.
Our work is an example of the latter, of which we first survey the most recent results.

\paragraph{Related work: Fully polynomial parameterized algorithms.}
FPT algorithms for polynomial-time solvable problems were first considered by Giannopoulou et al.~\cite{GMN17}.
Such a parameterized approach makes sense for any problem in P for which a conditional hardness result is proved, or simply no linear-time algorithms are known. 
Interestingly, the authors of~\cite{GMN17} proved that a matching of cardinality at least $k$ in a graph can be computed in ${\cal O}(kn + k^3)$-time.  
We stress that {\sc Maximum Matching} is a classical and intensively studied problem in computer science~\cite{DuP14,FGV99,FPT97,GaT83,KaS81,MiV80,MNN16,YuY93}. 
The well known ${\cal O} (m\sqrt{n})$-time algorithm in~\cite{MiV80} is essentially the best so far for {\sc Maximum Matching}. 
Approximate solutions were proposed by Duan and Pettie~\cite{DuP14}.

More related to our work is the seminal paper of Abboud, Williams and Wang~\cite{AVW16}.
They obtained rather surprising results when using {\em treewidth}: another important graph parameter that intuitively measures the closeness of a graph to a tree~\cite{Bod06}.
Treewidth has tremendous applications in pure graph theory~\cite{RoS86} and parameterized complexity~\cite{Cou90}.
Furthermore, improved algorithms have long been known for ''hard'' graph problems in P, such as {\sc Diameter} and {\sc Maximum Matching}, when restricted to trees~\cite{Jor69}.
However, it has been shown in~\cite{AVW16} that under the Strong Exponential Time Hypothesis, for any $\varepsilon >0$ there can be no $2^{o(k)} \cdot n^{2-\varepsilon}$-time algorithm for computing the diameter of graphs with treewidth at most $k$.
This hardness result even holds for {\em pathwidth}, that leaves little chance to find an improved algorithm for any interesting subclass of bounded-treewidth graphs while avoiding an exponential blow-up in the parameter.
We show that the situation is different for clique-width than for treewidth, in the sense that the hardness results for clique-width do not hold for important subclasses.

We want to stress that a familiar reader could ask why the hardness results above do not apply to clique-width directly since it is upper-bounded by a function of treewidth~\cite{CoR05}.
However, clique-width cannot be {\em polynomially} upper-bounded by the treewidth~\cite{CoR05}.
Thus, the hardness results from~\cite{AVW16} do not preclude the existence of, say, an ${\cal O}(k n)$-time algorithm for computing the diameter of graphs with clique-width at most $k$.

\medskip
On a more positive side, the authors in~\cite{AVW16} show that {\sc Radius} and {\sc Diameter} can be solved in $2^{{\cal O}(k\log{k})} \cdot n^{1+{\cal O}(1)}$-time, where $k$ is treewidth.
Husfeldt~\cite{Hus16} shows that the eccentricity of every vertex in an undirected graph on $n$ vertices can be computed in time $n \cdot \text{exp}\left[{\cal O}(k \log d) \right]$, where $k$ and $d$ are the treewidth and the diameter of the graph, respectively.
More recently, a tour de force was achieved by Fomin et al.~\cite{FLPS+17} who were the first to design parameterized algorithms with {\em polynomial dependency} on the treewidth, for {\sc Maximum Matching} and {\sc Maximum Flow}.
Furthermore they proved that for graphs with treewidth at most $k$, a tree decomposition of width ${\cal O}(k^2)$ can be computed in ${\cal O}(k^7 \cdot n\log n)$-time.
We observe that their algorithm for {\sc Maximum Matching} is {\em randomized}, whereas ours are deterministic.

We are not aware of the study of another parameter than treewidth for polynomial graph problems.
However, some authors choose a different approach where they study the parameterization of a fixed graph problem for a broad range of graph invariants~\cite{BFNN17,FKMN+17,MNN16}.
As examples, {\em clique-width} is part of the graph invariants used in the parameterized study of {\sc Triangle Listing}~\cite{BFNN17}.
Nonetheless, clique-width is not the main focus in~\cite{BFNN17}.
Recently, Mertzios, Nichterlein and Niedermeier~\cite{MNN16} propose algorithms for {\sc Maximum Matching} that run in time ${\cal O} (k^{{ \cal O} (1)} \cdot (n+m))$, for several parameters such as feedback vertex set or feedback edge set. 
Moreover, the authors in~\cite{MNN16} suggest that {\sc Maximum Matching} may become the ``drosophila'' of the study of the FPT algorithms in P. 
We advance in this research direction.

\subsection{Our results}

In this paper we study the parameterized complexity of several classical graph problems under a wide range of parameters such as clique-width and its upper-bounds {\em modular-width}~\cite{CoB00}, {\em split-width}~\cite{Rao08b}, {\em neighbourhood diversity}~\cite{Lam12} and {\em $P_4$-sparseness}~\cite{BaO99}. 
The results are summarized in Table~\ref{tab:summary}.

\smallskip
Roughly, it turns out that some hardness assumptions for general graphs do not hold anymore for graph classes of bounded clique-width.
This is the case in particular for {\sc Triangle Detection} and other cycle problems that are subcubic equivalent to it such as, {\it e.g.}, {\sc Girth}, that all can be solved in linear-time, with quadratic dependency on the clique-width, with the help of dynamic programming (Theorems~\ref{thm:cw-triangle} and~\ref{thm:cw-girth}).
The latter complements the results obtained for {\sc Triangle Listing} in~\cite{BFNN17}. 
However many hardness results for {\em distance problems} when using treewidth are proved to also hold when using clique-width (Theorems~\ref{thm:cw-diam},~\ref{thm:cw-bc} and~\ref{thm:cw-hyp}).
These negative results have motivated us to consider some upper-bounds for clique-width as parameters, for which better results can be obtained than for clique-width.
Another motivation stems from the fact that the existence of a parameterized algorithm for computing the clique-width of a graph remains a challenging open problem~\cite{CHLB+00}.
We consider some upper-bounds for clique-width that are defined via {\em linear-time} computable graph decompositions.
Thus if these parameters are small enough, say, in ${\cal O}(n^{1-\varepsilon})$ for some $\varepsilon > 0$, we get truly subcubic or even truly subquadratic algorithms for a wide range of problems.

\begin{table}[htb]
\centering
\begin{tabular}{| p{5cm} | p{4cm} p{4cm} |}
\hline
Problem & \multicolumn{2}{l|}{Parameterized time complexity}\\
\hline
\textsc{Diameter}, \textsc{Eccentricities}
   & ${\cal O}(mw(G)^2 \cdot n + m)$ & ${\cal O}(sw(G)^2 \cdot n +m)$ \\
   & ${\cal O}(mw(G)^3 + n+m)$         & ${\cal O}(q(G)^3 + n+m)$\\
\hline
\textsc{Betweenness Centrality}
   & ${\cal O}(mw(G)^2 \cdot n + m)$ & ${\cal O}(sw(G)^2 \cdot n +m)$ \\
   & ${\cal O}(nd(G)^3 + n+m)$ & \\
\hline
\textsc{Hyperbolicity}
   & ${\cal O}(mw(G)^3 \cdot n + m)$ & ${\cal O}(nd(G)^4 + n+m)$\\
   & ${\cal O}(sw(G)^3 \cdot n+m )$ & ${\cal O}(q(G)^3 \cdot n+m)$\\
\hline
\textsc{Maximum Matching}
   & ${\cal O}(mw(G)^4 \cdot n + m)$ & ${\cal O}(q(G)^4 \cdot n+m)$ \\
\hline
\textsc{Triangle Detection}, \textsc{Triangle Counting}, \textsc{Girth} 
   & \multicolumn{2}{l|}{${\cal O}(k^2 \cdot (n + m))$ for any $k\in\{cw(G), mw(G), sw(G), q(G)\}$} \\
   \hline
\end{tabular}
\caption{Summary of positive results.}
\label{tab:summary}
\end{table}

\subsubsection*{Graph parameters and decompositions considered}

Let us describe the parameters considered in this work as follows.
The following is only an informal high level description (formal definitions are postponed to Section~\ref{sec:prelim}).

\paragraph*{\sc Split Decomposition.}
A {\em join} is a set of edges inducing a complete bipartite subgraph.
Roughly, clique-width can be seen as a measure of how easy it is to reconstruct a graph by adding joins between some vertex-subsets.
A {\em split} is a join that is also an edge-cut.
By using pairwise non crossing splits, termed ``strong splits'', we can decompose any graph into degenerate and prime subgraphs, that can be organized in a treelike manner.
The latter is termed {\em split decomposition}~\cite{GiP12}.

We take advantage of the treelike structure of split decomposition in order to design dynamic programming algorithms for distance problems such as {\sc Diameter}, {\sc Gromov Hyperbolicity} and {\sc Betweenness Centrality} (Theorems~\ref{thm:sw-ecc},~\ref{thm:sw-hyp} and~\ref{thm:sw-bc}, respectively).  
Although clique-width is also related to some treelike representations of graphs~\cite{CHMPR15}, the same cannot be done for clique-width as for split decomposition because the edges in the treelike representations for clique-width may not represent a join.

\paragraph*{\sc Modular Decomposition.}
Then, we can improve the results obtained with split decomposition by further restricting the type of splits considered. 
As an example, let $(A,B)$ be a bipartition of the vertex-set that is obtained by removing a split.
If every vertex of $A$ is incident to some edges of the split then $A$ is called a {\em module} of $G$.
That is, for every vertex $v \in B$, $v$ is either adjacent or nonadjacent to every vertex of $A$.
The well-known {\em modular decomposition} of a graph is a hierarchical decomposition that partitions the vertices of the graph with respect to the modules~\cite{HaP10}.
Split decomposition is often presented as a refinement of modular decomposition~\cite{GiP12}.
We formalize the relationship between the two in Lemma~\ref{lem:mw-to-sw}, that allows us to also apply our methods for split decomposition to modular decomposition.

However, we can often do better with modular decomposition than with split decomposition.
In particular, suppose we partition the vertex-set of a graph $G$ into modules, and then we keep exactly one vertex per module.
The resulting {\em quotient graph} $G'$ keeps most of the distance properties of $G$.
Therefore, in order to solve a distance problem for $G$, it is often the case that we only need to solve it for $G'$.
We so believe that modular decomposition can be a powerful {\em Kernelization} tool in order to solve graph problems in P.
As an application, we improve the running time for some of our algorithms, from time ${\cal O}(k^{{\cal O}(1)} \cdot n + m)$ when parameterized by the {\em split-width} (maximum order of a prime subgraph in the split decomposition), to ${\cal O}(k^{{\cal O}(1)} + n + m)$-time when parameterized by the {\em modular-width} (maximum order of a prime subgraph in the modular decomposition). 
See Theorem~\ref{thm:mw-ecc}.

Furthermore, for some more graph problems, it may also be useful to further restrict the internal structures of modules.
We briefly explore this possibility through a case study for {\em neighbourhood diversity}.
Roughly, in this latter case we only consider modules that are either independent sets (false twins) or cliques (true twins).
New kernelization results are obtained for {\sc Hyperbolicity} and {\sc Betweenness Centrality} when parameterized by the neighbourhood diversity (Theorems~\ref{thm:nd-hyp} and~\ref{thm:nd-bc}, respectively).
It is worth pointing out that so far, we have been unable to obtain kernelization results for {\sc Hyperbolicity} and {\sc Betweenness Centrality} when only parameterized by the modular-width.
It would be very interesting to prove separability results between split-width, modular-width and neighbourhood diversity in the field of fully polynomial parameterized complexity.

\paragraph*{\sc Graphs with few $P_4$'s.}
We finally use modular decomposition as our main tool for the design of new linear-time algorithms when restricted to graphs with few induced $P_4$'s.
The $(q,t)$-graphs have been introduced by Babel and Olariu in~\cite{BaO98}.
They are the graphs in which no set of at most $q$ vertices can induce more than $t$ paths of length four.
Every graph is a $(q,t)$-graph for some large enough values of $q$ and $t$.
Furthermore when $q$ and $t$ are fixed constants, $t \leq q-3$, the class of $(q,t)$-graphs has bounded clique-width~\cite{MaR99}.
We so define the $P_4$-sparseness of a given graph $G$, denoted by $q(G)$, as the minimum $q \geq 7$ such that $G$ is a $(q,q-3)$-graph. 
The structure of the quotient graph of a $(q,q-3)$-graph, $q$ being a constant, has been extensively studied and characterized in the literature~\cite{Bab98,BaO98,BaO99,Bab00,JaO95}.
We take advantage of these existing characterizations in order to generalize our algorithms with modular decomposition to ${\cal O}(q(G)^{{\cal O}(1)} \cdot n + m)$-time algorithms (Theorems~\ref{thm:qq3-ecc} and~\ref{thm:qq4-hyp}).

Let us give some intuition on how the $P_4$-sparseness can help in the design of improved algorithms for hard graph problems in P. 
We consider the class of {\em split graphs} ({\it i.e.}, graphs that can be bipartitioned into a clique and an independent set).
Deciding whether a given split graph has diameter $2$ or $3$ is hard~\cite{BCH16}.
However, suppose now that the split graph is a $(q,q-3)$-graph $G$, for some fixed $q$.
An induced $P_4$ in $G$ has its two ends $u,v$ in the independent set, and its two middle vertices are, respectively, in $N_G(u) \setminus N_G(v)$ and $N_G(v) \setminus N_G(u)$.
Furthermore, when $G$ is a $(q,q-3)$-graph, it follows from the characterization of~\cite{Bab98,BaO98,BaO99,Bab00,JaO95} either it has a quotient graph of bounded order ${\cal O}(q)$ or it is part of a well-structured subclass where the vertices of all neighbourhoods in the independent set follow a rather nice pattern (namely, spiders and a subclass of $p$-trees, see Section~\ref{sec:prelim}).
As a result, the diameter of $G$ can be computed in ${\cal O}(\max\{q^3,n+m\})$-time when $G$ is a $(q,q-3)$ split graph.
We generalize this result to every $(q,q-3)$-graph by using modular decomposition.

\medskip
All the parameters considered in this work have already received some attention in the literature, especially in the design of FPT algorithms for NP-hard problems~\cite{Bab00,GaP03,GiP12,GLO13,Rao08b}.
However, we think we are the first to study clique-width and its upper-bounds for polynomial problems.
There do exist linear-time algorithms for {\sc Diameter}, {\sc Maximum Matching} and some other problems we study when restricted to some graph classes where the split-width or the $P_4$-sparseness is bounded ({\it e.g.}, cographs~\cite{YuY93}, distance-hereditary graphs~\cite{Dra97,DrN00}, $P_4$-tidy graphs~\cite{FPT97}, etc.).
Nevertheless, we find the techniques used for these specific subclasses hardly generalize to the case where the graph has split-width or $P_4$-sparseness at most $k$, $k$ being any fixed constant.
For instance, the algorithm that is proposed in~\cite{DrN00} for computing the diameter of a given distance-hereditary graph is based on some properties of LexBFS orderings.
Distance-hereditary graphs are exactly the graphs with split-width at most two~\cite{GiP12}.
However it does not look that simple to extend the properties found for their LexBFS orderings to bounded split-width graphs in general.
As a byproduct of our approach, we also obtain new linear-time algorithms when restricted to well-known graph families such as cographs and distance-hereditary graphs.

\subsubsection*{Highlight of our {\sc Maximum Matching} algorithms}

Finally we emphasize our algorithms for {\sc Maximum Matching}. 
Here we follow the suggestion of Mertzios, Nichterlein and Niedermeier~\cite{MNN16} that {\sc Maximum Matching} may become the ``drosophila'' of the study of the FPT algorithms in P.
Precisely, we propose ${\cal O} (k^4 \cdot n + m)$-time algorithms for {\sc Maximum Matching} when parameterized either by modular-width or by the $P_4$-sparseness of the graph (Theorems~\ref{thm:mw-maxmatching} and~\ref{thm:qq3-maxmatching}).
The latter subsumes many algorithms that have been obtained for specific subclasses~\cite{FPT97,YuY93}.

\smallskip
Let us sketch the main lines of our approach.
Our algorithms for {\sc Maximum Matching} are recursive.
Given a partition of the vertex-set into modules, first we compute a maximum matching for the subgraph induced by every module separately.
Taking the union of all the outputted matchings gives a matching for the whole graph, but this matching is not necessarily maximum. 
So, we aim at increasing its cardinality by using augmenting paths~\cite{Ber57}.

In an unpublished paper~\cite{Nov89}, Novick followed a similar approach and, based on an integer programming formulation, he obtained an ${\cal O}(k^{{\cal O}(k^3)}n + m)$-time algorithm for {\sc Maximum Matching} when parameterized by the modular-width.
Our approach is more combinatorial than his.

Our contribution in this part is twofold.
First we carefully study the possible ways an augmenting path can cross a module.
Our analysis reveals that in order to compute a maximum matching in a graph of modular-width at most $k$ we only need to consider augmenting paths of length ${\cal O}(k)$.
Then, our second contribution is an efficient way to compute such paths.
For that, we design a new type of characteristic graph of size ${\cal O}(k^4)$. 
The same as the classical quotient graph keeps most distance properties of the original graph, our new type of characteristic graph is tailored to enclose the main properties of the current matching in the graph.
We believe that the design of new types of characteristic graphs can be a crucial tool in the design of improved algorithms for graph classes of bounded modular-width. 

\medskip
We have been able to extend our approach with modular decomposition to an ${\cal O}(q^4 \cdot n + m)$-time algorithm for computing a maximum matching in a given $(q,q-3)$-graph.
However, a characterization of the quotient graph is not enough to do that.
Indeed, we need to go deeper in the $p$-connectedness theory of~\cite{BaO99} in order to better characterize the nontrivial modules in the graphs (Theorem~\ref{thm:stronger-mdc-qq3}).
Furthermore our algorithm for $(q,q-3)$-graph not only makes use of the algorithm with modular decomposition.
On our way to solve this case we have generalized different methods and reduction rules from the literature~\cite{KaS81,YuY93}, that is of independent interest.

We suspect that our algorithm with modular decomposition can be used as a subroutine in order to solve {\sc Maximum Matching} in linear-time for bounded split-width graphs.
However, this is left for future work.

\subsection{Organization of the paper}

In Section~\ref{sec:prelim} we introduce definitions and basic notations.

\medskip
Then, in Section~\ref{sec:cycle} we show FPT algorithms when parameterized by the clique-width.
The problems considered are \textsc{Triangle Counting} and \textsc{Girth}. 
To the best of our knowledge, we present the first known polynomial parameterized algorithm for {\sc Girth} (Theorem~\ref{thm:cw-girth}). 
Roughly, the main idea behind our algorithms is that given a labeled graph $G$ obtained from a $k$-expression, we can compute a minimum-length cycle for $G$ by keeping up to date the pairwise distances between every two label classes.
Hence, if a $k$-expression of length $L$ is given as part of the input we obtain algorithms running in time ${\cal O}(k^2L)$ and space ${\cal O}(k^2)$.

\medskip 
In Section~\ref{sec:dist} we consider distance related problems, namely: {\sc Diameter}, {\sc Eccentricities}, {\sc Hyperbolicity} and {\sc Betweenness Centrality}. 

\smallskip
We start proving, in Section~\ref{sec:hardness-cw}, none of these problems above can be solved in time $2^{o(k)} n^{2 - \varepsilon}$, for any $\varepsilon > 0$, when parameterized by the clique-width (Theorems~\ref{thm:cw-diam}---\ref{thm:cw-hyp}).
These are the first known hardness results for clique-width in the field of ``Hardness in P''.
Furthermore, as it is often the case in this field, our results are conditioned on the Strong Exponential Time Hypothesis~\cite{IPZ98}.
In summary, we take advantage of recent hardness results obtained for {\em bounded-degree} graphs~\cite{EvD16}.
Clique-width and treewidth can only differ by a constant-factor in the class of bounded-degree graphs~\cite{Cou12,GuW00}.
Therefore, by combining the hardness constructions for bounded-treewidth graphs and for bounded-degree graphs, we manage to derive hardness results for graph classes of bounded clique-width.

\smallskip
In Section~\ref{sec:algos-split} we describe fully polynomial FPT algorithms for {\sc Diameter}, {\sc Eccentricity}, {\sc Hyperbolicity} and {\sc Betweenness centrality} parameterized by the split-width. 
Our algorithms use split-decomposition as an efficient preprocessing method.  
Roughly, we define weighted versions for every problem considered (some of them admittedly technical).
In every case, we prove that solving the original distance problem can be reduced in linear-time to the solving of its weighted version for every subgraph of the split decomposition separately.

\smallskip
Then, in Section~\ref{sec:algos-modular-dec} we apply the results from Section~\ref{sec:algos-split} to modular-width. 
First, since $sw(G) \leq mw(G)+1$ for any graph $G$, all our algorithms parameterized by split-width are also algorithms parameterized by modular-width.
Moreover for {\sc Eccentricities}, and for {\sc Hyperbolicity} and {\sc Betweenness Centrality} when parameterized by the neighbourhood diversity, we show that it is sufficient only to process the quotient graph of $G$.
We thus obtain algorithms that run in ${\cal O}(mw(G)^{{\cal O}(1)} + n + m)$-time, or ${\cal O}(nd(G)^{{\cal O}(1)} + n + m)$-time, for all these problems.

\smallskip
In Section~\ref{sec:dist-qq3} we generalize our previous algorithms to be applied to the $(q,q-3)$-graphs. 
We obtain our results by carefully analyzing the cases where the quotient graph has size $\Omega(q)$. 
These cases are given by Lemma~\ref{lem:reduce-qq3}.
 
\medskip 
Section~\ref{sec:maxmatching} is dedicated to our main result, linear-time algorithms for {\sc Maximum Matching}.
First in Section~\ref{sec:matching-mw} we propose an algorithm parameterized by the modular-width that runs in ${\cal O}(mw(G)^{4} \cdot n + m)$-time.
In Section~\ref{sec:matching-qq3} we generalize this algorithm to $(q,q-3)$-graphs. 

\medskip
Finally, in Section~\ref{sec:applications} we discuss applications to other graph classes.

\section{Preliminaries}\label{sec:prelim}

We use standard graph terminology from~\cite{BoM08,Die10}.
Graphs in this study are finite, simple (hence without loops or multiple edges) and unweighted -- unless stated otherwise.
Furthermore we make the standard assumption that graphs are encoded as adjacency lists.

We want to prove the existence, or the nonexistence, of graph algorithms with running time of the form $k^{{\cal O}(1)} \cdot (n+m)$, $k$ being some fixed graph parameter.
In what follows, we introduce the graph parameters considered in this work.

\subsubsection*{Clique-width}

A labeled graph is given by a pair $\langle G, \ell \rangle$ where $G=(V,E)$ is a graph and $\ell : V \to \mathbb{N}$ is called a labeling function.
A {\em k-expression} can be seen as a sequence of operations for constructing a labeled graph $\langle G, \ell \rangle$, where the allowed four operations are:
\begin{enumerate}
\item Addition of a new vertex $v$ with label $i$ (the labels are taken in $\{1, 2, \ldots, k\}$), denoted $i(v)$;
\item Disjoint union of two labeled graphs $\langle G_1, \ell_1 \rangle$ and $\langle G_2, \ell_2 \rangle$, denoted $\langle G_1, \ell_1 \rangle \oplus \langle G_2, \ell_2 \rangle$;
\item Addition of a join between the set of vertices labeled $i$ and the set of vertices labeled $j$, where $i \neq j$, denoted $\eta(i,j)$;
\item Renaming label $i$ to label $j$, denoted $\rho(i,j)$.
\end{enumerate}
See Fig.~\ref{fig:cw-example} for examples.
The {\em clique-width} of $G$, denoted by $cw(G)$, is the minimum $k$ such that, for some labeling $\ell$, the labeled graph $\langle G, \ell \rangle$ admits a $k$-expression~\cite{CER93}.
We refer to~\cite{CMR00} and the references cited therein for a survey of the many applications of clique-width in the field of parameterized complexity.

\begin{figure}[h!]
\centering
\begin{subfigure}[b]{.15\textwidth}\centering
\includegraphics[width=.25\textwidth]{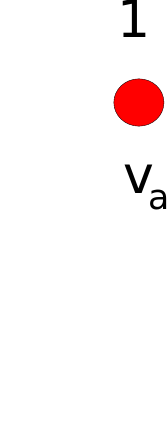}
\caption{$1(v_a)$}
\end{subfigure}\hfill
\begin{subfigure}[b]{.15\textwidth}\centering
\includegraphics[width=.5\textwidth]{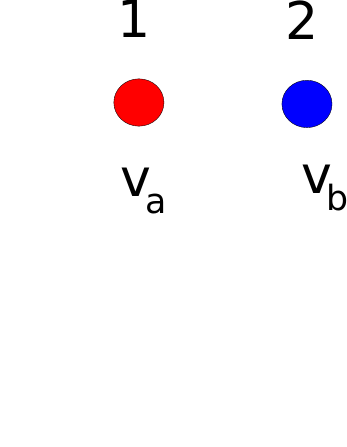}
\caption{$2(v_b)$}
\end{subfigure}\hfill
\begin{subfigure}[b]{.15\textwidth}\centering
\includegraphics[width=.5\textwidth]{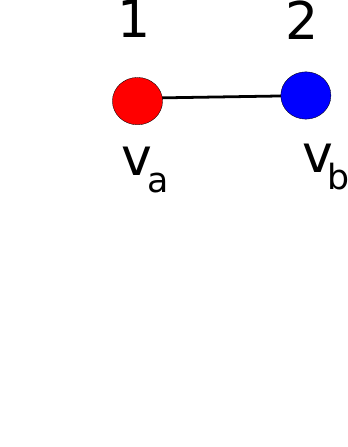}
\caption{$\eta(1,2)$}
\end{subfigure}\hfill
\begin{subfigure}[b]{.15\textwidth}\centering
\includegraphics[width=.5\textwidth]{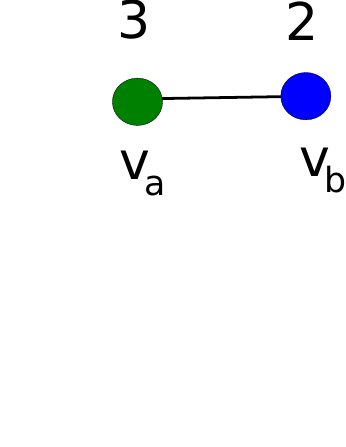}
\caption{$\rho(1,3)$}
\end{subfigure}\hfill
\begin{subfigure}[b]{.15\textwidth}\centering
\includegraphics[width=.75\textwidth]{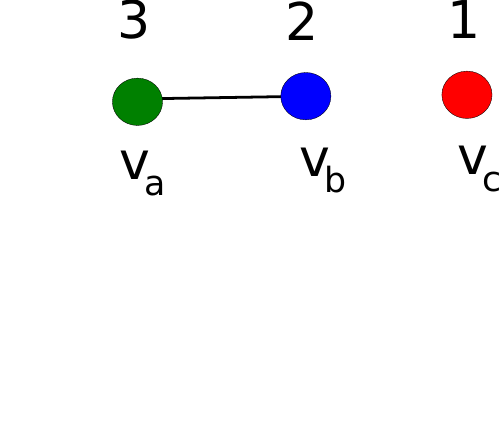}
\caption{$1(v_c)$}
\end{subfigure}\hfill
\begin{subfigure}[b]{.15\textwidth}\centering
\includegraphics[width=.75\textwidth]{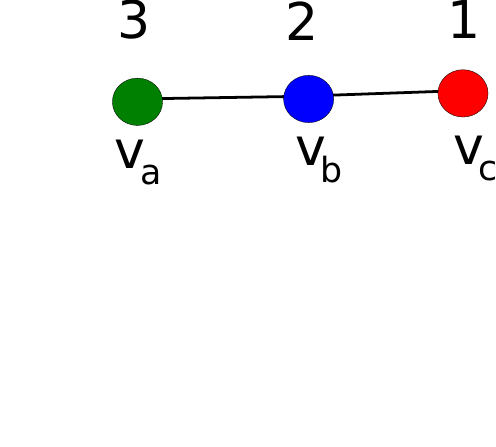}
\caption{$\eta(2,1)$}
\end{subfigure}\vspace{15pt}
\begin{subfigure}[b]{.15\textwidth}\centering
\includegraphics[width=.75\textwidth]{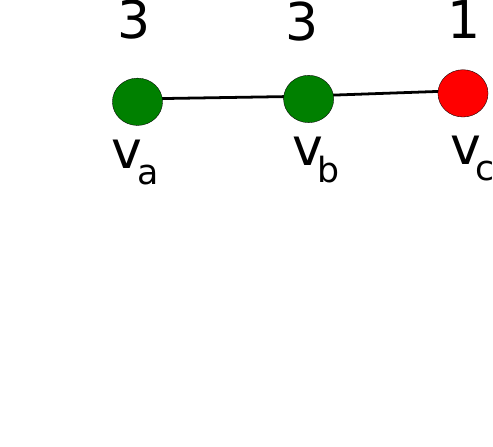}
\caption{$\rho(2,3)$}
\end{subfigure}\hfill
\begin{subfigure}[b]{.15\textwidth}\centering
\includegraphics[width=\textwidth]{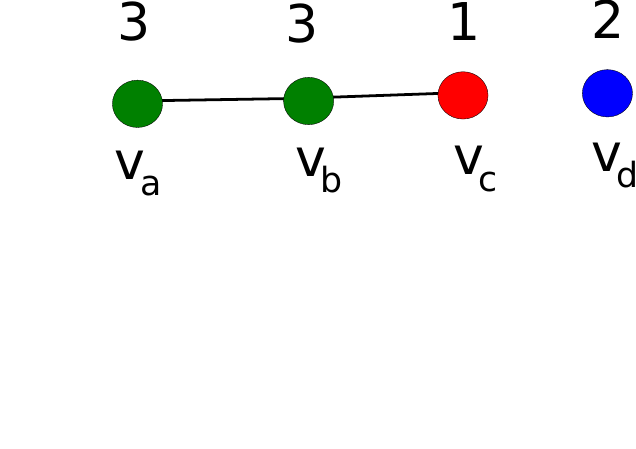}
\caption{$2(v_d)$}
\end{subfigure}\hfill
\begin{subfigure}[b]{.15\textwidth}\centering
\includegraphics[width=\textwidth]{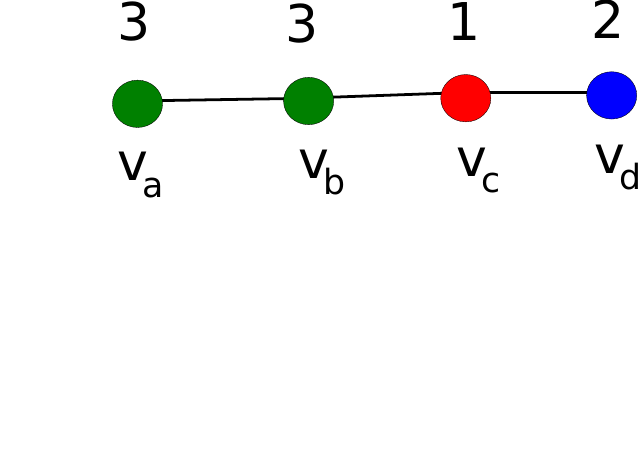}
\caption{$\eta(1,2)$}
\end{subfigure}
\caption{A $3$-expression for the path $P_4$.}
\label{fig:cw-example}
\end{figure}

Computing the clique-width of a given graph is NP-hard~\cite{FRRS09}.
However, on a more positive side the graphs with clique-width two are exactly the cographs and they can be recognized in linear-time~\cite{CPS85,CoB00}.
Clique-width three graphs can also be recognized in polynomial-time~\cite{CHLB+00}.
The parameterized complexity of computing the clique-width is open.
In what follows, we focus on upper-bounds on clique-width that are derived from some graph decompositions.  

\subsubsection*{Modular-width}

A {\em module} in a graph $G=(V,E)$ is any subset $M \subseteq V(G)$ such that for any $v \in V \setminus M$, either $M \subseteq N_G(v)$ or $M \cap N_G(v) = \emptyset$.
Note that $\emptyset, \ V, \ \mbox{and} \ \{v\}$ for every $v \in V$ are trivial modules of $G$.
A graph is called {\em prime} for modular decomposition if it only has trivial modules.

A module $M$ is {\em strong} if it does not overlap any other module, {\it i.e.}, for any module $M'$ of $G$, either one of $M$ or $M'$ is contained in the other or $M$ and $M'$ do not intersect. 
Furthermore, let ${\cal M}(G)$ be the family of all inclusion wise maximal strong modules of $G$ that are proper subsets of $V$.
The {\em quotient graph} of $G$ is the graph $G'$ with vertex-set ${\cal M}(G)$ and an edge between every two $M,M' \in {\cal M}(G)$ such that every vertex of $M$ is adjacent to every vertex of $M'$.

Modular decomposition is based on the following structure theorem from Gallai.

\begin{theorem}[~\cite{Gal67}]\label{thm:modular-dec}
For an arbitrary graph $G$ exactly one of the following conditions is satisfied.
\begin{enumerate}
\item $G$ is disconnected;
\item its complement $\overline{G}$ is disconnected;
\item or its quotient graph $G'$ is prime for modular decomposition.
\end{enumerate}
\end{theorem}

Theorem~\ref{thm:modular-dec} suggests the following recursive procedure in order to decompose a graph, that is sometimes called modular decomposition.
If $G = G'$ ({\it i.e.}, $G$ is complete, edgeless or prime for modular decomposition) then we output $G$.
Otherwise, we output the quotient graph $G'$ of $G$ and, for every strong module $M$ of $G$, the modular decomposition of $G[M]$.
The modular decomposition of a given graph $G=(V,E)$ can be computed in linear-time~\cite{TCHP08}.
See Fig.~\ref{fig:modular-dec} for an example.

\begin{figure}[h!]
\centering
\begin{subfigure}[b]{.46\textwidth}\centering
\includegraphics[width=.35\textwidth]{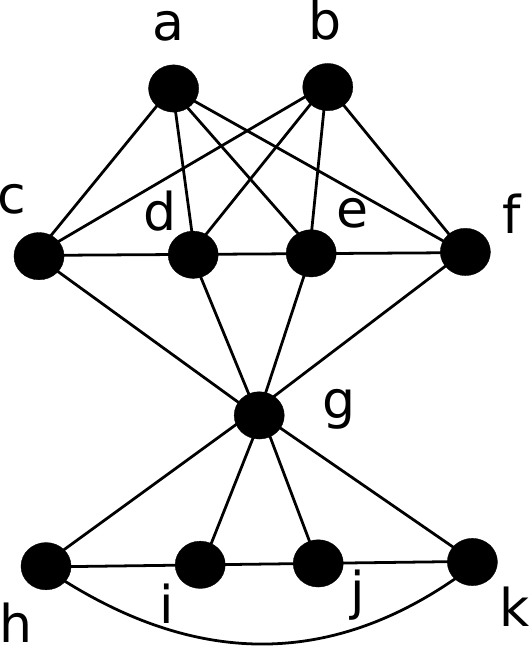}
\end{subfigure}\hfill
\begin{subfigure}[b]{.52\textwidth}\centering
\includegraphics[width=.65\textwidth]{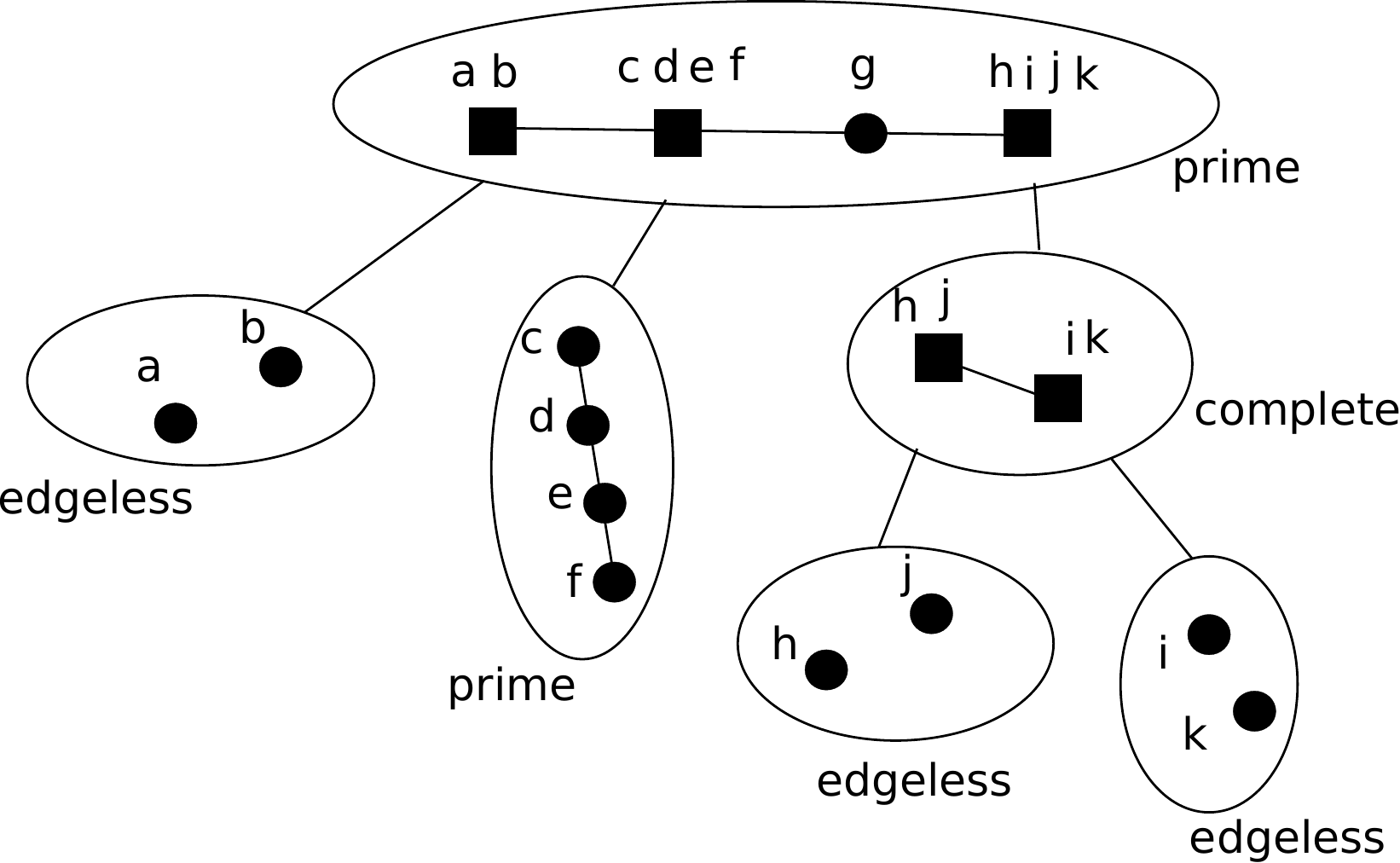}
\end{subfigure}
\caption{A graph and its modular decomposition.}
\label{fig:modular-dec}
\end{figure}

Furthermore, by Theorem~\ref{thm:modular-dec} the subgraphs from the modular decomposition are either edgeless, complete, or prime for modular decomposition. 
The {\em modular-width} of $G$, denoted by $mw(G)$, is the minimum $k \geq 2$ such that any prime subgraph in the modular decomposition has order (number of vertices) at most $k$~\footnote{This term has another meaning in~\cite{Rao08}. We rather follow the terminology from~\cite{CoB00}.}.
The relationship between clique-width and modular-width is as follows.

\begin{lemma}[~\cite{CMR00}]\label{lem:mw-cw}
For every $G=(V,E)$, we have $cw(G) \leq mw(G)$, and a $mw(G)$-expression defining $G$ can be constructed in linear-time.
\end{lemma}

We refer to~\cite{HaP10} for a survey on modular decomposition.
In particular, graphs with modular-width two are exactly the cographs, that follows from the existence of a cotree~\cite{Sum73}.
Cographs enjoy many algorithmic properties, including a linear-time algorithm for {\sc Maximum Matching}~\cite{YuY93}.
Furthermore, in~\cite{GLO13} Gajarsk{\`y}, Lampis and Ordyniak prove that for some $W$-hard problems when parameterized by clique-width there exist FPT algorithms when parameterized by modular-width.

\subsubsection*{Split-width}

A {\em split} $(A,B)$ in a {\em connected} graph $G=(V,E)$ is a partition $V = A \cup B$ such that: $\min \{ |A|, |B| \} \geq 2$; and there is a complete join between the vertices of $N_G(A)$ and $N_G(B)$.
For every split $(A,B)$ of $G$, let $a \in N_G(B), \ b \in N_G(A)$ be arbitrary.
The vertices $a,b$ are termed {\em split marker vertices}.
We can compute a ``simple decomposition'' of $G$ into the subgraphs $G_A = G[A \cup \{b\}]$ and $G_B = G[B \cup \{a\}]$.

There are two cases of ``indecomposable'' graphs.
Degenerate graphs are such that every bipartition of their vertex-set is a split.
They are exactly the complete graphs and the stars~\cite{Cun82}.
A graph is prime for split decomposition if it has no split.

A split decomposition of a connected graph $G$ is obtained by applying recursively a simple decomposition, until all the subgraphs obtained are either degenerate or prime.
A split decomposition of an {\em arbitrary} graph $G$ is the union of a split decomposition for each of its connected components.
Every graph has a canonical split decomposition, with minimum number of subgraphs, that can be computed in linear-time~\cite{CDR12}.
The {\em split-width}
of $G$, denoted by $sw(G)$, is the minimum $k \geq 2$ such that any prime subgraph in the canonical split decomposition of $G$ has order at most $k$.
See Fig.~\ref{fig:split-dec} for an illustration.

\begin{lemma}[~\cite{Rao08b}]\label{lem:sw-cw}
For every $G=(V,E)$, we have $cw(G) \leq 2 \cdot sw(G) + 1$, and a $(2 \cdot sw(G) + 1)$-expression defining $G$ can be constructed in linear-time.
\end{lemma}

We refer to~\cite{GaP03,GiP12,Rao08b} for some algorithmic applications of split decomposition.
In particular, graphs with split-width at most two are exactly the distance-hereditary graphs~\cite{BaM86}.
Linear-time algorithms for solving {\sc Diameter} and {\sc Maximum Matching} for distance-hereditary graphs are presented in~\cite{DrN00,Dra97}.

\begin{figure}[h!]
\centering
\begin{subfigure}[b]{.46\textwidth}\centering
\includegraphics[width=.45\textwidth]{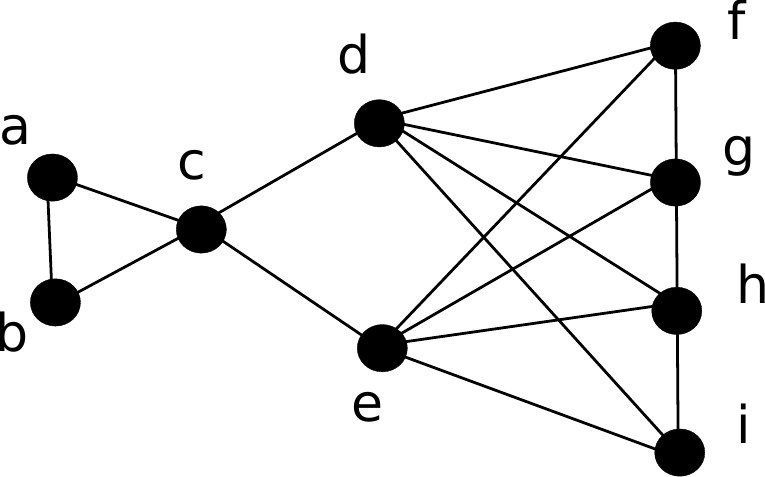}
\end{subfigure}\hfill
\begin{subfigure}[b]{.52\textwidth}\centering
\includegraphics[width=.75\textwidth]{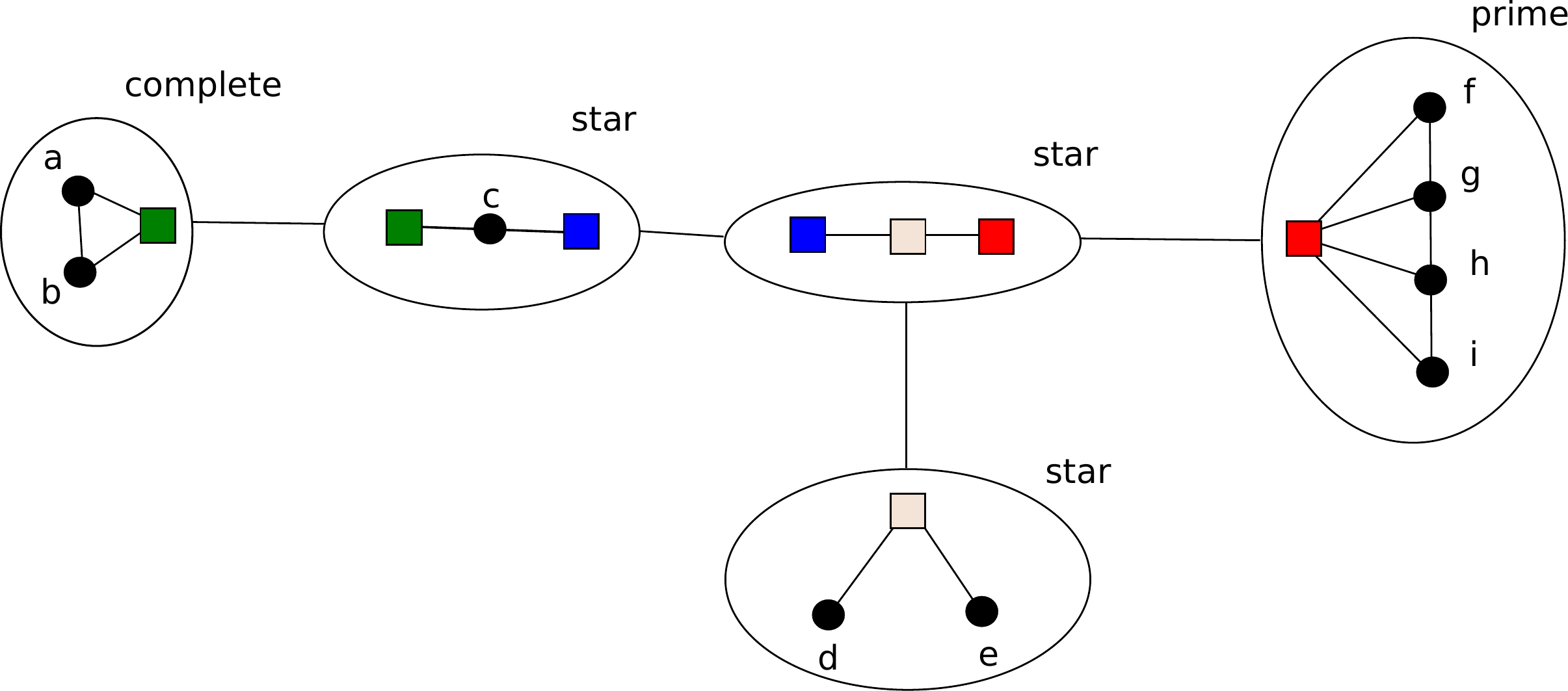}
\end{subfigure}
\caption{A graph and its split decomposition.}
\label{fig:split-dec}
\end{figure}

We stress that split decomposition can be seen as a refinement of modular decomposition.
Indeed, if $M$ is a module of $G$ and $\min\{|M|,|V\setminus M|\} \geq 2$ then $(M, V \setminus M)$ is a split.
In what follows, we prove most of our results with the more general split decomposition.

\subsubsection*{Graphs with few $P_4$'s}

A $(q,t)$-graph $G = (V,E)$ is such that for any $S \subseteq V$, $|S| \leq q$,  $S$ induces at most $t$ paths on four vertices~\cite{BaO98}.
The {\em $P_4$-sparseness} of $G$, denoted by $q(G)$, is the minimum $q \geq 7$ such that $G$ is a $(q,q-3)$-graph.

\begin{lemma}[~\cite{MaR99}]\label{lem:qq3-cw}
For every $q \geq 7$, every $(q,q-3)$-graph has clique-width at most $q$, and a $q$-expression defining it can be computed in linear-time.
\end{lemma}

The algorithmic properties of several subclasses of $(q,q-3)$-graphs have been considered in the literature.
We refer to~\cite{BaO99} for a survey.
Furthermore, there exists a canonical decomposition of $(q,q-3)$-graphs, sometimes called the {\em primeval decomposition}, that can be computed in linear-time~\cite{Bau96}.
Primeval decomposition can be seen as an intermediate between modular and split decomposition.
We postpone the presentation of primeval decomposition until Section~\ref{sec:maxmatching}.
Until then, we state the results in terms of modular decomposition.

\smallskip
More precisely, given a $(q,q-3)$-graphs $G$, the prime subgraphs in its modular decomposition may be of super-constant size $\Omega(q)$.
However, if they are then they are part of one of the well-structured graph classes that we detail next.

\smallskip
A {\em disc} is either a cycle $C_n$, or a co-cycle $\overline{C_n}$, for some $n \geq 5$.

\begin{figure}[h!]
\centering
\begin{subfigure}[b]{.48\textwidth}\centering
\includegraphics[width=.45\textwidth]{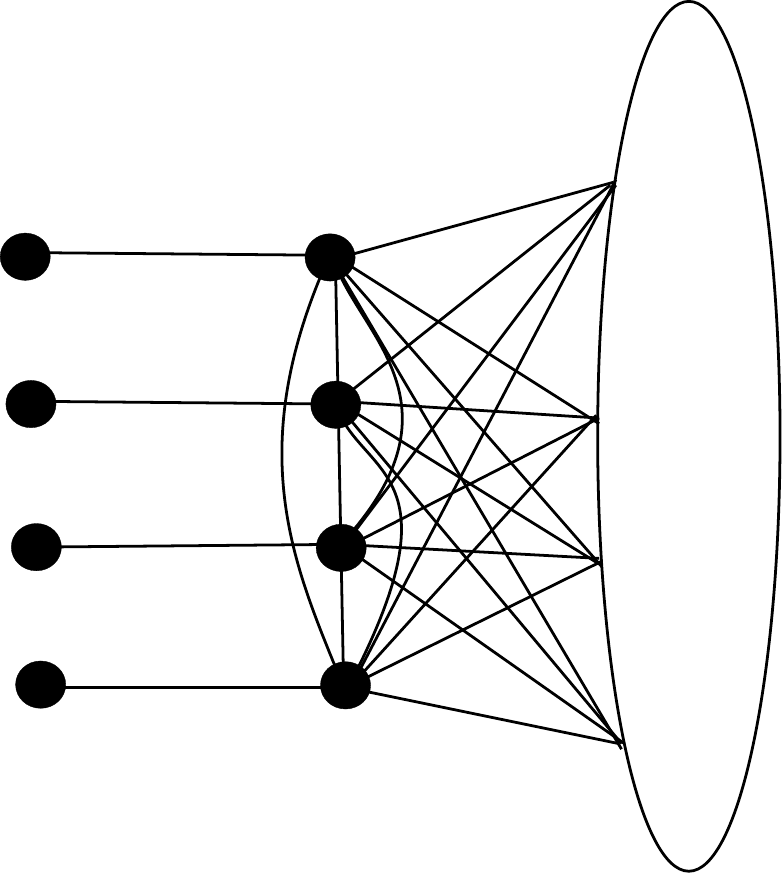}
\caption{Thin spider.}
\end{subfigure}\hfill
\begin{subfigure}[b]{.48\textwidth}\centering
\includegraphics[width=.5\textwidth]{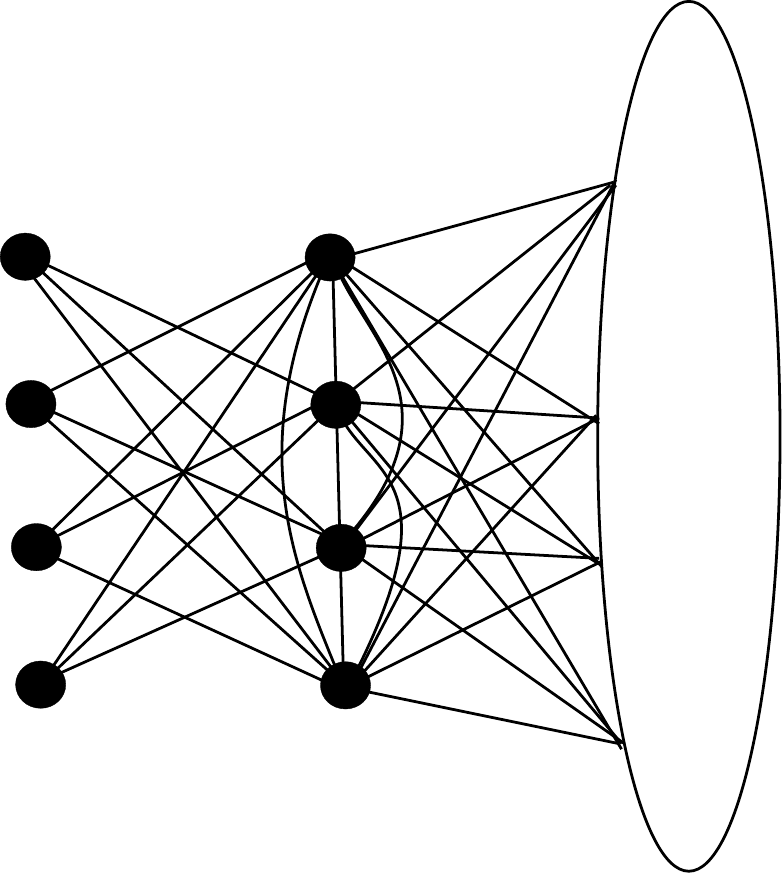}
\caption{Thick spider.}
\end{subfigure}
\caption{Spiders.}
\label{fig:splider}
\end{figure}

\smallskip
A \emph{spider} $G=(S \cup K \cup R,E)$ is a graph with vertex set $V=S \cup K \cup R$ and edge set $E$ such that:
\begin{enumerate}
	\item $(S,K,R)$ is a partition of $V$ and $R$ may be empty;
	\item the subgraph $G[K\cup R]$ induced by $K$ and $R$ is the complete join $K \oplus R$, and $K$ separates $S$ and $R$, i.e. any path from a vertex in $S$ and a vertex in $R$ contains a vertex in $K$;
	\item
	$S$ is a stable set, $K$ is a clique, $|S| = |K|\geq 2$, and
	there exists a bijection $f: S \longrightarrow K$ such that, either for all
	vertices $s\in S$, $N(s)\cap K = K - \{f(s)\}$ or $N(s)\cap K = \{f(s)\}$. Roughly speaking, the edges between $S$ and $K$ are either a matching or an anti-matching. In the former case or if $|S| = |K| \leq 2$, $G$ is called \emph{thin}, otherwise $G$ is \emph{thick}. 
See Fig.~\ref{fig:splider}.
\end{enumerate}
If furthermore $|R| \leq 1$ then we call $G$ a {\em prime spider}.

\begin{figure}[h!]
\centering
\includegraphics[width=.45\textwidth]{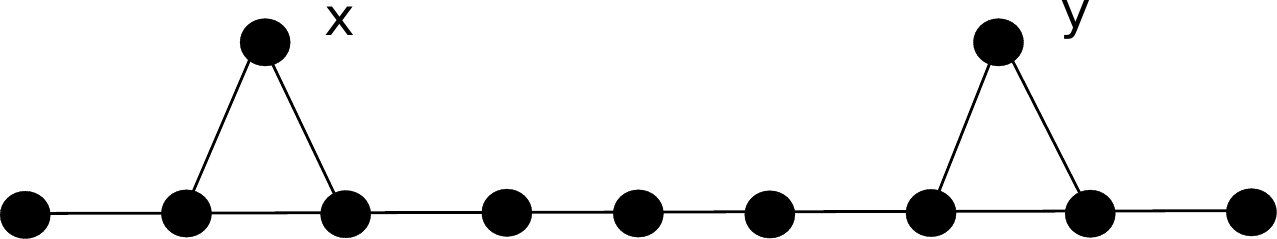}
\caption{Spiked $p$-chain $P_k$.}
\label{fig:spiked-pk}
\end{figure}

\smallskip
Let $P_k = (v_1,v_2,v_3, \ldots, v_k), \ k \geq 6$ be a path of length at least five.
A {\em spiked $p$-chain $P_k$} is a supergraph of $P_k$, possibly with the additional vertices $x,y$ such that: $N(x) = \{v_2,v_3\}$ and $N(y) = \{v_{k-2},v_{k-1}\}$.
See Fig.~\ref{fig:spiked-pk}.
Note that one or both of $x$ and $y$ may be missing.
In particular, $P_k$ is a spiked $p$-chain $P_k$.   
A {\em spiked $p$-chain $\overline{P_k}$} is the complement of a spiked $p$-chain $P_k$.

Let $Q_k$ be the graph with vertex-set $\{v_1,v_2, \ldots, v_k\}, \ k \geq 6$ such that, for every $i \geq 1$, 
$N_{Q_k}(v_{2i-1}) = \{ v_{2j} \mid j \leq i, \ j \neq i-1 \}$ 
and $N_{Q_k}(v_{2i}) = \{v_{2j} \mid j \neq i \} \cup \{ v_{2j-1} \mid j \geq i, \ j \neq i+1 \}$. 
A {\em spiked $p$-chain $Q_k$} is a supergraph of $Q_k$, possibly with the additional vertices $z_2, z_3, \ldots, z_{k-5}$ such that:
\begin{itemize}
\item $N(z_{2i-1}) = \{v_{2j} \mid j \in [1;i]\} \cup \{z_{2j} \mid j \in [1;i-1]\}$;
\item $\overline{N(z_{2i})} = \{v_{2j-1} \mid j \in [1;i+1]\} \cup \{z_{2j-1} \mid j \in [2;i]\}$
\end{itemize} 
Any of the vertices $z_i$ can be missing, so, in particular, $Q_k$ is a spiked $p$-chain $Q_k$.
See Fig.~\ref{fig:spiked-qk}.
A {\em spiked $p$-chain $\overline{Q_k}$} is the complement of a spiked $p$-chain $Q_k$.

\begin{figure}[h!]
\centering
\includegraphics[width=.25\textwidth]{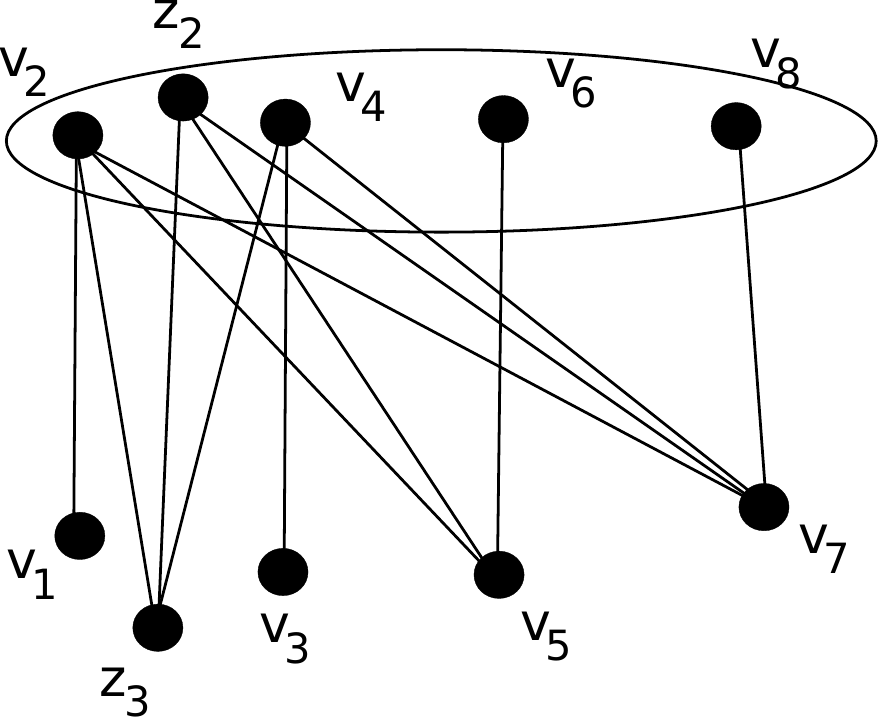}
\caption{Spiked $p$-chain $Q_k$.}
\label{fig:spiked-qk}
\end{figure}

Finally, we say that a graph is a {\em prime $p$-tree} if it is either: a spiked $p$-chain $P_k$, a spiked $p$-chain $\overline{P_k}$, a spiked $p$-chain $Q_k$, a spiked $p$-chain $\overline{Q_k}$, or part of the seven graphs of order at most $7$ that are listed in~\cite{MaR99}.

\begin{lemma}[~\cite{Bab00,MaR99}]\label{lem:reduce-qq3}
Let $G=(V,E)$, $q \geq 7$, be a connected $(q,q-3)$-graph such that $G$ and $\overline{G}$ are connected.
Then, one of the following must hold for its quotient graph $G'$:
\begin{itemize}
\item either $G'$ is a prime spider;
\item or $G'$ is a disc;
\item or $G'$ is a prime $p$-tree;
\item or $|V(G')| \leq q$.
\end{itemize}
\end{lemma}

A simpler version of Lemma~\ref{lem:reduce-qq3} holds for the subclass of $(q,q-4)$-graphs:

\begin{lemma}[~\cite{Bab00}]\label{lem:reduce-qq4}
Let $G=(V,E)$, $q \geq 4$, be a connected $(q,q-4)$-graph such that $G$ and $\overline{G}$ are connected.
Then, one of the following must hold for its quotient graph $G'$:
\begin{itemize}
\item $G'$ is a prime spider;
\item or $|V(G')| \leq q$.
\end{itemize}
\end{lemma}

The subclass of $(q,q-4)$-graphs has received more attention in the literature than $(q,q-3)$-graphs.
Our results hold for the more general case of $(q,q-3)$-graphs.

\section{Cycle problems on bounded clique-width graphs}
\label{sec:cycle}

Clique-width is the smallest parameter that is considered in this work.
We start studying the possibility for $k^{{\cal O}(1)} \cdot (n+m)$-time algorithms on graphs with clique-width at most $k$.
Positive results are obtained for two variations of {\sc Triangle Detection}, namely \textsc{Triangle Counting} and \textsc{Girth}.
We define the problems studied in Section~\ref{sec:pbs-cw}, then we describe the algorithms in order to solve these problems in Section~\ref{sec:algos-cw}.

\subsection{Problems considered}\label{sec:pbs-cw}

We start introducing our basic cycle problem.

\begin{center}
	\fbox{
		\begin{minipage}{.95\linewidth}
			\begin{problem}[\textsc{Triangle Detection}]\
				\label{prob:triangle-detect} 
					\begin{description}
					\item[Input:] A graph $G=(V,E)$.
					\item[Question:] Does there exist a triangle in $G$?
				\end{description}
			\end{problem}     
		\end{minipage}
	}
\end{center}

Note that for general graphs, {\sc Triangle Detection} is conjectured not to be solvable in ${\cal O}(n^{3-\varepsilon})$-time, for any $\varepsilon > 0 $, with a combinatorial algorithm~\cite{VaW10}.
It is also conjectured not to be solvable in ${\cal O}(n^{\omega - \varepsilon})$-time for any $\varepsilon > 0$, with $\omega$ being the exponent for fast matrix multiplication~\cite{AbV14}.
Our results in this section show that such assumptions do not hold when restricted to bounded clique-width graphs.

More precisely, we next describe fully polynomial parameterized algorithms for the two following generalizations of {\sc Triangle Detection}.

\begin{center}
	\fbox{
		\begin{minipage}{.95\linewidth}
			\begin{problem}[\textsc{Triangle Counting}]\
				\label{prob:triangle-count} 
					\begin{description}
					\item[Input:] A graph $G=(V,E)$.
					\item[Output:] The number of triangles in $G$.
				\end{description}
			\end{problem}     
		\end{minipage}
	}
\end{center}

\begin{center}
	\fbox{
		\begin{minipage}{.95\linewidth}
			\begin{problem}[\textsc{Girth}]\
				\label{prob:girth} 
					\begin{description}
					\item[Input:] A graph $G=(V,E)$.
					\item[Output:] The girth of $G$, that is the minimum size of a cycle in $G$.
				\end{description}
			\end{problem}     
		\end{minipage}
	}
\end{center}

In~\cite{VaW10}, the three of {\sc Triangle Detection}, {\sc Triangle Counting} and {\sc Girth} are proved to be subcubic equivalent when restricted to combinatorial algorithms.

\subsection{Algorithms}\label{sec:algos-cw}

Roughly, our algorithms in what follows are based on the following observation.
Given a labeled graph $\langle G, \ell \rangle$ (obtained from a $k$-expression), in order to detect a triangle in $G$, resp. a minimum-length cycle in $G$, we only need to store the adjacencies, resp. the distances, between every two label classes.
Hence, if a $k$-expression of length $L$ is given as part of the input we obtain algorithms running in time ${\cal O}(k^2L)$ and space ${\cal O}(k^2)$.

\smallskip
Our first result is for {\sc Triangle Counting} (Theorem~\ref{thm:cw-triangle}). It shares some similarities with a recent algorithm for listing all triangles in a graph~\cite{BFNN17}.
However, unlike the authors in~\cite{BFNN17}, we needn't use the notion of $k$-modules in our algorithms. 
Furthermore, since we only ask for {\em counting} triangles, and not to list them, we obtain a better time complexity than in~\cite{BFNN17}.

\begin{theorem}\label{thm:cw-triangle}
For every $G=(V,E)$, \textsc{Triangle Counting} can be solved in ${\cal O}(k^2 \cdot (n+m))$-time if a $k$-expression of $G$ is given.
\end{theorem}

\begin{proof}
We need to assume the $k$-expression is {\em irredundant}, that is, when we add a complete join between the vertices labeled $i$ and the verticed labeled $j$, there was no edge before between these two subsets.
Given a $k$-expression of $G$, an irredundant $k$-expression can be computed in linear-time~\cite{CoB00}.
Then, we proceed by dynamic programming on the irredundant $k$-expression.

More precisely, let $\langle G, \ell \rangle$ be a labeled graph, $\ell : V(G) \to \{1, \ldots, k\}$.
We denote by $T(\langle G, \ell \rangle)$ the number of triangles in $G$.
In particular, $T(\langle G, \ell \rangle) = 0$ if $G$ is empty.
Furthermore, $T(\langle G, \ell \rangle) = T(\langle G', \ell' \rangle)$ if $\langle G, \ell \rangle$ is obtained from $\langle G', \ell' \rangle$ by: the addition of a new vertex with any label, or the identification of two labels.
If $\langle G, \ell \rangle$ is the disjoint union of $\langle G_1, \ell_1 \rangle$ and $\langle G_2, \ell_2 \rangle$ then $T(\langle G, \ell \rangle) = T(\langle G_1, \ell_1 \rangle) + T(\langle G_2, \ell_2 \rangle)$.

Finally, suppose that $\langle G, \ell \rangle$ is obtained from $\langle G', \ell' \rangle$ by adding a complete join between the set $V_i$ of vertices labeled $i$ and the set $V_j$ of vertices labeled $j$.
For every $p,q \in \{1, \ldots, k\}$, we denote by $m_{p,q}$ the number of edges in $\langle G', \ell' \rangle$ with one end in $V_p$ and the other end in $V_q$.
Let $n_{p,q}$ be the number of (non necessarily induced) $P_3$'s with an end in $V_p$ and the other end in $V_q$.
Note that we are only interested in the number of {\em induced} $P_3$'s for our algorithm, but this looks more challenging to compute.
Nevertheless, since the $k$-expression is irredundant, $n_{i,j}$ is exactly the number of induced $P_3$'s with one end in $V_i$ and the other in $V_j$.
Furthermore after the join is added we get: $|V_i|$ new triangles per edge in $G'[V_j]$, $|V_j|$ new triangles per edge in $G'[V_i]$, and one triangle for every $P_3$ with one end in $V_i$ and the other in $V_j$.
Summarizing:
$$T(\langle G, \ell \rangle) = T(\langle G', \ell' \rangle) + |V_j| \cdot m_{i,i} + |V_i| \cdot m_{j,j} + n_{i,j}.$$
In order to derive the claimed time bound, we are now left to prove that, after any operation, we can update the values $m_{p,q} \ \mbox{and} \ n_{p,q}, \ p,q \in \{1, \ldots, k\}$, in ${\cal O}(k^2)$-time.
Clearly, these values cannot change when we add a new (isolated) vertex, with any label, and they can be updated by simple summation when we take the disjoint union of two labeled graphs.
We now need to distinguish between the two remaining cases.
In what follows, let $m_{p,q}'$ and $n_{p,q}'$ represent the former values.
\begin{itemize}
\item 
Suppose that label $i$ is identified with label $j$.
Then: 
$$m_{p,q} = \begin{cases} 
0 & \mbox{if} \ i \in \{p,q\} \\ 
m_{i,i}' + m_{i,j}' + m_{j,j}' & \mbox{if} \ p=q=j \\
m_{p,j}' + m_{p,i}' & \mbox{if} \ q = j, \ p \notin \{i,j\} \\ 
m_{j,q}' + m_{i,q}' & \mbox{if} \ p = j, \ q \notin \{i,j\} \\ 
m_{p,q}' & \mbox{else}
\end{cases},$$
$$n_{p,q} = \begin{cases} 0 & \mbox{if} \ i \in \{p,q\} \\ 
n_{j,j}' + n_{j,i}' + n_{i,i}' & \mbox{if} \ p=q=j \\
n_{p,j}' + n_{p,i}' & \mbox{if} \ q = j, \ p \notin \{i,j\} \\ 
n_{j,q}' + n_{i,q}' & \mbox{if} \ p = j, \ q \notin \{i,j\}  \\ 
n_{p,q}' & \mbox{else}\end{cases}.$$
\item Otherwise, suppose that we add a complete join between the set $V_i$ of vertices labeled $i$ and the set $V_j$ of vertices labeled $j$.
Then, since the $k$-expression is irredundant: $$m_{p,q} = \begin{cases} |V_i| \cdot |V_j| & \mbox{if} \ \{i,j\} = \{p,q\} \\ m_{p,q}' & \mbox{else}\end{cases}.$$
For every $u_i,v_i \in V_i$ and $w_j \in V_j$ we create a new $P_3$ $(u_i,w_j,v_i)$.
Similarly, for every $u_j,v_j \in V_j$ and $w_i \in V_i$ we create a new $P_3$ $(u_j,w_i,v_j)$.
These are the only new $P_3$'s with two edges from the complete join.
Furthermore, for every edge $\{u_i,v_i\}$ in $V_i$ and for every $w_j \in V_j$ we can create the two new $P_3$'s $(u_i,v_i,w_j)$ and $(v_i,u_i,w_j)$.
Similarly, for every edge $\{u_j,v_j\}$ in $V_j$ and for every $w_i \in V_i$ we can create the two new $P_3$'s $(u_j,v_j,w_i)$ and $(v_j,u_j,w_i)$.  
Finally, for every edge $\{v,u_j\}$ with $u_j \in V_j, \ v \notin V_i \cup V_j$, we create $|V_i|$ new $P_3$'s, and for every edge $\{v,u_i\}$ with $u_i \in V_i, \ v \notin V_i \cup V_j$, we create $|V_j|$ new $P_3$'s.
Altogether combined, we deduce the following update rules:
$$n_{p,q} = \begin{cases} n_{i,i}' + |V_i| \cdot |V_j| \cdot (|V_i|-1)/2 & \mbox{if} \ p = q = i \\ 
n_{j,j}' + |V_j| \cdot |V_i| \cdot (|V_j|-1)/2 & \mbox{if} \ p = q = j \\ 
n_{i,j}' + 2 \cdot |V_j| \cdot m_{i,i}' + 2 \cdot |V_i| \cdot m_{j,j}' & \mbox{if} \ \{p,q\} = \{i,j\} \\
n_{i,q}' + |V_i| \cdot m_{j,q}' & \mbox{if} \ p = i, q \notin \{i,j\} \\
n_{p,i}' + |V_i| \cdot m_{p,j}' & \mbox{if} \ q = i, p \notin \{i,j\} \\
n_{j,q}' + |V_j| \cdot m_{i,q}' & \mbox{if} \ p = j, q \notin \{i,j\} \\
n_{p,j}' + |V_j| \cdot m_{p,i}' & \mbox{if} \ q = j, p \notin \{i,j\} \\
n_{p,q}' & \mbox{else}\end{cases}.$$
\end{itemize}
\end{proof}

Our next result is about computing the {\em girth} of a graph (size of a smallest cycle).
To the best of our knowledge, the following Theorem~\ref{thm:cw-girth} gives the first known polynomial parameterized algorithm for {\sc Girth}.

\begin{theorem}\label{thm:cw-girth}
For every $G=(V,E)$, \textsc{Girth} can be solved in ${\cal O}(k^2 \cdot (n+m))$-time if a $k$-expression of $G$ is given.
\end{theorem}

\begin{proof}
The same as for Theorem~\ref{thm:cw-triangle}, we assume the $k$-expression to be irredundant.
It can be enforced up to linear-time preprocessing~\cite{CoB00}.
We proceed by dynamic programming on the $k$-expression.
More precisely, let $\langle G, \ell \rangle$ be a labeled graph, $\ell : V(G) \to \{1, \ldots, k\}$.
We denote by $\mu(\langle G, \ell \rangle)$ the girth of $G$.
By convention, $\mu(\langle G, \ell \rangle) = +\infty$ if $G$ is empty, or more generally if $G$ is a forest.
Furthermore, $\mu(\langle G, \ell \rangle) = \mu(\langle G', \ell' \rangle)$ if $\langle G, \ell \rangle$ is obtained from $\langle G', \ell' \rangle$ by: the addition of a new vertex with any label, or the identification of two labels.
If $\langle G, \ell \rangle$ is the disjoint union of $\langle G_1, \ell_1 \rangle$ and $\langle G_2, \ell_2 \rangle$ then $\mu(\langle G, \ell \rangle) = \min \{\mu(\langle G_1, \ell_1 \rangle), \mu(\langle G_2, \ell_2 \rangle) \}$.

Suppose that $\langle G, \ell \rangle$ is obtained from $\langle G', \ell' \rangle$ by adding a complete join between the set $V_i$ of vertices labeled $i$ and the set $V_j$ of vertices labeled $j$.
For every $p,q \in \{1, \ldots, k\}$, we are interested in the minimum length of a {\em nonempty} path with an end in $V_p$ and an end in $V_q$.
However, for making easier our computation, we consider a slightly more complicated definition.
If $p \neq q$ then we define $d_{p,q}$ as the minimum length of a $V_pV_q$-path of $G'$.
Otherwise, $p=q$, we define $d_{p,q}$ as the minimum length taken over all the paths with two distinct ends in $V_p$, and all the nontrivial closed walks that intersect $V_p$ ({\it i.e.}, there is at least one edge in the walk, we allow repeated vertices or edges, however a same edge does not appear twice consecutively).
Intuitively, $d_{p,p}$ may not represent the length of a path only in some cases where a cycle of length at most $d_{p,p}$ is already ensured to exist in the graph (in which case we needn't consider this value).  
Furthermore note that such paths or closed walks as defined above may not exist.
So, we may have $d_{p,q} = + \infty$. 
Then, let us consider a minimum-size cycle $C$ of $G$.
We distinguish between four cases.
\begin{itemize}
\item
If $C$ does not contain an edge of the join, then it is a cycle of $G'$.
\item
Else, suppose that $C$ contains exactly one edge of the join.
Then removing this edge leaves a $V_iV_j$-path in $G'$; this path has length at least $d_{i,j}$.
Conversely, if $d_{i,j} \neq +\infty$ then there exists a cycle of length $1 + d_{i,j}$ in $G$, and so, $\mu(\langle G, \ell \rangle) \leq 1 + d_{i,j}$.
\item 
Else, suppose that $C$ contains exactly two edges of the join.
In particular, since $C$ is of minimum-size, and so, it is an induced cycle, the two edges of the join in $C$ must have a common end in the cycle.
It implies that removing the two edges from $C$ leaves a path of $G'$ with either its two ends in $V_i$ or its two ends in $V_j$.
Such paths have respective length at least $d_{i,i}$ and $d_{j,j}$.
Conversely, there exist closed walks of respective length $2 + d_{i,i}$ and $2+d_{j,j}$ in $G$.
Hence, $\mu(\langle G, \ell \rangle) \leq 2 + \min \{ d_{i,i}, d_{j,j} \}$.
\item Otherwise, $C$ contains at least three edges of the join. Since $C$ is induced, it implies that $C$ is a cycle of length four with two vertices in $V_i$ and two vertices in $V_j$.
Such a (non necessarily induced) cycle exists if and only if $\min \{|V_i|,|V_j|\} \geq 2$.
\end{itemize}
Summarizing: $$\mu(\langle G, \ell \rangle) = \begin{cases} \min \{ \mu(\langle G', \ell' \rangle), 1 + d_{i,j}, 2 + d_{i,i}, 2 +d_{j,j} \} & \mbox{if} \ \min\{|V_i|,|V_j|\} = 1 \\ \min \{ \mu(\langle G', \ell' \rangle), 1 + d_{i,j}, 2 + d_{i,i}, 2 +d_{j,j}, 4 \} & \mbox{otherwise.}\end{cases}$$
In order to derive the claimed time bound, we are now left to prove that, after any operation, we can update the values $d_{p,q}, \ p,q \in \{1, \ldots, k\}$, in ${\cal O}(k^2)$-time.
Clearly, these values cannot change when we add a new (isolated) vertex, with any label, and they can be updated by taking the minimum values when we take the disjoint union of two labeled graphs.
We now need to distinguish between the two remaining cases.
In what follows, let $d_{p,q}'$ represent the former values.
\begin{itemize}
\item 
Suppose that label $i$ is identified with label $j$.
Then: $$ d_{p,q} = \begin{cases} + \infty & \mbox{if} \ i \in \{p,q\} \\ \min\{d_{i,i}', d_{i,j}', d_{j,j}'\} & \mbox{if} \ p=q=j \\ \min\{ d_{p,i}', d_{p,j}' \} & \mbox{if} \ q = j \\ \min\{ d_{i,q}', d_{j,q}' \} &\mbox{if} \ p = j \\ d_{p,q}' & \mbox{else}.\end{cases}$$
\item Otherwise, suppose that we add a complete join between the set $V_i$ of vertices labeled $i$ and the set $V_j$ of vertices labeled $j$.
The values $d_{p,q}'$ can only be decreased by using the edges of the join. 
In particular, using the fact that the $k$-expression is irredundant, we obtain: $$ d_{p,q} = \begin{cases} 
1 & \mbox{if} \ \{p,q\} = \{i,j\} \\
\min\{ 2, d_{p,q}' \} & \mbox{if} \ p=q=i, \ |V_i| \geq 2 \ \mbox{or} \ p=q=j, \ |V_j| \geq 2 \\ 
\min\{ d_{i,i}', 1 + d_{i,j}', 2 + d_{j,j}' \} & \mbox{if} \ p=q=i, \ |V_i| = 1 \\
\min\{ d_{j,j}', 1 + d_{i,j}', 2 + d_{i,i}' \} & \mbox{if} \ p=q=j, \ |V_j| = 1 \\ 
\min\{ d_{i,q}', 1 + d_{j,q}' \} & \mbox{if} \ p = i, \ q \notin \{i,j\} \\
\min\{ d_{p,i}', d_{p,j}' + 1 \} & \mbox{if} \ q = i, \ p \notin \{i,j\} \\
\min\{ d_{j,q}', 1 + d_{i,q}' \} & \mbox{if} \ p = j, \ q \notin \{i,j\} \\
\min\{ d_{p,j}', d_{p,i}' + 1 \} & \mbox{if} \ q = j, \ p \notin \{i,j\}
\end{cases}.$$
\end{itemize}
For all the remaining values of $p$ and $q$, the difficulty is to account for the cases where two consecutive edges of the join must be used in order to decrease the value $d_{p,q}'$.
We do so by using the {\em updated} values $d_{p,i},d_{i,q},d_{p,j},d_{j,q}$ instead of the former values $d_{p,i}',d_{i,q}',d_{j,p}',d_{q,j}'$.
More precisely,
$$d_{p,q} = \min \{ d_{p,q}', d_{p,i} + 1 + d_{j,q}, d_{p,j} + 1 + d_{i,q} \} \ \mbox{else}.$$
\end{proof}

The bottleneck of the above algorithms is that they require a $k$-expression as part of the input. 
So far, the best-known approximation algorithms for clique-width run in ${\cal O}(n^3)$-time, that dominates the total running time of our algorithms~\cite{OuS06}.
However, on a more positive side a $k$-expression can be computed in linear time for many classes of bounded clique-width graphs. 
In particular, combining Theorems~\ref{thm:cw-triangle} and~\ref{thm:cw-girth} with Lemmas~\ref{lem:mw-cw},~\ref{lem:sw-cw} and~\ref{lem:qq3-cw} we obtain the following result.

\begin{corollary}
For every $G=(V,E)$, {\sc Triangle Counting} and {\sc Girth} can be solved in ${\cal O}(k^2 \cdot (n+m))$-time, for every $k \in \{mw(G),sw(G),q(G)\}$.
\end{corollary}

\section{Parameterization, Hardness and Kernelization for some distance problems on graphs}\label{sec:dist}

We prove separability results between clique-width and the upper-bounds for clique-width presented in Section~\ref{sec:prelim}.
More precisely, we consider the problems {\sc Diameter}, {\sc Eccentricities}, {\sc Hyperbolicity} and {\sc Betweenness Centrality} (defined in Section~\ref{sec:dist-pbs}), that have already been well studied in the field of ``Hardness in P''. 
On the negative side, we show in Section~\ref{sec:hardness-cw} that we cannot solve these three above problems with a {\em fully polynomial} parameterized algorithm, when parameterized by clique-width.
However, on a more positive side, we prove the existence of such algorithms in Sections~\ref{sec:algos-split},~\ref{sec:algos-modular-dec} and~\ref{sec:dist-qq3}, when parameterized by either the modular-width, the split-width or the $P_4$-sparseness.

\subsection{Distance problems considered}\label{sec:dist-pbs}

\subsubsection*{Eccentricity-based problems}

The first problem considered is computing the diameter of a graph (maximum length of a shortest-path).

\begin{center}
	\fbox{
		\begin{minipage}{.95\linewidth}
			\begin{problem}[\textsc{Diameter}]\
				\label{prob:diam} 
					\begin{description}
					\item[Input:] A graph $G=(V,E)$.
					\item[Output:] The diameter of $G$, that is $\max_{u,v \in V} dist_G(u,v)$.
				\end{description}
			\end{problem}     
		\end{minipage}
	}
\end{center}

Hardness results for {\sc Diameter} have been proved, {\it e.g.},  in~\cite{RoV13,AGV15,BCH16,AVW16,EvD16}.

Our new hardness results are proved for {\sc Diameter}, while our fully polynomial parameterized algorithms apply to the following more general version of the problem. 
The {\em eccentricity} of a given vertex $v$ is defined as $ecc_G(v) = \max_{u \in V} dist_G(u,v)$.
Observe that $diam(G) = \max_v ecc_G(v)$. 

\begin{center}
	\fbox{
		\begin{minipage}{.95\linewidth}
			\begin{problem}[\textsc{Eccentricities}]\
				\label{prob:ecc} 
					\begin{description}
					\item[Input:] A graph $G=(V,E)$.
					\item[Output:] The eccentricities of the vertices in $G$, that is $\max_{u \in V} dist_G(u,v)$ for every $v \in V$.
				\end{description}
			\end{problem}     
		\end{minipage}
	}
\end{center}

\subsubsection*{Gromov hyperbolicity}

Then, we consider the parameterized complexity of computing the Gromov hyperbolicity of a given graph.
Gromov hyperbolicity is a measure of how close (locally) the shortest-path metric of a graph is to a tree metric~\cite{Gro87}.
We refer to~\cite{ducoffe:tel-01485328} for a survey on the applications of Gromov hyperbolicity in computer science.

\begin{boxedproblem}{Hyperbolicity}{prob:hyp}
\begin{description}
  \item[Input:] A graph $G=(V,E)$.
  \item[Output:] The hyperbolicity $\delta$ of $G$, that is:
  \begin{align*}
  \max_{u,v,x,y \in V} \frac {dist_G(u,v) + dist_G(x,y) - \max \{ dist_G(u,x) + dist_G(v,y), dist_G(u,y) + dist_G(v,y) \}} 2.
  \end{align*}
\end{description}
\end{boxedproblem}

Hardness results for {\sc Hyperbolicity} have been proved in~\cite{BCH16,CoDu14,FIV15}.
Some fully polynomial parameterized algorithms, with a different range of parameters than the one considered in this work, have been designed in~\cite{FKMN+17}.

\subsubsection*{Centrality problems}

There are different notions of centrality in graphs.
For clarity, we choose to keep the focus on one centrality measurement, sometimes called the Betweenness Centrality~\cite{Free77}.
More precisely, let $G=(V,E)$ be a graph and let $s,t \in V$.
We denote by $\sigma_G(s,t)$ the number of shortest $st$-paths in $G$.
In particular, for every $v \in V$, $\sigma_G(s,t,v)$ is defined as the number of shortest $st$-paths passing by $v$ in $G$.

\begin{boxedproblem}{Betweenness Centrality}{prob:bc}
\begin{description}
  \item[Input:] A graph $G=(V,E)$.
  \item[Output:] The betweenness centrality of every vertex $v \in V$, defined as:
  \begin{align*}
  BC_G(v) = \sum_{s,t \in V \setminus v} \sigma_G(s,t,v)/\sigma_G(s,t).
  \end{align*}
\end{description}
\end{boxedproblem}

See~\cite{AGV15,BCH16,EvD16} for hardness results on {\sc Betweenness Centrality}.

\subsection{Hardness results for clique-width}\label{sec:hardness-cw}

The goal in this section is to prove that we cannot solve the problems of Section~\ref{sec:dist-pbs} in time $2^{o(cw)} n^{2 - \varepsilon}$, for any $\varepsilon > 0$ (Theorems~\ref{thm:cw-diam}---\ref{thm:cw-hyp}).
These are the first known hardness results for clique-width in the field of ``Hardness in P''.
Our results are conditioned on the Strong Exponential Time Hypothesis ({\sc SETH}): {\sc SAT} cannot be solved in ${\cal O}^*(2^{c \cdot n})$-time, for any $c < 1$~\cite{IPZ98}.
Furthermore, they are derived from similar hardness results obtained for {\em treewidth}.

Precisely, a {\it tree decomposition} $(T,{\cal X})$ of $G=(V,E)$ is a pair consisting of a tree $T$ and of a family ${\cal X}=(X_t)_{t \in V(T)}$ of subsets of $V$ indexed by the nodes of $T$ and satisfying:
\begin{itemize}
	\item $\bigcup_{t \in V(T)}X_t=V$;
	\item for any edge $e=\{u,v\} \in E$, there exists $t\in V(T)$ such that $u,v \in X_t$;
	\item for any $v \in V$, the set of nodes $\{ t \in V(T) \mid v \in X_t\}$ induces a subtree, denoted by $T_v$, of $T$. 
\end{itemize}
The sets $X_t$ are called {\it the bags} of the decomposition.
The {\it width} of a tree decomposition is the size of a largest bag minus one.
Finally, the {\it treewidth} of a graph $G$, denoted by $tw(G)$, is the least possible width over its tree decompositions.

\smallskip
Several hardness results have already been obtained for treewidth~\cite{AVW16}.
However, $cw(G) \leq 2^{tw(G)}$ for general graphs~\cite{CoR05}, that does not help us to derive our lower-bounds.
Roughly, we use relationships between treewidth and clique-width in some graph classes ({\it i.e.}, bounded-degree graphs~\cite{Cou12,GuW00}) in order to transpose the hardness results for treewidth into hardness results for clique-width.
Namely:

\begin{lemma}[\cite{Cou12,GuW00}]\label{lem:rel-tw-cw}
If $G$ has maximum degree at most $d$ (with $d \geq 1$), we have:
\begin{itemize}
\item $tw(G) \leq 3d \cdot cw(G) - 1$;
\item $cw(G) \leq 20d \cdot tw(G) + 22$.
\end{itemize}
\end{lemma}

Our reductions in what follows are based on Lemma~\ref{lem:rel-tw-cw}, and on previous hardness results for bounded treewidth graphs and bounded-degree graphs~\cite{AVW16,EvD16}. 

\begin{theorem}\label{thm:cw-diam}
Under {\sc SETH}, we cannot solve {\sc Diameter} in $2^{o(k)} \cdot n^{2-\varepsilon}$-time on graphs with maximum degree $4$ and treewidth at most $k$, for any $\varepsilon > 0$.

In particular, we cannot solve {\sc Diameter} in $2^{o(k)} \cdot n^{2-\varepsilon}$-time on graphs with clique-width at most $k$, for any $\varepsilon > 0$. 
\end{theorem}

\begin{proof}
In~\cite{AVW16}, they proved that under {\sc SETH}, we cannot solve {\sc Diameter} in ${\cal O}(n^{2-\varepsilon})$-time, for any $\varepsilon > 0$, in the class of tripartite graphs $G = (A \cup C \cup B, E)$ such that: $|A| = |B| = n$, $|C| = {\cal O}(\log n)$, and all the edges in $E$ are between $C$ and $A \cup B$.
Note that there exists a tree decomposition $(T,{\cal X})$ of $G$ such that $T$ is a path and the bags are the sets $\{a\} \cup C, \ a \in A$ and $\{b\} \cup C, \ b \in B$.
Hence, $tw(G) = {\cal O}(|C|) = {\cal O}(\log n)$~\cite{AVW16}.

Then, we use the generic construction of~\cite{EvD16} in order to transform $G$ into a bounded-degree graph.
We prove that graphs with treewidth ${\cal O}(\log n)$ can be generated with this construction\footnote{Our construction has less degrees of freedom than the construction presented in~\cite{EvD16}.}.
More precisely, let $T_{big}$ and $T_{small}$ be rooted balanced binary trees with respective number of leaves $|A| = |B| = n$ and $|C| = {\cal O}(\log n)$.
There is a bijective correspondance between the leaves of $T_{big}$ and the vertices in $A$, resp. between the leaves of $T_{big}$ and the vertices in $B$.
Similarly, there is a bijective correspondance between the leaves of $T_{small}$ and the vertices of $C$.
In order to construct $G'$ from $G$, we proceed as follows:
\begin{itemize}
\item We replace every vertex $u \in A \cup B$ with a disjoint copy $T_{small}^u$ of $T_{small}$.
We also replace every vertex $c \in C$ with two disjoint copies $T_{big}^{c,A}, T_{big}^{c,B}$ of $T_{big}$ with a common root.
\item For every $a \in A, c \in C$ adjacent in $G$, we add a path of length $p$ (fixed by the construction) between the leaf of $T_{small}^a$ corresponding to $c$ and the leaf of $T_{big}^{c,A}$ corresponding to $a$.
In the same way, for every $b \in B, c \in C$ adjacent in $G$, we add a path of length $p$ between the leaf of $T_{small}^b$ corresponding to $c$ and the leaf of $T_{big}^{c,B}$ corresponding to $b$.
\item Let $T_{big}^{A}$ and $T_{big}^{B}$ be two other disjoint copies of $T_{big}$.
For every $a \in A$ we add a path of length $p$ between the leaf corresponding to $a$ in $T_{big}^{A}$ and the root of $T_{small}^a$.
In the same way, for every $b \in B$ we add a path of length $p$ between the leaf corresponding to $b$ in $T_{big}^{B}$ and the root of $T_{small}^b$.
\item Finally, for every $u \in A \cup B$, we add a path of length $p$ with one end being the root of $T_{small}^u$.
\end{itemize}
The resulting graph $G'$ has maximum degree $4$.
In~\cite{EvD16}, they prove that, under {\sc SETH}, we cannot compute $diam(G')$ in ${\cal O}(n^{2-\varepsilon})$-time, for any $\varepsilon > 0$.

\smallskip
We now claim that $tw(G') = {\cal O}(\log n)$.
Since by Lemma~\ref{lem:rel-tw-cw} we have $tw(G') = \Theta(cw(G'))$ for bounded-degree graphs, it will imply $cw(G') = {\cal O}(\log n)$.
In order to prove the claim, we assume w.l.o.g. that the paths added by the above construction have length $p=1$\footnote{The hardness result of~\cite{EvD16} holds for $p=\omega(\log n)$. We reduce to the case $p=1$ {\em only} for computing the treewidth.}.
Indeed, subdividing an edge does not change the treewidth~\cite{BoK06}.
Note that in this situation, we can also ignore the pending vertices added for the last step of the construction.
Indeed, removing the pending vertices does not change the treewidth either~\cite{Bod06}.
Hence, from now on we consider the graph $G'$ resulting from the three first steps of the construction by taking $p=1$. 

Let $(T',{\cal X}')$ be a tree decomposition of $T_{big}$ of unit width.
There is a one-to-one mapping between the nodes $t \in V(T')$ and the edges $e_t \in E(T_{big})$.
Furthermore, let $e_t^A, \ e_t^B, e_t^{c,A} \ \mbox{and} \ e_t^{c,B}$ be the copies of edge $e_t$ in the trees $T_{big}^A, \ T_{big}^B, \ T_{big}^{c,A} \ \mbox{and} \ T_{big}^{c,B}, \ c \in C$, respectively.
For every node $t \in V(T')$, we define a new bag $Y_t$ as follows.
If $e_t$ is not incident to a leaf-node then we set $Y_t =  e_t^A \cup e_t^B \cup \left[ \bigcup_{c \in C} \left( e_t^{c,A} \cup e_t^{c,B} \right) \right]$.
Otherwise, $e_t$ is incident to some leaf-node.
Let $a_t \in A, \ b_t \in B$ correspond to the leaf.
We set $Y_t =  V\left(T_{small}^{a_t}\right) \cup V\left(T_{small}^{b_t}\right) \cup e_t^A \cup e_t^B \cup \left[ \bigcup_{c \in C} \left( e_t^{c,A} \cup e_t^{c,B} \right) \right]$.
By construction, $(T', (Y_t)_{t \in V(T')})$ is a tree decomposition of $G'$.
In particular, $tw(G') \leq \max_{t \in V(T')} |Y_t| = {\cal O}(\log n)$, that finally proves the claim.

\smallskip
Finally, suppose by contradiction that $diam(G')$ can be computed in $2^{o(tw(G'))} \cdot n^{2-\varepsilon}$-time, for some $\varepsilon > 0$.
Since $tw(G') = {\cal O}(\log n)$, it implies that $diam(G')$ can be computed in ${\cal O}(n^{2-\varepsilon})$-time, for some $\varepsilon > 0$.
The latter refutes {\sc SETH}.
Hence, under {\sc SETH} we cannot solve {\sc Diameter} in $2^{o(k)} \cdot n^{2-\varepsilon}$-time on graphs with maximum degree $4$ and treewidth at most $k$, for any $\varepsilon > 0$. 
This negative result also holds for clique-width since $cw(G') = \Theta(tw(G'))$.
\end{proof}

The following reduction to {\sc Betweenness Centrality} is from~\cite{EvD16}.
Our main contribution is to upper-bound the clique-width and the treewidth of their construction.

\begin{theorem}\label{thm:cw-bc}
Under {\sc SETH}, we cannot solve {\sc Betweenness Centrality} in $2^{o(k)} \cdot n^{2-\varepsilon}$-time on graphs with maximum degree $4$ and treewidth at most $k$, for any $\varepsilon > 0$.

In particular, we cannot solve {\sc Betweenness Centrality} in $2^{o(k)} \cdot n^{2-\varepsilon}$-time on graphs with clique-width at most $k$, for any $\varepsilon > 0$. 
\end{theorem}

\begin{proof}
Let $G'$ be the graph from the reduction of Theorem~\ref{thm:cw-diam}.
In~\cite{EvD16}, the authors propose a reduction from $G'$ to $H$ such that, under {\sc SETH}, we cannot solve {\sc Betweenness Centrality} for $H$ in ${\cal O}(n^{2-\varepsilon})$-time, for any $\varepsilon > 0$. 
In order to prove the theorem, it suffices to prove $tw(H) = \Theta(tw(G'))$.
Indeed, the construction of $H$ from $G'$ is as follows.
\begin{itemize}
\item For every $u \in A \cup B$, we remove the path of length $p$ with one end being the root of $T_{small}^u$, added at the last step of the construction of $G'$.
This operation can only decrease the treewidth.
\item Then, we add a path of length $p$ between the respective roots of $T_{big}^A$ and $T_{big}^B$.
Recall that we can assume $p=1$ since subdividing an edge does not modify the treewidth~\cite{BoK06}.
Adding an edge to a graph increases its treewidth by at most one.
\item Finally, let $H_1$ be the graph so far constructed.
We make a disjoint copy $H_2$ of $H_1$.
This operation does not modify the treewidth.
Then, for every $b \in B$, let ${T_{small}^b}'$ be the copy of $T_{small}^b$ in $H_2$.
We add a new vertex $b'$ that is uniquely adjacent to the root of ${T_{small}^b}'$.
The addition of pending vertices does not modify the treewidth either~\cite{Bod06}. 
\end{itemize}
Overall, $tw(H) \leq tw(G') + 1 = {\cal O}(\log n)$.
Furthermore, since $H$ has maximum degree at most $4$, by Lemma~\ref{lem:rel-tw-cw} we have $cw(H) = \Theta(tw(H)) = {\cal O}(\log n)$.
\end{proof}

Our next reduction for {\sc Hyperbolicity} is inspired from the one presented in~\cite{BCH16}.
However, the authors in~\cite{BCH16} reduce from a special case of {\sc Diameter} where we need to distinguish between graphs with diameter either $2$ or $3$.
In order to reduce from a more general case of {\sc Diameter} we need to carefully refine their construction.

\begin{theorem}\label{thm:cw-hyp}
Under {\sc SETH}, we cannot solve {\sc Hyperbolicity} in $2^{o(k)} \cdot n^{2-\varepsilon}$-time on graphs with clique-width and treewidth at most $k$, for any $\varepsilon > 0$.
\end{theorem}

\begin{proof}
We use the graph $G'$ from the reduction of Theorem~\ref{thm:cw-diam}.
More precisely, let us take $p = \omega(\log n)$ for the size of the paths in the construction.
It has been proved in~\cite{EvD16} that either $diam(G') = (4+o(1)) p$ or $diam(G') = (6+o(1)) p$.
Furthermore, under {\sc SETH} we cannot decide in which case we are in truly subquadratic time.

Our reduction is inspired from~\cite{BCH16}.
Let $H$ be constructed from $G'$ as follows (see also Fig.~\ref{fig:reduction-hyp}).
\begin{itemize}
\item We add two disjoint copies $V_x,V_y$ of $V(G')$ and the three vertices $x,y,z \notin V(G')$.
We stress that $V_x$ and $V_y$ are independent sets.
Furthermore, for every $v \in V$, we denote by $v_x$ and $v_y$ the copies of $v$ in $V_x$ and $V_y$, respectively.

\item For every $v \in V(G')$, we add a $vv_x$-path $P_v^x$ of length $(3/2+o(1))p$, and similarly we add a $vv_y$-path $P_v^y$ of length $(3/2+o(1))p$.

\item Furthermore, for every $v \in V(G')$ we also add a $xv_x$-path $Q_v^x$ of length $(3/2+o(1))p$; a $yv_y$-path $Q_v^y$ of length $(3/2+o(1))p$; a $zv_x$-path $Q_v^{z,x}$ of length $(3/2+o(1))p$ and a $zv_y$-path $Q_v^{z,y}$ of length $(3/2+o(1))p$.
\end{itemize}

\begin{figure}[h!]
\centering
\includegraphics[width=.45\textwidth]{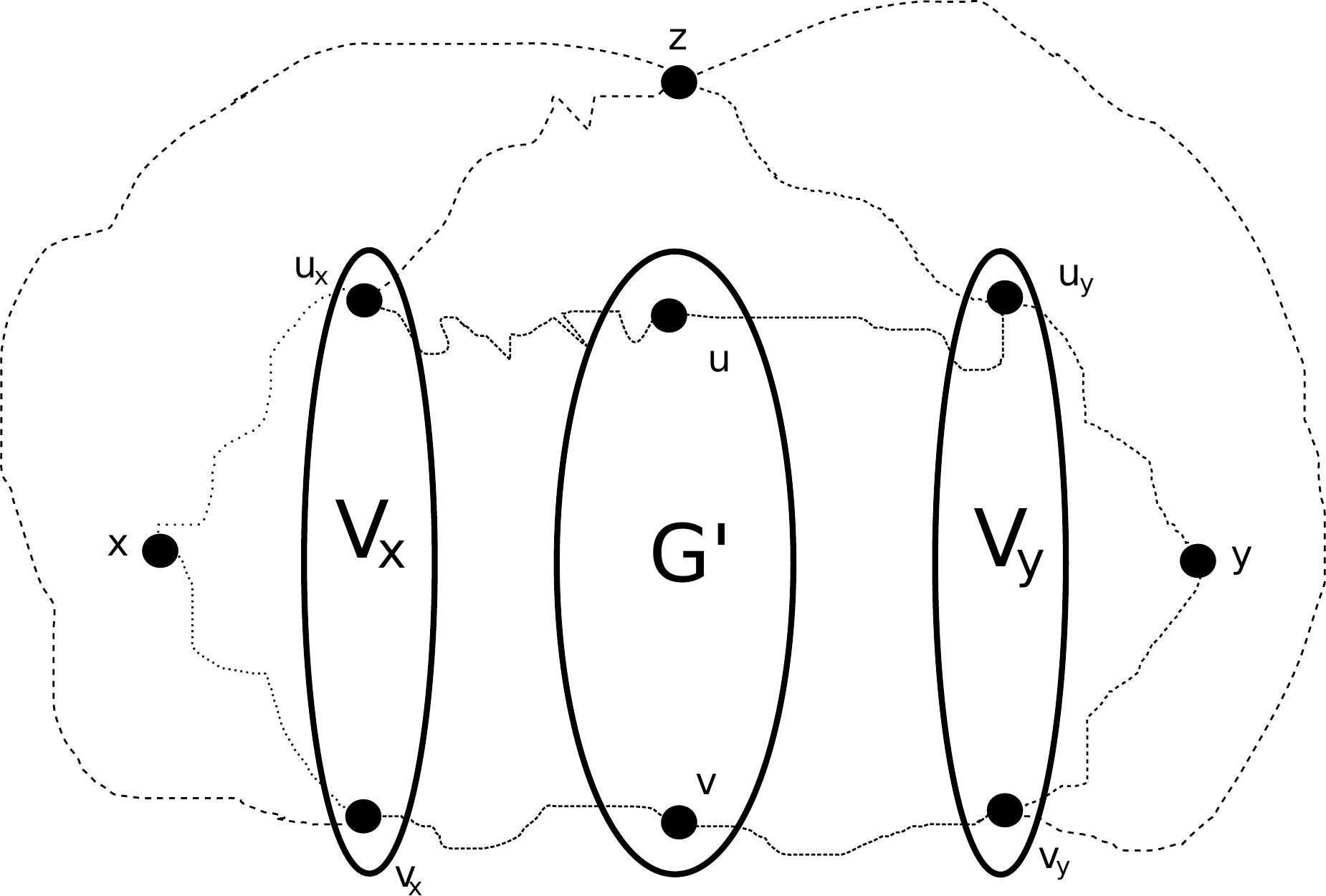}
\caption{The graph $H$ from the reduction of Theorem~\ref{thm:cw-hyp}.}
\label{fig:reduction-hyp}
\end{figure}

\smallskip
We claim that the resulting graph $H$ is such that $tw(H) = tw(G') + {\cal O}(1)$ and $cw(H) = cw(G') + {\cal O}(1)$.
Indeed, let us first consider $H' = H \setminus \{x,y,z\}$.
The graph $H'$ is obtained from $G'$ by adding some disjoint trees rooted at the vertices of $V(G')$.
In particular, it implies $tw(H') = tw(G')$, hence (by adding $x,y,z$ in every bag) $tw(H) \leq tw(H') + 3 \leq tw(G') + 3$.

Furthermore, let us fix a $k$-expression for $G'$. We transform it to a $(k+16)$-expression for $H$ as follows.
We start adding $x,y,z$ with three distinct new labels.
Then, we follow the $k$-expression for $G'$.
Suppose a new vertex $v \in V(G')$, with label $i$ is introduced.
It corresponds to some tree $T_v$ in $H'$, that is rooted at $v$.
Every such a tree has clique-width at most $3$~\cite{GoR00}.
So, as an intermediate step, let us fix a $3$-expression for $T_v$.
We transform it to a $12$-expression for $T_v$: with each new label encoding the former label in the $3$ expression ($3$ possibilities), and whether the node is either the root $v$ or adjacent to one of $x,y,z$ ($4$ possibilities). 
This way, we can make $x,y,z$ adjacent to their neighbours in $T_v$, using the join operation.
Then, since the root $v$ has a distinguished label, we can ``freeze'' all the other nodes in $T_v \setminus v$ using an additional new label and relabeling operations.
Finally, we relabel $v$ with its original label $i$ in the $k$-expression of $G'$, and then we continue following this $k$-expression.
Summarizing, $cw(H) \leq cw(G') + 16$.

\smallskip
Next, we claim that $\delta(H) \geq (3+o(1))p$ if $diam(G') = (6+o(1))p$, while $\delta(H) \leq (11/4+o(1))p$ if $diam(G') = (4+o(1))p$. 
Recall that by Theorem~\ref{thm:cw-diam}, under {\sc SETH} we cannot decide in which case we are in time $2^{o(tw(G'))} n^{2-\varepsilon} = 2^{o(cw(G'))} n^{2-\varepsilon}$, for any $\varepsilon > 0$.
Therefore, proving the claim will prove the theorem.

First suppose that $diam(G') = (6+o(1))p$.
Let $u,v \in V(G')$ satisfy $dist_{G'}(u,v) = (6+o(1))p$.
Observe that $diam(G') \leq (6+o(1))p = 4 \cdot (3/2 + o(1))p$, therefore $G'$ is an isometric subgraph of $H$ by construction.
Then, $S_1 = dist_H(u,v) + dist_H(x,y) = (12+o(1))p$;
$S_2 = dist_H(u,x) + dist_H(v,y) = (6+o(1))p$;
$S_3 = dist_H(u,y) + dist_H(v,x) = S_2$.
As a result, we obtain $\delta(H) \geq (S_1 - \max\{S_2,S_3\})/2 = (3+o(1))p$.

Second, suppose that $diam(G') = (4+o(1))p$.
We want to prove $\delta(H) \leq (11/4+o(1))p$.
By contradiction, let $a,b,c,d \in V(H)$ satisfy:
$$S_1 = dist_H(a,b) + dist_H(c,d) \geq S_2 = dist_H(a,c) + dist_H(b,d) \geq S_3 = dist_H(a,d) + dist_H(b,c),$$ 
$$S_1 - S_2 >  (11/2+o(1))p.$$
The hyperbolicity of a given $4$-tuple is upper-bounded by the minimum distance between two vertices of the $4$-tuple~\cite{BCCM15,CCL15,Sot11}.
So, let us consider the distances in $H$.

\begin{itemize}
\item
Let $v \in V(G')$.
For every $u \in V(G'), \ dist_H(u,v) \leq dist_G(u,v) \leq (4+o(1))p$.

Furthermore for every $u' \in P_u^x$, $dist_H(v,u') \leq dist_H(v,u) + dist_H(u,u') \leq (11/2 + o(1))p$.
Similarly for every $u' \in P_u^y$, $dist_H(v,u') \leq dist_H(v,u) + dist_H(u,u') \leq (11/2 + o(1))p$.

For every $u' \in Q_u^x$, $dist_H(v,u') \leq dist_H(v,x) + dist_H(x,u') \leq (9/2+o(1))p$.
We prove in the same way that for every $u' \in Q_u^y \cup Q_u^{z,x} \cup Q_u^{z,y}, \ dist_H(v,u') \leq (9/2+o(1))p$.

Summarizing, $ecc_H(v) \leq (11/2 + o(1))p$.

\item
Let $v' \in P_v^x$, for some $v \in V(G')$.

For every $u \in V(G')$ and $u' \in P_u^x$ there are two cases.
Suppose that $dist_H(u',u_x) \leq p+o(1)$ or $dist_H(v',v_x) \leq p+o(1)$.
Then, $dist_H(v',u') \leq dist_H(v',v_x) + dist_H(v_x,u_x) + dist_H(u_x,u') \leq (1 + 3 + 3/2 + o(1))p = (11/2+o(1))p$.
Otherwise, $\max\{dist_H(u',u),dist_H(v',v)\} \leq (1/2+o(1))p$, and so, $dist_H(u',v') \leq dist_H(u',u) + dist_H(u,v) + dist_H(v',v) \leq (5+o(1))p$.
Similarly (replacing $u_x$ with $u_y$), for every $u' \in P_u^y$ we have $dist(v',u') \leq (11/2+o(1))p$.

For every $u' \in Q_u^x, \ dist_H(v',u') \leq dist_H(v',v_x) + dist_H(x,v_x) + dist_H(x,u') \leq (9/2+o(1))p$.
In the same way for every $u' \in Q_u^{z,x} \cup Q_u^{z,y}, \ dist_H(v',u') \leq dist_H(v',v_x) + dist_H(z,v_x) + dist_H(z,u') \leq (9/2+o(1))p$.

For every $u' \in Q_u^y$, we first need to observe that $dist_H(v_x,u_y) = (3+o(1))p\ \mbox{and} \ dist_H(v,y) = (3+o(1))p$.
In particular if $dist_H(v,v') \leq p +o(1)$ then, $dist_H(v',u') \leq dist_H(v,v') + dist_H(v,y) + dist_H(y,u') \leq (11/2+o(1))p$.
Otherwise, $dist_H(v',u') \leq dist_H(v',v_x) + dist_H(v_x,u_y) + dist_H(u_y,u') \leq (1/2 + 3 + 3/2 + o(1))p = (5+o(1))p$.

Summarizing, $ecc_H(v') \leq (11/2 + o(1))p$.

\item
Let $v' \in P_v^y$, for some $v \in V(G')$.
In the same way as above, we prove $ecc_H(v') \leq (11/2 + o(1))p$.

\item
Let $v' \in Q_v^{z,x} \cup Q_v^{z,y}$, for some $v \in V(G')$.

For every $u \in V(G')$ and for every $u' \in Q_u^{z,x} \cup Q_u^{z,y}$ we have $dist_H(v',u') \leq dist_H(v',z) + dist_H(z,u') \leq (3+o(1))p$.

For every $u' \in Q_u^x$ we have $dist_H(v',u') \leq dist_H(v',z) + dist_H(z,u_x) + dist_H(u_x,u') \leq (9/2+o(1))p$.
Similarly for every $u' \in Q_u^y$ we have $dist_H(v',u') \leq (9/2+o(1))p$.

Summarizing, $ecc_H(v') \leq (11/2 + o(1))p$.

\end{itemize}

In particular, every vertex in $H$ has eccentricity at most $(11/2 + o(1))p$, except maybe those in $\bigcup_{v \in V(G')} Q_v^x = X$ and those in $\bigcup_{v \in V(G')} Q_v^y = Y$.
However, $S_1 - S_2 \leq \min \{ dist_H(a,b), dist_H(c,d) \}$~\cite{CCL15}.
So, we can assume w.l.o.g. $a,c \in X$ and $b,d \in Y$.
Furthermore, $S_1 - S_2 \leq 2 \cdot dist_H(a,c)$~\cite{BCCM15,Sot11}.
Hence, $(11/2+o(1))p < S_1 - S_2 \leq 2 \cdot dist_H(a,c) \leq 2 \cdot ( dist_H(a,x) + dist_H(c,x) )$.
It implies $\max\{dist_H(a,x),dist_H(c,x)\} > (11/8 + o(1))p = (3/2 - 1/8 + o(1))p$.
Assume by symmetry that $dist_H(a,x) > (3/2 - 1/8 + o(1))p$.
Then, $dist_H(a,V_x)  < (1/8+o(1))p$.
However, $dist_H(a,c) \leq dist_H(a,V_x) + (3+o(1))p + dist_H(c,V_y) < (1/8+3+3/2+o(1))p < (11/2+o(1))p$.
A contradiction.
Therefore, we obtain as claimed that $\delta(H) \leq (11/4 + o(1))p$.
\end{proof}

It is open whether any of these above problems can be solved in time $2^{{\cal O}(k)} \cdot n$ on graphs with clique-width at most $k$ (resp., on graphs with treewidth at most $k$, see~\cite{AVW16,Hus16}).

\subsection{Parameterized algorithms with split decomposition}\label{sec:algos-split}

We show how to use split decomposition as an efficient preprocessing method for {\sc Diameter},  {\sc Eccentricities},  {\sc Hyperbolicity} and  {\sc Betweenness Centrality}.
Improvements obtained with modular decomposition will be discussed in Section~\ref{sec:algos-modular-dec}.
Roughly, we show that in order to solve the problems considered, it suffices to solve some {\em weighted} variant of the original problem for every split component (subgraphs of the split decomposition) separately.
However, weights intuitively represent the remaining of the graph, so, we need to account for some dependencies between the split components in order to define the weights properly.

In order to overcome this difficulty, we use in what follows a tree-like structure over the split components in order to design our algorithms.
A {\em split decomposition tree} of $G$ is a tree $T$ where the nodes are in bijective correspondance with the subgraphs of the split decomposition of $G$, and the edges of $T$ are in bijective correspondance with the simple decompositions used for their computation.

More precisely:
\begin{itemize}
\item If $G$ is either degenerate, or prime for split decomposition, then $T$ is reduced to a single node;
\item Otherwise, let $(A,B)$ be a split of $G$ and let $G_A = (A \cup \{b\}, E_A), \ G_B = (B \cup \{a\}, E_B)$ be the corresponding subgraphs of $G$.
We construct the split decomposition trees $T_A, T_B$ for $G_A$ and $G_B$, respectively.
Furthermore, the split marker vertices $a$ and $b$ are contained in a unique split component of $G_A$ and $G_B$, respectively.
We obtain $T$ from $T_A$ and $T_B$ by adding an edge between the two nodes that correspond to these subgraphs.
\end{itemize}

A split decomposition tree can be constructed in linear-time~\cite{CDR12}.

\subsubsection*{Diameter and Eccentricities}

\begin{lemma}\label{lem:diam-split}
Let $(A,B)$ be a split of $G=(V,E)$ and let $G_A = (A \cup \{b\}, E_A), \ G_B = (B \cup \{a\}, E_B)$ be the corresponding subgraphs of $G$.
Then, for every $u \in A$ we have: $$ecc_G(u) = \max \{ ecc_{G_A}(u), dist_{G_A}(u,b) + ecc_{G_B}(a) - 1\}.$$
\end{lemma}

\begin{proof}
Let $C = N_G(B) \subseteq A$ and $D = N_G(A) \subseteq B$.
In order to prove the claim, we first need to observe that, since $(A,B)$ is a split of $G$, we have, for every $v \in V$:
$$dist_G(u,v) = \begin{cases} dist_{G_A}(u,v) \ \mbox{if} \ v \in A \\
dist_G(u,C) + 1 + dist_G(v,D) \ \mbox{if} \ v \in B.\end{cases}$$
Furthermore, $dist_G(u,C) = dist_{G_A}(u,b) - 1$, and similarly $dist_G(v,D) = dist_{G_B}(v,a) - 1 \leq ecc_{G_B}(a) - 1$.
Hence, $ecc_G(u) \leq \max \{ ecc_{G_A}(u), dist_{G_A}(u,b) + ecc_{G_B}(a) - 1\}$.

Conversely, $ecc_{G_A}(u) = \max \{dist_{G_A}(u,b)\} \cup \{dist_{G_A}(u,v) \mid v \in A\} = \max \{dist_{G}(u,D)\} \cup \{dist_{G}(u,v) \mid v \in A\} \leq ecc_G(u)$.
In the same way, let $v \in B$ maximize $dist_G(v,C)$.
We have: $dist_G(u,v) = dist_{G_A}(u,b) + ecc_{G_B}(a) - 1 \leq ecc_G(u)$.
\end{proof}

\begin{theorem}\label{thm:sw-ecc}
For every $G=(V,E)$, {\sc Eccentricities} can be solved in ${\cal O}(sw(G)^2 \cdot n +m)$-time.

In particular, {\sc Diameter} can be solved in ${\cal O}(sw(G)^2 \cdot n +m)$-time.
\end{theorem}

\begin{proof}
Let $T$ be a split decomposition tree of $G$, with its nodes being in bijective correspondance with the split components $C_1, C_2, \ldots, C_k$.
It can be computed in linear-time~\cite{CDR12}.
We root $T$ in $C_1$.
For every $1 \leq i \leq k$, let $T_i$ be the subtree of $T$ that is rooted in $C_i$.
If $i > 1$ then let $C_{p(i)}$ be its parent in $T$.
By construction of $T$, the edge $\{C_{p(i)}, C_i\} \in E(T)$ corresponds to a split $(A_i,B_i)$ of $G$, where $V(C_i) \subseteq A_i$.
Let $G_{A_i} = (A_i \cup \{b_i\}, E_{A_i}), \ G_{B_i} = (B_i \cup \{a_i\}, E_{B_i})$ be the corresponding subgraphs of $G$.
We observe that $T_i$ is a split decomposition tree of $G_{A_i}$, $T \setminus T_i$ is a split decomposition tree of $G_{B_i}$.

Our algorithm proceeds in two main steps, with each step corresponding to a different traversal of the tree $T$.
First, let $G_1 = G$ and let $G_i = G_{A_i}$ for every $i > 1$.
We first compute, for every $1 \leq i \leq k$ and for every $v_i \in V(C_i)$, its eccentricity in $G_i$.
In order to do so, we proceed by dynamic programming on the tree $T$:
\begin{itemize}
\item If $C_i$ is a leaf of $T$ then {\sc Eccentricities} can be solved: in ${\cal O}(|V(C_i)|)$-time if $C_i$ induces a star or a complete graph; and in ${\cal O}(|V(C_i)|^3) = {\cal O}(sw(G)^2 \cdot |V(C_i)|)$-time else.
\item Otherwise $C_i$ is an internal node of $T$.
Let $C_{i_1}, C_{i_2}, \ldots, C_{i_l}$ be the children of $C_i$ in $T$.
Every edge $\{C_i, C_{i_t}\} \in E(T), \ 1 \leq t \leq l$ corresponds to a split $(A_{i_t},B_{i_t})$ of $G_i$, where $V(C_{i_t}) \subseteq A_{i_t}$.
We name $b_{i_t} \in V(C_{i_t}), \ a_{i_t} \in V(C_i)$ the vertices added after the simple decomposition.
Furthermore, let us define $e(a_{i_t}) = ecc_{G_{i_t}}(b_{i_t}) - 1$.
For every other vertex $u \in V(C_i) \setminus \{a_{i_1}, a_{i_2}, \ldots, a_{i_k}\}$, we define $e(u) = 0$.
Then, applying Lemma~\ref{lem:diam-split} for every split $(A_{i_t},B_{i_t})$ we get: $$\forall u \in V(C_i), \ ecc_{G_i}(u) = \max\limits_{v \in V(C_i)} dist_{C_i}(u,v) + e(v).$$
We distinguish between three cases.
\begin{enumerate}
\item
If $C_i$ is complete, then we need to compute $x_i \in V(C_i)$ maximizing $e(x_i)$, and $y_i \in V(C_i) \setminus \{x_i\}$ maximizing $e(y_i)$.
It can be done in ${\cal O}(|V(C_i)|)$-time.
Furthermore, for every $u \in V(C_i)$, we have $ecc_{G_i}(u) = 1 + e(x_i)$ if $u \neq x_i$, and $ecc_{G_i}(x_i) = \max \{e(x_i), 1 + e(y_i)\}$.
\item 
If $C_i$ is a star with center node $r$, then we need to compute a leaf $x_i \in V(C_i) \setminus \{r\}$  maximizing $e(x_i)$, and another leaf $y_i \in V(C_i) \setminus \{x_i,r\}$ maximizing $e(y_i)$.
It can be done in ${\cal O}(|V(C_i)|)$-time.
Furthermore, $ecc_{G_i}(r) = \max \{e(r), 1 + e(x_i)\}, \ ecc_{G_i}(x_i) = \max\{ e(x_i), 1 + e(r), 2 + e(y_i) \}$, and for every other $u \in V(C_i) \setminus \{x_i,r\}$ we have $ecc_{G_i}(u) = \max \{1 + e(r), 2 + e(x_i)\}$.
\item Otherwise, $|V(C_i)| \leq sw(G)$, and so, all the eccentricities can be computed in ${\cal O}(|V(C_i)||E(C_i)|) = {\cal O}(sw(G)^2 \cdot |V(C_i)|)$-time.
\end{enumerate}
\end{itemize}
Overall, this step takes total time ${\cal O}(sw(G)^2 \cdot \sum_i |V(C_i)|) = {\cal O}(sw(G)^2 \cdot n)$.
Furthermore, since $G_1 = G$, we have computed $ecc_G(v_1)$ for every $v_1 \in V(C_1)$.

\medskip
Second, for every $2 \leq i \leq k$, we recall that by Lemma~\ref{lem:diam-split}: $$\forall v_i \in V(G_i), \ ecc_{G}(v_i) = \max \{ecc_{G_i}(v_i), dist_{G_i}(v_i,b_i) + ecc_{G_{B_i}}(a_i) - 1\}.$$
In particular, since we have already computed $ecc_{G_i}(v_i)$ for every $v_i \in V(C_i)$ (and as a byproduct, $dist_{G_i}(v_i,b_i)$), we can compute $ecc_G(v_i)$ from $ecc_{G_{B_i}}(a_i)$.
So, we are left to compute $ecc_{G_{B_i}}(a_i)$ for every $2 \leq i \leq k$.
In order to do so, we proceed by reverse dynamic programming on the tree $T$.

More precisely, let $C_{p(i)}$ be the parent node of $C_i$ in $T$, and let $C_{j_0} = C_i, C_{j_1}, C_{j_2}, \ldots, C_{j_k}$ denote the children of $C_{p(i)}$ in $T$.
For every $0 \leq t \leq k$, the edge $\{C_{p(i)}, C_{j_t}\}$ represents a split $(A_{j_t},B_{j_t})$, where $V(C_{j_t}) \subseteq A_{j_t}$.  
So, there has been vertices $b_{j_t} \in V(C_{j_t}), \ a_{j_t} \in V(C_{p(i)})$ added by the corresponding simple decomposition.
We define $e'(a_{j_t}) = ecc_{G_{j_t}}(b_{j_t}) - 1$.
Furthermore, if $p(i) > 1$, let $C_{p^2(i)}$ be the parent of $C_{p(i)}$ in $T$.
Again, the edge $\{C_{p^2(i)},C_{p(i)}\}$ represents a split $(A_{p(i)},B_{p(i)})$, where $V(C_{p(i)}) \subseteq A_{p(i)}$.
So, there has been vertices $b_{p(i)} \in V(C_{p(i)}), \ a_{p(i)} \in V(C_{p^2(i)})$ added by the corresponding simple decomposition.
Let us define $e'(b_{p(i)}) = ecc_{G_{B_{p(i)}}}(a_{p(i)}) - 1$ (obtained by reverse dynamic programming on $T$).
Finally, for any other vertex $u \in V(C_{p(i)})$, let us define $e'(u) = 0$.
Then, by applying Lemma~\ref{lem:diam-split} it comes: $$\forall 0 \leq t \leq k, \ ecc_{G_{B_{i_t}}}(a_{i_t}) = \max\limits_{v \in V(C_{p(i)}) \setminus \{a_{i_t}\}} dist_{C_{p(i)}}(a_{i_t},v) + e'(v).$$
We can adapt the techniques of the first step in order to compute all the above values in ${\cal O}(sw(G)^2 \cdot |V(C_{p(i)})|)$-time.
Overall, the time complexity of the second step is also ${\cal O}(sw(G)^2 \cdot n)$.  

Finally, since a split decomposition can be computed in ${\cal O}(n+m)$-time, and all of the subsequent steps take ${\cal O}(sw(G)^2 \cdot n)$-time, the total running time of our algorithm is an ${\cal O}(sw(G)^2 \cdot n + m)$.
\end{proof}

\subsubsection*{Gromov hyperbolicity}

It has been proved in~\cite{Sot11} that for every graph $G$, if every split component of $G$ is $\delta$-hyperbolic then $\delta(G) \leq \max\{1,\delta\}$.
We give a self-contained proof of this result, where we characterize the gap between $\delta(G)$ and the maximum hyperbolicity of its split components.

\begin{lemma}\label{lem:hyp-split}
Let $(A,B)$ be a split of $G=(V,E)$ and let $C = N_G(B) \subseteq A, \ D = N_G(A) \subseteq B$.
Furthermore, let $G_A = (A \cup \{b\}, E_A), \ G_B = (B \cup \{a\}, E_B)$ be the corresponding subgraphs of $G$.

Then, $\delta(G) = \max \{\delta(G_A), \delta(G_B), \delta^*\}$ where:
$$\delta^* = 
\begin{cases} 
1 & \mbox{if neither} \ C \ \mbox{nor} \ D \ \mbox{is a clique};\\
1/2 & \mbox{if} \ \min\{|C|,|D|\} \geq 2 \ \mbox{and exactly one of} \ C \ \mbox{or} \ D \ \mbox{is a clique}; \\
0 & \mbox{otherwise}.
\end{cases}$$
\end{lemma}

\begin{proof}
Since $G_A,G_B$ are isometric subgraphs of $G$, we have $\delta(G) \geq \max \{\delta(G_A), \delta(G_B)\}$.
Conversely, for every $u,v,x,y \in V$ define $L$ and $M$ to be the two largest sums amongst $\{ dist_G(u,v) + dist_G(x,y), dist_G(u,x) + dist_G(v,y), dist_G(u,y) + dist_G(v,x) \}$.
Write $\delta(u,v,x,y) = (L-M)/2$.
Furthermore, assume that $\delta(u,v,x,y) = \delta(G)$.
W.l.o.g., $|\{u,v,x,y\} \cap A| \geq |\{u,v,x,y\} \cap B|$.
In particular, if $u,v,x,y \in A$ then $\delta(u,v,x,y) \leq \delta(G_A)$.
Otherwise, there are two cases.
\begin{itemize}
\item Suppose $|\{u,v,x,y\} \cap A| = 3$.
W.l.o.g., $y \in B$.
Then, for every $w \in \{u,v,x\}$ we have $dist_G(w,y) = dist_{G_A}(w,b) + dist_{G_B}(a,y) - 1$.
Hence, $\delta(u,v,x,y) = \delta(u,v,x,b) \leq \delta(G_A)$.
\item Otherwise, $|\{u,v,x,y\} \cap A| = 2$.
W.l.o.g. $x,y \in B$.
Observe that $M = dist_G(u,x) + dist_G(v,y) = dist_G(u,y) + dist_G(v,x) = dist_{G_A}(u,b) + dist_{G_A}(v,b) + dist_{G_B}(a,x) + dist_{G_B}(a,y) - 2$.
Furthermore, $L = dist_G(u,v) + dist_G(x,y) \leq dist_{G_A}(u,b) + dist_{G_A}(v,b) + dist_{G_B}(a,x) + dist_{G_B}(a,y)$.
Hence, $\delta(u,v,x,y) = \max \{0, L - M\}/2 \leq 1$.
In particular:
\begin{itemize}
\item Suppose $\min \{|C|,|D|\} = 1$.
Then, the $4$-tuple $u,v,x,y$ is disconnected by some cut-vertex $c$.
In particular, $M = dist_G(u,c) + dist_G(v,c) + dist_G(c,x) + dist_G(c,y) \geq L$, and so, $\delta(u,v,x,y) = 0$.
Thus we assume from now on that $\min \{|C|,|D|\} \geq 2$.
\item Suppose $L - M = 2$.
It implies both $a$ is on a shortest $xy$-path (in $G_B$) and $b$ is on a shortest $uv$-path (in $G_A$).
Since there can be no simplicial vertices on a shortest path, we obtain that neither $a$ nor $b$ can be simplicial.
Thus, $C$ and $D$ are not cliques.
Conversely, if $C$ and $D$ are not cliques then there exists an induced $C_4$ with two ends in $C$ and two ends in $D$.
As a result, $\delta(G) \geq 1$.
\item Suppose $L - M = 1$.
Either $C$ or $D$ is not a clique.
Conversely, if either $C$ or $D$ is not a clique then, since we also assume $\min\{|C|,|D|\} \geq 2$, there exists either an induced $C_4$ or an induced diamond with two vertices in $C$ and two vertices in $D$.
As a result, $\delta(G) \geq 1/2$.
\end{itemize}
\end{itemize}
\end{proof}

\begin{theorem}\label{thm:sw-hyp}
For every $G=(V,E)$, {\sc Hyperbolicity} can be solved in ${\cal O}(sw(G)^3 \cdot n+m )$-time.
\end{theorem}

\begin{proof}
First we compute in linear-time the split components $C_1, C_2, \ldots, C_k$ of $G$.
By Lemma~\ref{lem:hyp-split}, we have $\delta(G) \geq \max_i\delta(C_i)$.
Furthermore, for every $1 \leq i \leq k$ we have: if $C_i$ induces a star or a complete graph, then $\delta(C_i) = 0$; otherwise, $|V(C_i)| \leq sw(G)$, and so, $\delta(C_i)$ can be computed in ${\cal O}(|V(C_i)|^4) = {\cal O}(sw(G)^3 \cdot |V(C_i)|)$-time, simply by iterating over all possible $4$-tuples.
Summarizing, we can compute  $\max_i\delta(C_i)$ in ${\cal O}(sw(G)^3 \cdot \sum_i |V(C_i)|) = {\cal O}(sw(G)^3 \cdot n)$-time.
By Lemma~\ref{lem:hyp-split} we have $\delta(G) \leq  \max \{1, \max_i\delta(C_i) \}$.
Therefore, if $\max_i\delta(C_i) \geq 1$ then we are done.
Otherwise, in order to compute $\delta(G)$, by Lemma~\ref{lem:hyp-split} it suffices to check whether the sides of every split used for the split decomposition induce a complete subgraph.
For that, we use a split decomposition tree $T$ of $G$.
Indeed, recall that the edges of $T$ are in bijective correspondance with the splits.

Let us root $T$ in $C_1$.
Notations are from the proof of Theorem~\ref{thm:sw-ecc}.
In particular, for every $1 \leq i \leq k$ let $T_i$ be the subtree of $T$ that is rooted in $C_i$.
If $i > 1$ then let $C_{p(i)}$ be its parent in $T$.
By construction of $T$, the edge $\{C_{p(i)}, C_i\} \in E(T)$ corresponds to a split $(A_i,B_i)$ of $G$, where $V(C_i) \subseteq A_i$.
Let $G_{A_i} = (A_i \cup \{b_i\}, E_{A_i}), \ G_{B_i} = (B_i \cup \{a_i\}, E_{B_i})$ be the corresponding subgraphs of $G$. 
Vertex $a_i$ is simplicial in $G_{B_i}$ if and only if the side $N_G(A_i)$ is a clique.
Similarly, vertex $b_i$ is simplicial in $G_{A_i}$ if and only if the side $N_G(B_i)$ is a clique.
So, we perform tree traversals of $T$ in order to decide whether $a_i$ and $b_i$ are simplicial.

More precisely, we recall that $T_i$ and $T \setminus T_i$ are split decomposition trees of $G_{A_i}$ and $G_{B_i}$, respectively.
We now proceed in two main steps.

\begin{itemize}
\item 
First, we decide whether $b_i$ is simplicial in $G_{A_i}$ by dynamic programming.
More precisely, let $C_{i_1}, C_{i_2}, \ldots, C_{i_k}$ be the children of $C_i$ in $T$. (possibly, $k=0$ if $C_i$ is a leaf).
Then, $b_i$ is simplicial in $G_{A_i}$ if and only if: it is simplicial in $C_i$; and for every $1 \leq t \leq k$ such that $\{b_i,a_{i_t}\} \in E(C_i)$, we have that $b_{i_t}$ is simplicial in $G_{A_{i_t}}$.
In particular, testing whether $b_{i}$ is simplicial in $C_i$ takes time: ${\cal O}(1)$ if $C_i$ induces a star or a complete graph; and ${\cal O}(|V(C_i)|^2) = {\cal O}(sw(G) \cdot |V(C_i)|)$ otherwise.
Since a vertex can have at most $|V(C_i)| - 1$ neighbours in $C_i$, testing whether $b_i$ is simplicial in $G_{A_i}$ can be done in ${\cal O}(|V(C_i)|)$ additional time.
So, overall, the first step takes ${\cal O}(sw(G) \cdot \sum_i|V(C_i)|) = {\cal O}(sw(G) \cdot n)$-time.

\item 
Second, we decide whether $a_i$ is simplicial in $G_{B_i}$ by reverse dynamic programming.
Let $C_{j_0} = C_i, C_{j_1}, C_{j_2}, \ldots, C_{j_k}$ denote the children of $C_{p(i)}$ in $T$.
Furthermore, if $p(i) \neq 1$ then let $C_{p^2(i)}$ be the parent of $C_{p(i)}$ in $T$.
Then, $a_i$ is simplicial in $G_{B_i}$ if and only if: it is simplicial in $C_{p(i)}$; for every $1 \leq t \leq k$ such that $\{a_i,a_{j_t}\} \in E(C_{p(i)})$, we have that $b_{j_t}$ is simplicial in $G_{A_{j_t}}$; if $p(i) \neq 1$ and $\{a_i,b_{p(i)}\} \in E(C_{p(i)})$, we also have that $a_{p(i)}$ is simplicial in $G_{B_{p(i)}}$.
Testing, for every $0 \leq t \leq k$, whether $a_{j_t}$ is simplicial in $C_{p(i)}$ takes total time: ${\cal O}(|V(C_{p(i)})|)$ if $C_{p(i)}$ induces a star or a complete graph; and ${\cal O}(|V(C_{p(i)})|^3) = {\cal O}(sw(G)^2 \cdot |V(C_{p(i)})|)$ otherwise.

Then, for stars and prime components, we can test, for every $0 \leq t \leq k$, whether $a_{j_t}$ is simplicial in $G_{B_{j_t}}$ in total ${\cal O}(|E(C_{p(i)})|)$-time, that is ${\cal O}(|V(C_{p(i)})|)$ for stars and ${\cal O}(|V(C_{p(i)})|^2) = {\cal O}(sw(G) \cdot |V(C_{p(i)})|)$ for prime components.
For the case where $C_{p(i)}$ is a complete graph then, since all the vertices in $C_{p(i)}$ are pairwise adjacent, we only need to check whether there is at least one vertex $a_{j_t}$ such that $b_{j_t}$ is non simplicial in $G_{A_{j_t}}$, and also if $p(i) > 1$ whether $a_{p(i)}$ is non simplicial in $G_{B_{p(i)}}$.
It takes ${\cal O}(|V(C_{p(i)})|)$-time.

So, overall, the second step takes ${\cal O}(sw(G)^2 \cdot n)$-time.  
\end{itemize}
\end{proof}

\begin{corollary}[~\cite{Sot11}]
For every connected $G=(V,E)$ we have $\delta(G) \leq \max\{1, \left\lfloor (sw(G)-1) /2 \right\rfloor \}$.
\end{corollary}

\subsubsection*{Betweenness Centrality}

The following subsection can be seen as a broad generalization of the preprocessing method presented in~\cite{PEZB14}.
We start introducing a generalization of {\sc Betweenness Centrality} for {\em vertex-weighted} graphs.
Admittedly, the proposed generalization is somewhat technical.
However, it will make easier the dynamic programming of Theorem~\ref{thm:sw-bc}.

Precisely, let $G=(V,E,\alpha,\beta)$ with $\alpha,\beta : V \to \mathbb{N}$ be weight functions.
Intuitively, for a split marker vertex $v$, $\alpha(v)$ represents the side of the split replaced by $v$, while $\beta(v)$ represents the total number of vertices removed by the simple decomposition. 
For every path $P=(v_1,v_2,\ldots,v_{\ell})$ of $G$, the {\em length} of $P$ is equal to the number $\ell$ of edges in the path, while the {\em cost} of $P$ is equal to $\prod_{i=1}^{\ell} \alpha(v_i)$.
Furthermore, for every $s,t \in V$, the value $\sigma_G(s,t)$ is obtained by summing the cost over all the shortest $st$-paths in $G$.
Similarly, for every $s,t,v \in V$, the value $\sigma_G(s,t,v)$ is obtained by summing the cost over all the shortest $st$-paths in $G$ that contain $v$.
The betweenness centrality of vertex $v$ is defined as: $$\frac 1 {\alpha(v)} \sum_{s,t \in V \setminus v} \beta(s)\beta(t) \frac {\sigma_G(s,t,v)} {\sigma_G(s,t)}.$$
Note that if all weights are equal to $1$ then this is exactly the definition of Betweenness Centrality for unweighted graphs.

\begin{lemma}\label{lem:bc-split}
Let $(A,B)$ be a split of $G=(V,E,\alpha,\beta)$ and let $C = N_G(B) \subseteq A, \ D = N_G(A) \subseteq B$.
Furthermore, let $G_A = (A \cup \{b\}, E_A, \alpha_A, \beta_A), \ G_B = (B \cup \{a\}, E_B, \alpha_B, \beta_B)$ be the corresponding subgraphs of $G$, where:
$$\begin{cases}
\alpha_A(v) = \alpha(v), \ \beta_A(v) = \beta(v) \ \mbox{if} \ v \in A \\
\alpha_B(u) = \alpha(u), \ \beta_B(u) = \beta(u) \ \mbox{if} \ u \in B \\
\alpha_A(b) = \sum_{u \in D} \alpha(u), \ \beta_A(b) = \sum_{u \in B} \beta(u) \\
\alpha_B(a) = \sum_{v \in C} \alpha(v), \ \beta_B(a) = \sum_{v \in A} \beta(v).
\end{cases}$$
Then for every $v \in A$ we have:
$$BC_G(v) = BC_{G_A}(v) + [v \in C] BC_{G_B}(a).$$
\end{lemma}

\begin{proof}
Let $v \in A$ be fixed.
We consider all possible pairs $s,t \in V \setminus v$ such that $dist_G(s,t) = dist_G(s,v) + dist_G(v,t)$.

\smallskip
Suppose that $s,t \in A \setminus v$.
Since $(A,B)$ is a split, the shortest $st$-paths in $G$ are contained in $N_G[A] = A \cup D$.
In particular, the shortest $st$-paths in $G_A$ are obtained from the shortest $st$-paths in $G$ by replacing any vertex $d \in D$ by the split marker vertex $b$.
Conversely, the shortest $st$-paths in $G$ are obtained from the shortest $st$-paths in $G_A$ by replacing $b$ with any vertex $d \in D$.
Hence, $\sigma_{G_A}(s,t,b) = \sum_{d \in D} \sigma_G(s,t,d)$, that implies $\sigma_G(s,t) = \sigma_{G_A}(s,t)$.
Furthermore, $\sigma_G(s,t,v) = \sigma_{G}(s,v) \sigma_G(v,t) = \sigma_{G_A}(s,v)\sigma_{G_A}(v,t) = \sigma_{G_A}(s,t,v)$.
As a result, $\sigma_G(s,t,v)/\sigma_G(s,t) = \sigma_{G_A}(s,t,v)/\sigma_{G_A}(s,t)$.

\smallskip
Suppose that $s \in B, \ t \in A \setminus v$.
Every shortest $st$-path in $G$ is the concatenation of a shortest $sD$-path with a shortest $tC$-path.   
Therefore, $\sigma_G(s,t) = \frac {\sigma_{G_B}(s,a) \cdot \sigma_{G_A}(b,t)}{\alpha_B(a)\cdot \alpha_A(b)}$.
We can furthermore observe $v$ is on a shortest $st$-path in $G$ if, and only if, $v$ is on a shortest $bt$-path in $G_A$.
Then, $\sigma_G(s,t,v) = \sigma_{G}(s,v) \sigma_G(v,t) = \frac {\sigma_{G_B}(s,a) \cdot \sigma_{G_A}(b,v)}{\alpha_B(a)\cdot \alpha_A(b)} \sigma_{G_A}(v,t)$.
As a result, $\sigma_G(s,t,v)/\sigma_G(s,t) = \sigma_{G_A}(b,t,v)/\sigma_{G_A}(b,t)$.

\smallskip
Finally, suppose that $s,t \in B$.
Again, since $(A,B)$ is a split the shortest $st$-paths in $G$ are contained in $N_G[B] = B \cup C$.
In particular, $\sigma_G(s,t,v) \neq 0$ if, and only if, we have $v \in C$ and $\sigma_{G_B}(s,t,a) \neq 0$.
More generally, if $v \in C$ then $\sigma_G(s,t,v) = \frac {\alpha_A(v)} {\alpha_B(a)} \sigma_{G_B}(s,t,a)$.
As a result, if $v \in C$ then $\sigma_G(s,t,v)/\sigma_G(s,t) = \frac {\alpha_A(v)} {\alpha_B(a)} \cdot \sigma_{G_B}(s,t,a)/\sigma_{G_B}(s,t)$.

\medskip
Overall, we have:
\begin{align*}
BC_G(v) &= \frac 1 {\alpha(v)} \sum_{s,t \in V \setminus v} \beta(s)\beta(t) \frac {\sigma_G(s,t,v)} {\sigma_G(s,t)} \\
&= \frac 1 {\alpha(v)} \sum_{s,t \in A \setminus v} \beta(s)\beta(t) \frac {\sigma_G(s,t,v)} {\sigma_G(s,t)} +  \frac 1 {\alpha(v)} \sum_{s \in B, \ t \in A \setminus v} \beta(s)\beta(t) \frac {\sigma_G(s,t,v)} {\sigma_G(s,t)} + \frac 1 {\alpha(v)} \sum_{s,t \in B} \beta(s)\beta(t) \frac {\sigma_G(s,t,v)} {\sigma_G(s,t)}\\
&= \frac 1 {\alpha_A(v)} \sum_{s,t \in A \setminus v} \beta_A(s)\beta_A(t) \frac {\sigma_{G_A}(s,t,v)}{\sigma_{G_A}(s,t)} +  \frac 1 {\alpha_A(v)} \sum_{s \in B, \ t \in A \setminus v} \beta_B(s)\beta_A(t)\frac {\sigma_{G_A}(b,t,v)}{\sigma_{G_A}(b,t)} \\ &\ + \frac 1 {\alpha_A(v)} [v \in C] \sum_{s,t \in B} \beta_B(s)\beta_B(t)\frac {\alpha_A(v)} {\alpha_B(a)} \cdot \frac {\sigma_{G_B}(s,t,a)}{\sigma_{G_B}(s,t)}\\
&= \left( BC_{G_A}(v) - \frac {\beta_A(b)} {\alpha_A(v)} \sum_{t \in A \setminus v} \beta_A(t) \frac{\sigma_{G_A}(b,t,v)}{\sigma_{G_A}(b,t)} \right) + \frac {\sum_{s \in B} \beta(s)} {\alpha_A(v)} \sum_{t \in A \setminus v} \beta_A(t) \frac{\sigma_{G_A}(b,t,v)}{\sigma_{G_A}(b,t)} + [v \in C] BC_{G_B}(a) \\
&= BC_{G_A}(v) + [v \in C] BC_{G_B}(a),
\end{align*}
that finally proves the lemma.
\end{proof}

\begin{theorem}\label{thm:sw-bc}
For every $G=(V,E)$, {\sc Betweenness Centrality} can be solved in ${\cal O}(sw(G)^2 \cdot n + m)$-time.
\end{theorem}

\begin{proof}
Let $T$ be a split decomposition tree of $G$, with its nodes being in bijective correspondance with the split components $C_1, C_2, \ldots, C_k$.
It can be computed in linear-time~\cite{CDR12}.
As for Theorem~\ref{thm:sw-ecc}, we root $T$ in $C_1$.
For every $1 \leq i \leq k$, let $T_i$ be the subtree of $T$ that is rooted in $C_i$.
If $i > 1$ then let $C_{p(i)}$ be its parent in $T$.
We recall that by construction of $T$, the edge $\{C_{p(i)}, C_i\} \in E(T)$ corresponds to a split $(A_i,B_i)$ of $G$, where $V(C_i) \subseteq A_i$.
Furthermore, let $G_{A_i} = (A_i \cup \{b_i\}, E_{A_i}), \ G_{B_i} = (B_i \cup \{a_i\}, E_{B_i})$ be the corresponding subgraphs of $G$.
We observe that $T_i$ is a split decomposition tree of $G_{A_i}$, while $T \setminus T_i$ is a split decomposition tree of $G_{B_i}$.

Let us assume $G=(V,E,\alpha,\beta)$ to be vertex-weighted, with initially $\alpha(v) = \beta(v) = 1$ for every $v \in V$.
For every $i > 1$, let $G_{A_i} = (A_i \cup \{b_i\}, E_{A_i}, \alpha_{A_i}, \beta_{A_i}), \ G_{B_i} = (B_i \cup \{a_i\}, E_{B_i}, \alpha_{B_i}, \beta_{B_i})$ be as described in Lemma~\ref{lem:bc-split}. 
In particular, for every $i > 1$:
$$\begin{cases}
\alpha_{A_i}(v) = \alpha(v) = 1, \ \beta_{A_i}(v) = \beta(v) = 1 \ \mbox{if} \ v \in A_i \\
\alpha_{B_i}(u) = \alpha(u) = 1, \ \beta_{B_i}(u) = \beta(u) = 1 \ \mbox{if} \ u \in B_i \\
\alpha_{A_i}(b_i) = |N_G(A_i)|, \ \beta_{A_i}(b_i) = |B_i| \\
\alpha_{B_i}(a_i) = |N_G(B_i)|, \ \beta_{B_i}(a_i) = |A_i|.
\end{cases}$$
Hence, all the weights can be computed in linear-time by dynamic programming over $T$.
We set $G_1 = G$ while $G_i = G_{A_i}$ for every $i >1$.
Furthermore, we first aim at computing $BC_{G_i}(v)$ for every $v \in V(C_i)$.

If $C_i$ is a leaf of $T$ then there are three cases to be considered.

\begin{enumerate}

\item Suppose $G_i$ is a complete graph.
Then, for every $v \in V(C_i)$ we have $BC_{G_i}(v) = 0$.

\item Suppose $G_i$ is a star, with center node $r$.
In particular, $BC_{G_i}(v) = 0$ for every $v \in V(C_i) \setminus \{r\}$.
Furthermore, since $r$ is onto the unique shortest path between every two leaves $s,t \in V(C_i) \setminus \{r\}$, we have $\sigma_{G_i}(s,t,r) = \sigma_{G_i}(s,t)$.
Let us write $\beta(G_i) = \sum_{v \in V(C_i) \setminus \{r\}} \beta_{G_i}(v)$.
We have:
\begin{align*}
BC_{G_i}(r) &= \frac 1 {\alpha_{G_i}(r)} \sum_{s,t \in V(C_i) \setminus \{r\}} \beta_{G_i}(s)\beta_{G_i}(t) \\
&= \frac 1 {2\alpha_{G_i}(r)} \sum_{s \in V(C_i) \setminus \{r\}} \beta_{G_i}(s) \left(\sum_{t \in V(C_i) \setminus \{r,s\}} \beta_{G_i}(t)\right) \\
&= \frac 1 {2\alpha_{G_i}(r)} \sum_{s \in V(C_i) \setminus \{r\}} \beta_{G_i}(s) \left(\beta(G_i) - \beta_{G_i}(s)\right).
\end{align*}
It can be computed in ${\cal O}(|V(C_i)|)$-time.

\item Finally, suppose $G_i$ is prime for split decomposition.
Brandes algorithm~\cite{Bra01} can be generalized to that case.
For every $v \in V(C_i)$, we first compute a BFS ordering from $v$.
It takes ${\cal O}(|E(C_i)|)$-time.
Furthermore for every $u \in V(C_i) \setminus \{v\}$, let $N^+(u)$ be the neighbours $w \in N_{C_i}(u)$ such that $w$ is on a shortest $uv$-path.
We compute $\sigma_{G_i}(u,v)$ by dynamic programming.
Precisely, $\sigma_{G_i}(v,v) = \alpha_{G_i}(v)$, and for every $u \neq v, \ \sigma_{G_i}(u,v) = \alpha_{G_i}(u) \cdot \left(\sum_{w \in N^+(u)}\sigma_{G_i}(w,v)\right)$.
It takes ${\cal O}(|E(C_i)|)$-time. 

Overall in ${\cal O}(|V(C_i)||E(C_i)|)$-time, we have computed $\sigma_{G_i}(u,v)$ and $dist_{G_i}(u,v)$ for every $u,v \in V(C_i)$.
Then, for every $v \in V(C_i)$, we can compute $BC_{G_i}(v)$ in ${\cal O}(|V(C_i)|^2)$-time by enumerating all the pairs $s,t \in V(C_i) \setminus \{v\}$.
Since $G_i$ is prime, the total running time is in ${\cal O}(|V(C_i)|^3) = {\cal O}(sw(G)^3)$, and so, in ${\cal O}(sw(G)^2 \cdot |V(C_i)|)$.
\end{enumerate}

Otherwise, $C_i$ is an internal node of $T$.
Let $C_{i_1}, C_{i_2}, \ldots, C_{i_k}$ be the children of $C_i$ in $T$.
Assume that, for every $1 \leq t \leq k$, $BC_{G_{i_t}}(b_{i_t})$ has been computed (by dynamic programming over $T$).
Let us define the following weight functions for $C_i$:
$$\begin{cases}
\alpha_i(a_{i_t}) = \alpha_{B_{i_t}}(a_{i_t}), \ \beta_i(a_{i_t}) = \beta_{B_{i_t}}(a_{i_t}) \\
\alpha_i(v) = \alpha_{A_{i}}(v), \ \beta_i(v) = \beta_{A_{i}}(v) \ \mbox{otherwise.}
\end{cases}$$
Observe that every edge $\{C_i,C_{i_t}\}$ also corresponds to a split $(A_{i_t}',B_{i_t}')$ of $G_i$, where $V(C_{i_t}) \subseteq A_{i_t}' = A_{i_t}$.
By applying all the corresponding simple decompositions, one finally obtains $H_i = (V(C_i),E(C_i),\alpha_i,\beta_i)$.
Then, let us define $\ell_i(a_{i_t}) = BC_{G_{i_t}}(b_{i_t})$ and $\ell_i(v) = 0$ else.
Intuitively, the function $\ell_i$ is a corrective term updated after each simple decomposition.
More precisely, we obtain by multiple applications of Lemma~\ref{lem:bc-split}, for every $v \in V(C_i)$: $$BC_{G_i}(v) = BC_{H_i}(v) + \sum_{u \in N_{H_i}(v)} \ell_i(u) $$
Clearly, this can be reduced in ${\cal O}(|E(H_i)|)$-time, resp. in ${\cal O}(|V(H_i)|)$-time when $H_i$ is complete, to the computation of $BC_{H_i}(v)$.
So, it can be done in ${\cal O}(sw(G)^2 \cdot |V(C_i)|)$-time ({\it i.e.}, as explained for the case of leaf nodes).

Overall, this first part of the algorithm takes time ${\cal O}(sw(G)^2 \cdot \sum_i |V(C_i)|) = {\cal O}(sw(G)^2 \cdot n)$. 
Furthermore, since $G_1 = G$, we have computed $BC_G(v)$ for every $v \in V(C_1)$.
Then, using the same techniques as above, we can compute $BC_{G_{B_i}}(a_i)$ for every $i > 1$ by reverse dynamic programming over $T$.
It takes ${\cal O}(sw(G)^2 \cdot n)$-time. 
Finally, by Lemma~\ref{lem:bc-split} we can compute $BC_G(v)$ from $BC_{G_i}(v)$ and $BC_{G_{B_i}}(a_i)$, for every $v \in V(C_i)$.
It takes linear-time.
\end{proof}

\subsection{Kernelization methods with modular decomposition}\label{sec:algos-modular-dec}

The purpose of the subsection is to show how to apply the previous results, obtained with split decomposition, to modular decomposition.
On the way, improvements are obtained for the running time.
Indeed, it is often the case that only the quotient graph $G'$ needs to be considered.
We thus obtain algorithms that run in ${\cal O}(mw(G)^{{\cal O}(1)} + n + m)$-time.
See~\cite{MNN16} for an extended discussion on the use of {\em Kernelization} for graph problems in P.

We start with the following lemma:

\begin{lemma}[folklore]\label{lem:mw-to-sw}
For every $G=(V,E)$ we have $sw(G) \leq mw(G)+1$.
\end{lemma}

\begin{proof}
First we claim that $mw(H) \leq mw(G)$ for every {\em induced} subgraph $H$ of $G$.
Indeed, for every module $M$ of $G$ we have that $M \cap V(H)$ is a module of $H$, thereby proving the claim.
We show in what follows that a ``split decomposition'' can be computed from the modular decomposition of $G$ such that all the non degenerate split components have size at most $mw(G)+1$\footnote{Formally this is only a partial split decomposition, since there are subgraphs that could be further decomposed.}.
Applying this result to every prime split component of $G$ in its canonical split decomposition proves the lemma.

W.l.o.g., $G$ is connected (otherwise, we consider each connected component separately).
Let ${\cal M}(G) = \{ M_1,M_2, \ldots, M_k \}$ ordered by decreasing size.

\begin{enumerate}
\item If $|M_1| =1$ ($G$ is either complete or prime for modular decomposition) then we output $G$.

\item Otherwise, suppose $|M_1| < n-1$.
We consider all the maximal strong modules $M_1, M_2, \ldots, M_t$ such that $|M_i| \geq 2$ sequentially.
For every $1 \leq i \leq t$, we have that $(M_i, V \setminus M_i)$ is a split.
Furthermore if we apply the corresponding simple decomposition then we obtain two subgraphs, one being the subgraph $G_i$ obtained from $G[M_i]$ by adding a universal vertex $b_i$, and the other being obtained from $G$ by replacing $M_i$ by a unique vertex $a_i$ with neighbourhood $N_G(M_i)$.
Then, there are two subcases.
\begin{itemize}
\item Subcase ${\cal M}(G) = \{M_1,M_2\}$.
In particular, $|M_2| \geq 2$.
We perform a simple decomposition for $M_1$.
The two resulting subgraphs are exactly $G_1$ and $G_2$.
\item Subcase $\{M_1,M_2\} \subsetneq {\cal M}(G)$.
We apply simple decompositions for $M_1,M_2, \ldots, M_t$ sequentially.
Indeed, let $i \in \{1, \ldots, t\}$ and suppose we have already applied simple decompositions for $M_1, M_2, \ldots, M_{i-1}$.
Then, since there are at least three modules in ${\cal M}(G)$ we have that $(M_i, \{a_1,a_2, \ldots a_{i-1}\} \cup \bigcup_{j > i}M_j)$ remains a split, and so, we can apply a simple decomposition.
The resulting components are exactly: the quotient graph $G'$ and, for every $1 \leq i \leq t$, the subgraph $G_i$ obtained from $G[M_i]$.
\end{itemize}
Furthermore, in both subcases we claim that the modular decomposition of $G_i$ can be updated from the modular decomposition of $G[M_i]$ in constant-time.
Indeed, the set of all universal vertices in a graph is a clique and a maximal strong module. 
We output $G'$ (only if $\{M_1,M_2\} \subsetneq {\cal M}(G)$) and, for every $1 \leq i \leq t$, we apply the procedure recursively for $G_i$.
\item Finally, suppose $|M_1| = n-1$.
In particular, ${\cal M}(G) = \{ M_1, M_2 \}$ and $M_2$ is trivial.
Let ${\cal M}(G[M_1]) = \{ M_1', M_2', \ldots, M_p' \}$ ordered by decreasing size.
If $|M_1'| = 1$ ({\it i.e.}, $G[M_1]$ is either edgeless, complete or prime for modular decomposition) then we output $G$.
Otherwise we apply the previous Step 2 to the modular partition  $M_1', M_2', \ldots, M_p', M_2$.
\end{enumerate}
The procedure takes linear-time if the modular decomposition of $G$ is given.
Furthermore, the subgraphs obtained are either: the quotient graph $G'$; a prime subgraph for modular decomposition with an additional universal vertex; or a degenerate graph (that is obtained from either a complete subgraph or an edgeless subgraph by adding a universal vertex).
\end{proof}

\begin{corollary}
For every $G=(V,E)$ we can solve:
\begin{itemize}
\item {\sc Eccentricities} and {\sc Diameter} in ${\cal O}(mw(G)^2 \cdot n + m)$-time;
\item {\sc Hyperbolicity} in ${\cal O}(mw(G)^3 \cdot n + m)$-time;
\item {\sc Betweenness Centrality} in ${\cal O}(mw(G)^2 \cdot n + m)$-time.
\end{itemize}
\end{corollary}

In what follows, we explain how to improve the above running times in some cases.

\begin{theorem}\label{thm:mw-ecc}
For every $G=(V,E)$, {\sc Eccentricities} can be solved in ${\cal O}(mw(G)^3 + n+m)$-time.

In particular, {\sc Diameter} can be solved in ${\cal O}(mw(G)^3 + n+m)$-time.
\end{theorem}

\begin{proof}
W.l.o.g., $G$ is connected.
Consider the (partial) split decomposition obtained from the modular decomposition of $G$ (Lemma~\ref{lem:mw-to-sw}).
Let $T$ be the corresponding split decomposition tree.
By construction, there exists a modular partition $M_1,M_2, \ldots, M_k$ of $G$ with the two following properties:
\begin{itemize}
\item All but at most one split components of $G$ are split components of some $G_i, \ 1 \leq i \leq k$, where the graph $G_i$ is obtained from $G[M_i]$ by adding a universal vertex $b_i$. 
\item Furthermore, the only remaining split component (if any) is the graph $G'$ obtained by replacing every module $M_i$ with a single vertex $a_i$.
Either $G'$ is degenerate (and so, $diam(G') \leq 2$) or $k \leq mw(G)+1$.
We can also observe in this situation that if we root $T$ in $G'$ then the subtrees of $T \setminus \{G'\}$ are split decomposition trees of the graphs $G_i, \ 1 \leq i \leq k$. 
\end{itemize} 
We can solve {\sc Eccentricities} for $G$ as follows.
First for every $1 \leq i \leq k$ we solve {\sc Eccentricities} for $G_i$.
In particular, $diam(G_i) \leq 2$, and so, for every $v \in V(G_i)$ we have: $ecc_{G_i}(v) = 0$ if and only if $V(G_i) = \{v\}$; $ecc_{G_i}(v) = 1$ if and only if $v$ is universal in $G_i$; otherwise, $ecc_{G_i}(v) = 2$.
Therefore, we can solve {\sc Eccentricities} for $G_i$ in ${\cal O}(|V(G_i)|+|E(G_i)|)$-time.
Overall, this step takes ${\cal O}(\sum_{i=1}^k |V_i| + |E_i|) = {\cal O}(n+m)$-time.
Then there are two subcases.

Suppose $G'$ is not a split component.
We deduce from Lemma~\ref{lem:mw-to-sw} $G=G[M_1] \oplus G[M_2]$.
In this situation, for every $i \in \{1,2\}$, for every $v \in V(G_i)$ we have $ecc_{G}(v) = \max\{ ecc_{G_i}(v), 1 \}$.

Otherwise, let us compute {\sc Eccentrities} for $G'$.
It takes ${\cal O}(|V(G')|) = {\cal O}(n)$-time if $G'$ is degenerate, and ${\cal O}(mw(G)^3)$-time otherwise.
Applying the algorithmic scheme of Theorem~\ref{thm:sw-ecc}, one obtains $ecc_G(v) = \max\{ecc_{G_i}(v), dist_{G_i}(v,a_i) + ecc_{G'}(b_i) - 1\} = \max \{ecc_{G_i}(v), ecc_{G'}(b_i) \}$ for every $v \in M_i$.
Hence, we can compute $ecc_G(v)$ for every $v \in V$ in ${\cal O}(n)$-time. 
\end{proof}

\begin{corollary}
For every connected $G=(V,E)$, $diam(G) \leq \max \{ mw(G), 2 \}$.
\end{corollary}

Next, we consider {\sc Hyperbolicity}.
It is proved in~\cite{Sot11} that, for every $G=(V,E)$ with quotient graph $G'$, $\delta(G') \leq \delta(G) \leq \max\{\delta(G'),1\}$.
The latter immediately implies the following result:

\begin{theorem}\label{thm:mw-hyp}
For every $G=(V,E)$, we can decide whether $\delta(G) >1$, and if so, compute $\delta(G)$, in ${\cal O}(mw(G)^4 + n + m)$-time.
\end{theorem}

However, we did not find a way to preprocess $G$ in linear-time so that we can compute $\delta(G)$ from $\delta(G')$.
Indeed, let $G_M$ be a graph of diameter at most $2$.
Solving {\sc Eccentricities} for $G_M$ can be easily done in linear-time.
However, the following shows that it is not that simple to do so for {\sc Hyperbolicity}.

\begin{lemma}[~\cite{CoDu14}]\label{lem:hyp-diam1}
For every $G=(V,E)$ we have $\delta(G) \leq \left\lfloor diam(G) /2 \right\rfloor$.
Furthermore, if $diam(G) \leq 2$ then $\delta(G) < 1$ if and only if $G$ is $C_4$-free.
\end{lemma}

The detection of an induced $C_4$ in ${\cal O}(mw(G)^{{\cal O}(1)} + n +m)$-time remains an open problem.

\subsubsection*{Short digression: using neighbourhood diversity}

We show that by imposing more constraints on the modular partition, some more kernels can be computed for the problems in Section~\ref{sec:dist-pbs}.
Two vertices $u,v$ are {\em twins} in $G$ if $N_G(u) \setminus v = N_G(v) \setminus u$.
Being twins induce an equivalence relationship over $V(G)$.
The number of equivalence classes is called the {\em neighbourhood diversity} of $G$, sometimes denoted by $nd(G)$~\cite{Lam12}.
Observe that every set of pairwise twins is a module of $G$.
Hence, $mw(G) \leq nd(G)$.

\begin{theorem}\label{thm:nd-hyp}
For every $G=(V,E)$, {\sc Hyperbolicity} can be solved in ${\cal O}(nd(G)^4 + n+m)$-time.
\end{theorem}

\begin{proof}
Let $V_1,V_2,\ldots,V_k, \ k = nd(G)$, partition the vertex-set $V$ in twin classes.
The partition can be computed in linear-time~\cite{Lam12}.
Furthermore, since it is a modular partition, we can compute a (partial) split decomposition as described in Lemma~\ref{lem:mw-to-sw}.
Let $G'=(V',E')$ such that $V' = \{v_1,v_2,\ldots,v_k\}$ and $E' = \{ \{v_i,v_j\} \mid V_i \times V_j \subseteq E \}$.
Then, the split components are either: $G'$, stars $S^i$ (if the vertices of $V_i$ are pairwise nonadjacent, {\it i.e.}, false twins) or complete graphs $K^i$ (if the vertices of $V_i$ are pairwise adjacent, {\it i.e.}, true twins). 

Applying the algorithmic scheme of Theorem~\ref{thm:sw-hyp}, in order to solve {\sc Hyperbolicity} for $G$ it suffices to compute, for every split component $C_j$, the hyperbolicity value $\delta(C_j)$ {\em and} all the simplicial vertices in $C_j$.
This can be done in ${\cal O}(|V(C_j)|)$-time if $C_j$ is a star or a complete graph, and in ${\cal O}(nd(G)^4)$-time if $C_j = G'$.
Therefore, we can solve {\sc Hyperbolicity} for $G$ in total ${\cal O}(nd(G)^4 + n+m)$-time.
\end{proof}

In~\cite{FKMN+17}, the authors propose an ${\cal O}(2^{{\cal O}(k)} + n +m)$-time algorithm for computing {\sc Hyperbolicity} with $k$ being the vertex-cover of the graph.
Their algorithm is pretty similar to Theorem~\ref{thm:nd-hyp}.
This is no coincidence since every graph with vertex-cover at most $k$ has neighbourhood diversity at most $2^{{\cal O}(k)}$~\cite{Lam12}.

Finally, the following was proved implicitly in~\cite{PEZB14}.

\begin{theorem}[~\cite{PEZB14}]\label{thm:nd-bc}
For every $G=(V,E)$, {\sc Betweenness Centrality} can be solved in ${\cal O}(nd(G)^3 + n+m)$-time.
\end{theorem}

\subsection{Applications to graphs with few $P_4$'s}
\label{sec:dist-qq3}

Before ending Section~\ref{sec:dist}, we apply the results of the previous subsections to the case of $(q,q-3)$-graphs.  
For that we need to consider all the cases where the quotient graph has super-constant size $\Omega(q)$ (see Lemma~\ref{lem:reduce-qq3}). 

\subsubsection*{Eccentricities}

\begin{theorem}\label{thm:qq3-ecc}
For every $G=(V,E)$, {\sc Eccentricities} can be solved in ${\cal O}(q(G)^3 + n+m)$-time.
\end{theorem}

\begin{proof}
By Lemma~\ref{lem:mw-to-sw}, there exists a partial split decomposition of $G$ such that the only split component with diameter possibly larger than $2$ is its quotient graph $G'$.
Furthermore, as shown in the proof of Theorem~\ref{thm:mw-ecc}, solving {\sc Eccentricities} for $G$ can be reduced in ${\cal O}(n+m)$-time to the solving of {\sc Eccentricities} for $G'$.
By Lemma~\ref{lem:reduce-qq3} we only need to consider the following cases. 
We can check in which case we are in linear-time~\cite{Bab00}. 

\begin{itemize}
\item Suppose $G'=(S'\cup K'\cup R',E')$ is a prime spider.
There are two subcases.
\begin{enumerate}
\item
If $G'$ is a thick spider then it has diameter two.
Since in addition, there is no universal vertex in $G'$, therefore every vertex of $G'$ has eccentricity exactly two.
\item
Otherwise, $G'$ is a thin spider.
Since there is no universal vertex in $G'$, every vertex has eccentricity at least two.
In particular, since $K'$ is a clique dominating set of $G'$, $ecc_{G'}(v) = 2$ for every $v \in K'$.
Furthermore, since there is a join between $K'$ and $R'$, $ecc_{G'}(v) = 2$ for any $v \in R'$.
Finally, since every two vertices of $S'$ are pairwise at distance three, $ecc_{G'}(v) = 3$ for every $v \in S$.
\end{enumerate}
\item Suppose $G'$ is isomorphic either to a cycle $C_{n'}$, or to a co-cycle $\overline{C_{n'}}$, for some $n' \geq 5$.
\begin{enumerate}
\item
If $G'$ is isomorphic to a cycle $C_{n'}$ then every vertex of $G'$ has eccentricity $\left\lfloor n' / 2 \right\rfloor$.
\item
Otherwise, $G'$ is isomorphic to a co-cycle $\overline{C_{n'}}$.
We claim that every vertex of $G'$ has eccentricity $2$.
Indeed, let $v \in \overline{C_{n'}}$ be arbitrary and let $u,w \in \overline{C_{n'}}$ be the only two vertices nonadjacent to $v$.
Furthermore, let $u',w'$ be the unique vertices of $\overline{C_{n'}} \setminus v$ that are respectively nonadjacent to $u$ and to $w$.
Since $n' \geq 5$, we have $u' \neq w'$.
In particular, $(v,u',w)$ and $(v,w',u)$ are, respectively, a shortest $vw$-path and a shortest $vu$-path.
Hence, $ecc_{G'}(v) = 2$.
\end{enumerate}

\item Suppose $G'$ is a spiked $p$-chain $P_k$, or its complement.
\begin{itemize}
\item Subcase $G'$ is a spiked $p$-chain $P_k$.
In particular, $G'$ contains the $k$-node path $P_k = (v_1,v_2,\ldots,v_k)$ as an isometric subgraph.
Furthermore, if $x \in V(G')$ then $dist_{G'}(v_1,x) = 2$, and $x$ and $v_2$ are twins in $G' \setminus v_1$.
Similarly, if $y \in V(G')$ then $dist_{G'}(v_k,y) = 2$, and $y$ and $v_{k-1}$ are twins in $G' \setminus v_k$.
As a result: for every $1 \leq i \leq k, \ ecc_{G'}(v_i) = ecc_{P_k}(v_i) = \max\{i-1,k-i\}$; if $x \in V(G')$ then $ecc_{G'}(x) = ecc_{P_k}(v_2) = k-2$; if $y \in V(G')$ then $ecc_{G'}(y) = ecc_{P_k}(v_{k-1}) = k-2$.
\item Subcase $G'$ is a spiked $p$-chain $\overline{P_k}$.
In particular, $\overline{G'}$ is a spiked $p$-chain $P_k$.
Since $k \geq 6$, every spiked $p$-chain $P_k$ has diameter more than four.
Hence, $diam(G') \leq 2$, that implies {\sc Eccentricities} can be solved for $G'$ in linear-time.
\end{itemize}
\item Suppose $G'$ is a spiked $p$-chain $Q_k$, or its complement.
\begin{itemize}
\item Subcase $G'$ is a spiked $p$-chain $Q_k$.
There is a clique-dominating set $K' = \{v_2,v_4,\ldots,v_{2j},\ldots\}$ of $G'$.
In particular, every vertex of $K'$ has eccentricity two.
Furthermore, any $z_i$ is adjacent to both $v_2,v_4$.
Every vertex $v_{2i-1}$, except $v_3$, is adjacent to $v_2$.
Finally, every vertex $v_{2i-1}$, except $v_1$ and $v_5$, is adjacent to $v_4$.
As a result, any vertex $z_i$ has eccentricity two; any vertex $v_{2i-1}, \ i \notin \{1,2,3\}$, also has eccentricity two.
However, since $v_2$ and $v_4$ are, respectively, the only neighbours of $v_1$ and $v_3$, we get $dist_{G'}(v_1,v_3) = dist_{G'}(v_3,v_5) = 3$.
Hence $ecc_{G'}(v_1) = ecc_{G'}(v_3) = ecc_{G'}(v_5) = 3$.
\item Subcase $G'$ is a spiked $p$-chain $\overline{Q_k}$.
Roughly, we reverse the roles of vertices $v_{2i}$ with even index with the roles of vertices $v_{2i-1}$ with odd index.
More precisely, there is a clique-dominating set $K' = \{v_1,v_3,\ldots,v_{2j-1},\ldots\}$ of $G'$.
In particular, every vertex of $K'$ has eccentricity two.
Furthermore, any $z_i$ is adjacent to both $v_1,v_3$.
Every vertex $v_{2i}$, except $v_2$, is adjacent to $v_1$.
Finally, every vertex $v_{2i}$, except $v_4$, is adjacent to $v_3$.
As a result, any vertex $z_i$ has eccentricity two; any vertex $v_{2i}, \ i \notin \{1,2\}$, also has eccentricity two.
However, since $v_3$ is the only neighbour of $v_2$, we get $dist_{G'}(v_2,v_4) = 3$, hence $ecc_{G'}(v_2) = ecc_{G'}(v_4) = 3$.
\end{itemize}
\item Otherwise, $|V(G')| \leq q(G)$.
Then, solving {\sc Eccentricities} for $G'$ can be done in ${\cal O}(q(G)^3)$-time.
\end{itemize}
Therefore, in all the above cases, {\sc Eccentricities} can be solved for $G'$ in ${\cal O}(\min\{q(G)^3,n+m\})$-time.
\end{proof}

\begin{corollary}\label{cor:qq4-diam}
For every connected $(q,q-4)$-graph $G=(V,E)$, $diam(G) \leq q$.
\end{corollary}

Corollary~\ref{cor:qq4-diam} does not hold for $(q,q-3)$-graphs because of cycles and spiked $p$-chains $P_k$.

\subsubsection*{Gromov hyperbolicity}

\begin{theorem}\label{thm:qq4-hyp}
For every $G=(V,E)$, {\sc Hyperbolicity} can be solved in ${\cal O}(q(G)^3 \cdot n+m)$-time.
\end{theorem}

\begin{proof}
By Lemma~\ref{lem:mw-to-sw}, we can compute a partial split decomposition from the modular decomposition of $G$.
It takes ${\cal O}(n+m)$-time.
Let $C_1, C_2, \ldots, C_k$ be the split components.
By using the algorithmic scheme of Theorem~\ref{thm:sw-hyp}, solving {\sc Hyperbolicity} can be reduced in ${\cal O}(\sum_i |V(C_i)|+|E(C_i)|)$-time to the computation, for every $1 \leq i \leq k$, of the hyperbolicity value $\delta(C_i)$ {\em and} of all the simplicial vertices in $C_i$.
We claim that it can be done in ${\cal O}(q(G)^3 \cdot |V(C_i)| + |E(C_i)|)$-time.
Since $\sum_i |V(C_i)| = {\cal O}(n)$ and $\sum_i |E(C_i)| = {\cal O}(n+m)$~\cite{Rao08b}, the latter claim will prove the desired time complexity.

If $C_i$ is degenerate then the above can be done in ${\cal O}(|V(C_i)|)$-time.
Otherwise, $C_i$ is obtained from a prime subgraph $G'$ in the modular decomposition of $G$ by possibly adding a universal vertex.
In particular, we have: $\delta(G') = \delta(C_i)$ if $G'=C_i$; $\delta(C_i) = 0$ can be decided in ${\cal O}(|V(C_i)|+|E(C_i)|)$-time~\cite{How79}; otherwise, $diam(C_i) \leq 2$, and so, by Lemma~\ref{lem:hyp-diam1} we have $\delta(C_i) = 1$ if and only if $C_i$ contains an induced cycle of length four (otherwise, $\delta(C_i) = 1/2$). 
Therefore, we are left to compute the following for every prime subgraph $G'$ in the modular decomposition of $G$:
\begin{itemize}
\item Compute $\delta(G')$;
\item Decide whether $G'$ contains an induced cycle of length four;
\item Compute the simplicial vertices in $G'$.
\end{itemize}
In particular, if $|V(G')| \leq q(G)$ then it can be done in ${\cal O}(|V(G')|^4)={\cal O}(q(G)^3 \cdot |V(G')|)$-time.
Otherwise, by Lemma~\ref{lem:reduce-qq3} we only need to consider the following cases. 
We can check in which case we are in linear-time~\cite{Bab00}.
\begin{itemize}
\item Suppose $G'$ is a prime spider.
In particular it is a split graph, and so, it does not contain an induced cycle of length more than three.
Furthermore the simplicial vertices of any chordal graph, and so, of $G'$, can be computed in linear-time.
If $G'$ is a thin spider then it is a block-graph, and so, $\delta(G') = 0$~\cite{How79}.
Otherwise, $G'$ is a thick spider, and so, it contains an induced diamond.
The latter implies $\delta(G') \geq 1/2$.
Since $diam(G') \leq 2$ and $G'$ is $C_4$-free, by Lemma~\ref{lem:hyp-diam1} $\delta(G') < 1$, hence we have $\delta(G') = 1/2$.
\item Suppose $G'$ is a cycle or a co-cycle of order at least five.
Since cycles and co-cycles are non complete regular graphs they do not contain any simplicial vertex~\cite{AiF84}.
Furthermore, a cycle of length at least five of course does not contain an induced cycle of length four; a co-cycle of order five is a $C_5$, and a co-cycle or order at least six always contains an induced cycle of length at least four since there is an induced $2K_2 = \overline{C_4}$ in its complement.
Finally, the hyperbolicity of a given cycle can be computed in linear-time~\cite{CCDL17+}; for every co-cycle of order at least six, since it has diameter at most two and it contains an induced cycle of length four, by Lemma~\ref{lem:hyp-diam1} it has hyperbolicity equal to $1$.
\item Suppose $G'$ is a spiked $p$-chain $P_k$, or its complement.
In particular, if $G'$ is a spiked $p$-chain $P_k$ then it is a block-graph, and so, a chordal graph.
It implies $\delta(G') = 0$~\cite{How79}, $G'$ does not contain any induced cycle of length four, furthermore all the simplicial vertices of $G'$ can be computed in linear-time.
Else, $G'$ is a spiked $p$-chain $\overline{P_k}$.
Since, in $\overline{G'}$, every vertex is nonadjacent to at least one edge of $P_k$, it implies that $G'$ has no simplicial vertex.
Furthermore, since $P_k$, and so, $\overline{G'}$, contains an induced $2K_2$, the graph $G'$ contains an induced cycle of length four.
Since $diam(G')=2$, it implies by Lemma~\ref{lem:hyp-diam1} $\delta(G') = 1$.
\item Otherwise, $G'$ is a spiked $p$-chain $Q_k$, or its complement.
In both cases, $G'$ is a split graph, and so, a chordal graph.
It implies that $G'$ does not contain an induced cycle of length four, and that all the simplicial vertices of $G'$ can be computed in linear-time.
Furthermore, we can decide in linear-time whether $\delta(G') = 0$~\cite{How79}.
Otherwise, it directly follows from the characterization in~\cite{BKM01} that a necessary condition for a chordal graph to have hyperbolicity at least one is to contain two disjoint pairs of vertices at distance $3$.
Since there are no such pairs in $G'$, $\delta(G') = 1/2$.
\end{itemize} 
\end{proof}

The solving of {\sc Betweenness Centrality} for $(q,q-3)$-graphs is left for future work.
We think it is doable with the techniques of Theorem~\ref{thm:sw-bc}.
However, this would require to find ad-hoc methods for every graph family in Lemma~\ref{lem:reduce-qq3}.
The main difficulty is that we need to consider weighted variants of these graph families, and the possibility to add a universal vertex.

\section{New Parameterized algorithms for {\sc Maximum Matching}}\label{sec:maxmatching}

A matching in a graph is a set of edges with pairwise disjoint end vertices.
We consider the problem of computing a matching of maximum size.

\begin{center}
	\fbox{
		\begin{minipage}{.95\linewidth}
			\begin{problem}[\textsc{Maximum Matching}]\
				\label{prob:matching} 
					\begin{description}
					\item[Input:] A graph $G=(V,E)$.
					\item[Output:] A matching of $G$ with maximum cardinality.
				\end{description}
			\end{problem}     
		\end{minipage}
	}
\end{center}

{\sc Maximum Matching} can be solved in polynomial time with Edmond's algorithm~\cite{Edm65}. A naive implementation of the algorithm runs in ${\cal O}(n^4)$ time. Nevertheless, Micali and Vazirani~\cite{MiV80} show how to implement Edmond's algorithm in time ${\cal O}(m\sqrt{n})$.
In~\cite{MNN16}, Mertzios, Nichterlein and Niedermeier design some new algorithms to solve {\sc Maximum Matching}, that run in ${\cal O}(k^{{\cal O}(1)} \cdot (n+m))$-time for various graph parameters $k$.
They also suggest to use {\sc Maximum Matching} as the ``drosophilia'' of the study of fully polynomial parameterized algorithms.

In this section, we present ${\cal O}(k^4 \cdot n + m)$-time algorithms for solving {\sc Maximum Matching}, when parameterized by either the modular-width or the $P_4$-sparseness of the graph.
The latter subsumes many algorithms that have been obtained for specific subclasses~\cite{FPT97,YuY93}. 

\subsection{Computing short augmenting paths using modular decomposition}
\label{sec:matching-mw}

Let $G=(V,E)$ be a graph and $F \subseteq E$ be a matching of $G$.
A vertex is termed matched if it is incident to an edge of $F$, and unmatched otherwise.
An $F$-augmenting path is a path where the two ends are unmatched, all edges $\{x_{2i},x_{2i+1}\}$ are in $F$ and all edges $\{x_{2j-1}, x_{2j}\}$ are not in $F$.
We can observe that, given an $F$-augmenting path $P = (x_1,x_2, \ldots, x_{2k})$, the matching $E(P)\Delta F$ (obtained by replacing the edges $\{x_{2i},x_{2i+1}\}$ with the edges $\{x_{2j-1}, x_{2j}\}$) has larger size than $F$.

\begin{theorem}[Berge,~\cite{Ber57}]\label{thm:berge}
A matching $F$ in $G=(V,E)$ is maximum if and only if there is no $F$-augmenting path.
\end{theorem}

We now sketch our approach.
Suppose that, for every module $M_i \in {\cal M}(G)$, a maximum matching $F_i$ of $G[M_i]$ has been computed.
Then, $F = \bigcup_i F_i$ is a matching of $G$, but it is not necessarily maximum.
Our approach consists in computing short augmenting paths (of length ${\cal O}(mw(G))$) using the quotient graph $G'$, until we obtain a maximum matching.
For that, we need to introduce several reduction rules.

\smallskip
The first rule (proved below) consists in removing, from every module $M_i$, the edges that are not part of its maximum matching $F_i$.

\begin{lemma}\label{lem:mw-matching-reduction}
Let $M$ be a module of $G=(V,E)$, let $G[M] = (M,E_M)$ and let $F_M  \subseteq E_M$ be a maximum matching of $G[M]$.
Then, every maximum matching of $G_M' = (V, (E \setminus E_M) \cup F_M)$ is a maximum matching of $G$.
\end{lemma}

\begin{proof}
Let us consider an arbitrary maximum matching of $G$.
We totally order $M = \{v_1,v_2, \ldots, v_l\}$, in such a way that unmatched vertices appear first, and for every edge in the matching $F_M$ the two ends of it are consecutive.
Let $S \subseteq M, \ |S| = k$ be the vertices of $M$ that are matched with a vertex of $V \setminus M$.
We observe that with the remaining $|M| - k$ vertices of $M \setminus S$, we can only obtain a matching of size at most $\mu_M = \min \{|F_M|, \left\lfloor (|M| - k)/2 \right\rfloor\}$.
Conversely, if $S=\{v_1,v_2,\ldots,v_k\}$ then we can always create a matching of size exactly $\mu_M$ with the vertices of $M \setminus S$ and the edges of $F_M$.
Since $M$ is a module of $G$, this choice can always be made without any loss of generality.
\end{proof}

From now on we shall assume each module induces a matching.
In particular, for every $M \in {\cal M}(G)$, the set $V(E(G[M]))$ stands for the non isolated vertices in the subgraph $G[M]$.

Then, we need to upper-bound the number of edges in an augmenting path that are incident to a same module.

\begin{lemma}\label{lem:bound-edges-augmenting-path}
Let $G=(V,E)$ be a graph such that every module $M \in {\cal M}(G)$ induces a matching.
Furthermore let $G'=({\cal M}(G),E')$ be the quotient graph of $G$, and let $F \subseteq E$ be a non maximum matching of $G$.
There exists an $F$-augmenting path $P=(x_1,x_2,\ldots,x_{2\ell})$ such that the following hold for every $M \in {\cal M}(G)$:
\begin{itemize}
\item $|\{ i \mid x_{2i-1}, x_{2i} \in M\}| \leq 1$;

furthermore if $|\{ i \mid x_{2i-1}, x_{2i} \in M\}| = 1$ then, for every $M' \in N_{G'}(M)$ we have $\{ i \mid x_{2i-1}, x_{2i} \in M'\} = \emptyset$;

\item $|\{ i \mid x_{2i}, x_{2i+1} \in M\}| \leq 1$;
\item $|\{ i \mid x_{2i-1} \notin M, \ x_{2i} \in M\}| \leq 1$;
\item $|\{ i \mid x_{2i} \notin M, \ x_{2i+1} \in M\}| \leq 2$;

furthermore if $|\{ i \mid x_{2i} \notin M, \ x_{2i+1} \in M\}| = 2$ then there exist $x_{2i_0+1},x_{2i_0+3},x_{2i_0+4} \in M$;

\item $|\{ i \mid x_{2i-1} \in M, \ x_{2i} \notin M\}| \leq 1$;
\item $|\{ i \mid x_{2i} \in M, \ x_{2i+1} \notin M\}| \leq 2$;

furthermore if $|\{ i \mid x_{2i} \in M, \ x_{2i+1} \notin M\}| = 2$ then there exist $x_{2i_0-1},x_{2i_0},x_{2i_0+2} \in M$.

\end{itemize}
In particular, $P$ has length ${\cal O}(|{\cal M}(G)|)$.
\end{lemma}

\begin{proof}
Let $P$ be a {\em shortest} $F$-augmenting path that minimizes $i(P) = \left|E(P) \cap \left(\bigcup_{M \in {\cal M}(G)}E(G[M])\right)\right|$.
Equivalently, $P$ is a shortest augmenting path with the minimum number of edges $i(P)$ with their two ends in a same module.
There are four cases.

\begin{enumerate}

\item
Suppose by contradiction there exist $i_1 < i_2$ such that $x_{2i_1-1}, x_{2i_1}, x_{2i_2-1}, x_{2i_2} \in M$.
See Fig.~\ref{fig-matching-1}-\ref{fig-matching-2}.
In particular, $i_2 - i_1 \geq 2$ since $M$ induces a matching.
Furthermore, $x_{2i_1+1}, x_{2i_2-2} \in N_G(M)$.
Then, $(x_1,\ldots, x_{2i_1-1}, x_{2i_1+1}, x_{2i_1}, x_{2i_2-2}, x_{2i_2-1}, x_{2i_2}, \ldots x_{2\ell})$ is an $F$-augmenting path, thereby contradicting the minimality of $i(P)$.


\begin{figure}[h!]
\hfill\begin{minipage}{.48\textwidth}
  \centering
  \includegraphics[width=.45\textwidth]{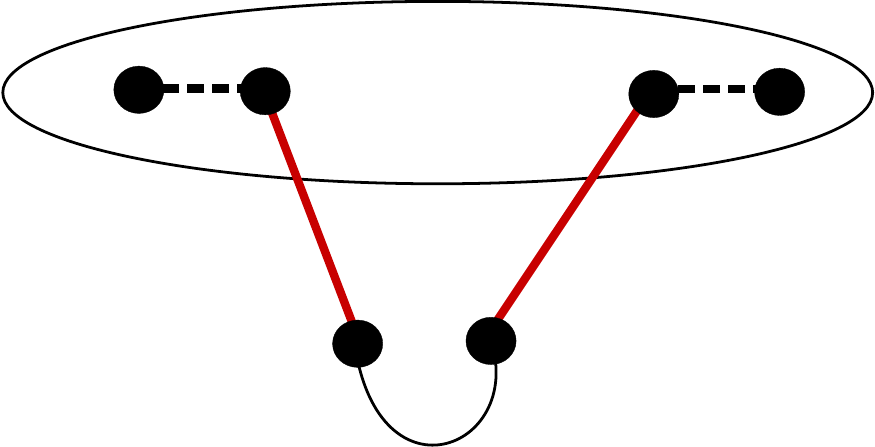}
  \caption{Case $x_{2i_1-1}, x_{2i_1}, x_{2i_2-1}, x_{2i_2} \in M$.}
  \label{fig-matching-1}
\end{minipage}\hfil
\begin{minipage}{.48\textwidth}
  \centering
  \includegraphics[width=.45\textwidth]{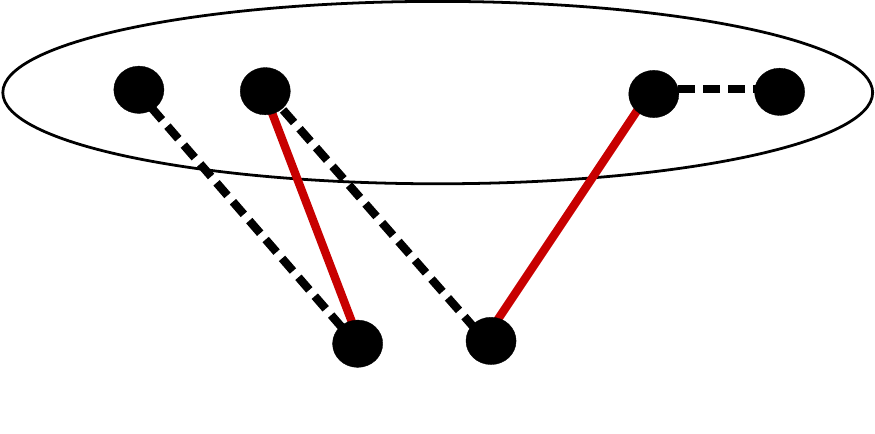}
  \caption{Local replacement of $P$.}
  \label{fig-matching-2}
\end{minipage}
\end{figure}

Similarly, suppose by contradiction there exist $x_{2i_1-1},x_{2i_1} \in M$ and there exist $x_{2i_2-1}, x_{2i_2} \in M', \ M' \in N_{G'}(M)$.
See Fig~\ref{fig-matching-3}.
We assume by symmetry $i_1 < i_2$.
In this situation, either $i_1 = 1$, and so, $x_{2i_1-1} = x_1$ is unmatched, or $i_1 > 1$ and so, $x_{2i_1-1}$ is matched to $x_{2i_1-2} \neq x_{2i_2}$.
Then, $(x_1,\ldots,x_{2i_1-1},x_{2i_2},\ldots,x_{2\ell})$ is an $F$-augmenting path, thereby contradicting the minimality of $|V(P)|$.

\begin{figure}[h!]
\centering
\includegraphics[width=.25\textwidth]{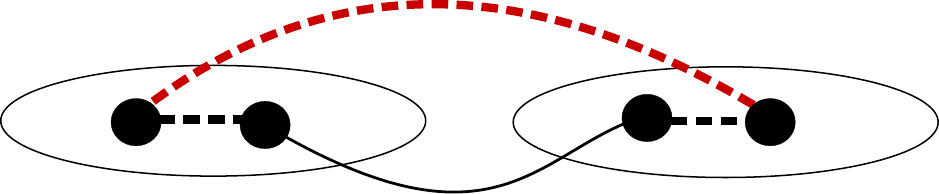}
\caption{Case $x_{2i_1-1},x_{2i_1} \in M$ and $x_{2i_2-1}, x_{2i_2} \in M'$.}
\label{fig-matching-3}
\end{figure}

\item
Suppose by contradiction there exist $i_1 < i_2$ such that $x_{2i_1}, x_{2i_1+1}, x_{2i_2}, x_{2i_2+1} \in M$.
See Fig~\ref{fig-matching-4}.
In particular, $x_{2i_1-1} \in N_G(M)$ since $M$ induces a matching.
Then, $(x_1,\ldots, x_{2i_1-1}, x_{2i_2}, x_{2i_2+1}, \ldots x_{2\ell})$ is an $F$-augmenting path, thereby contradicting the minimality of $|V(P)|$.

\begin{figure}[h!]
\centering
\includegraphics[width=.25\textwidth]{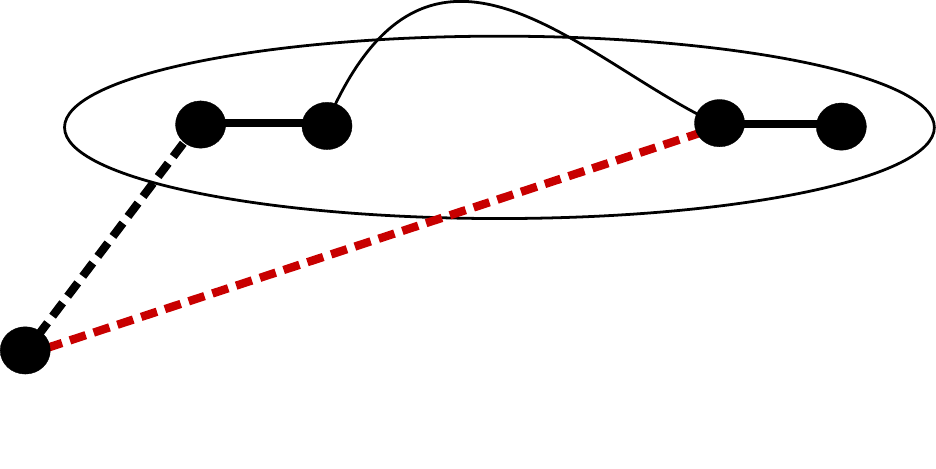}
\caption{Case $x_{2i_1}, x_{2i_1+1}, x_{2i_2}, x_{2i_2+1} \in M$.}
\label{fig-matching-4}
\end{figure}

\item
Suppose by contradiction there exist $i_1 < i_2$ such that $x_{2i_1-1},x_{2i_2-1} \notin M, \ x_{2i_1},x_{2i_2} \in M$.
See Fig~\ref{fig-matching-5}.
Either $i_1 = 1$, and so, $x_{2i_1-1} = x_1$ is unmatched, or $i_1 > 1$ and so, $x_{2i_1-1}$ is matched to $x_{2i_1-2} \neq x_{2i_2}$.
Then, $(x_1,\ldots,x_{2i_1-1},x_{2i_2},\ldots,x_{2\ell})$ is an $F$-augmenting path, thereby contradicting the minimality of $|V(P)|$.

\begin{figure}[h!]
\centering
\includegraphics[width=.16\textwidth]{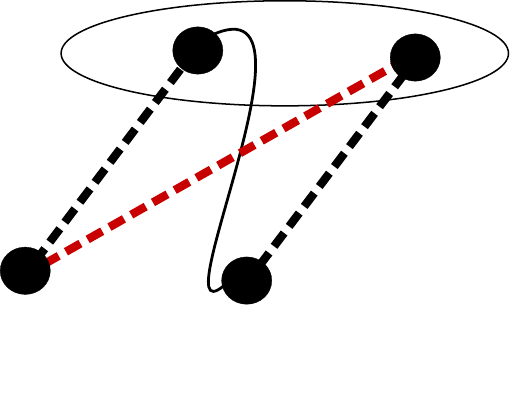}
\caption{Case $x_{2i_1-1},x_{2i_2-1} \notin M, \ x_{2i_1},x_{2i_2} \in M$.}
\label{fig-matching-5}
\end{figure}

By symmetry, the latter also proves that $|\{ i \mid x_{2i-1} \in M, \ x_{2i} \notin M\}| \leq 1$.

\item
Finally suppose by contradiction there exist $i_1 < i_2 < i_3$ such that $x_{2i_1},x_{2i_2},x_{2i_3} \notin M$, $x_{2i_1+1},x_{2i_2+1},x_{2i_3+1} \in M$.
See Fig~\ref{fig-matching-6}.
Then, $(x_1,\ldots,x_{2i_1},x_{2i_1+1},x_{2i_3},x_{2i_3+1}\ldots,x_{2\ell})$ is an $F$-augmenting path, thereby contradicting the minimality of $|V(P)|$.

\begin{figure}[h!]
\centering
\includegraphics[width=.25\textwidth]{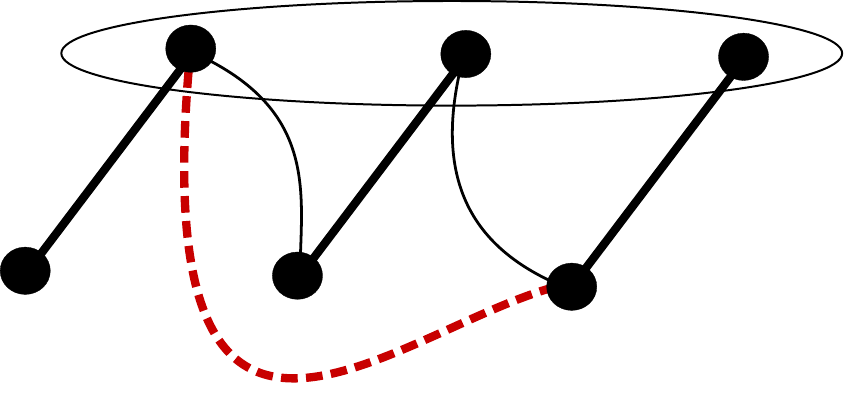}
\caption{Case $x_{2i_1},x_{2i_2},x_{2i_3} \notin M$, $x_{2i_1+1},x_{2i_2+1},x_{2i_3+1} \in M$.}
\label{fig-matching-6}
\end{figure}

We prove in the same way that if there exist $i_1 < i_2$ such that $x_{2i_1},x_{2i_2}\notin M$ and $x_{2i_1+1},x_{2i_2+1} \in M$ then $i_2 = i_1+1$.
See Fig~\ref{fig-matching-7}.
Furthermore, if $x_{2i_2+2} = x_{2i_1+4} \notin M$ then $(x_1,\ldots,x_{2i_1},x_{2i_1+1},x_{2i_2+2},x_{2i_2+3}\ldots,x_{2\ell})$ is an $F$-augmenting path, thereby contradicting the minimality of $|V(P)|$.

\begin{figure}[h!]
\centering
\includegraphics[width=.2\textwidth]{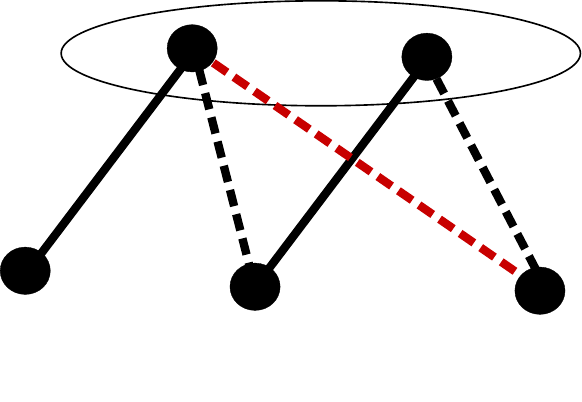}
\caption{Case $x_{2i_2+2} = x_{2i_1+4} \notin M$.}
\label{fig-matching-7}
\end{figure}

By symmetry, the same proof as above applies to $\{ i \mid x_{2i} \in M, \ x_{2i+1} \notin M\}$.

\end{enumerate}

Overall, every $M \in {\cal M}(G)$ is incident to at most $8$ edges of $P$, and so, $P$ has length ${\cal O}(|{\cal M}(G)|)$.

\end{proof}

Based on Lemmas~\ref{lem:mw-matching-reduction} and~\ref{lem:bound-edges-augmenting-path}, we introduce in what follows a {\em witness subgraph} in order to find a matching.
We think the construction could be improved but we chose to keep is as simple as possible.

\begin{definition}\label{def:witness-subgraph}
Let $G=(V,E)$ be a graph, $G'=({\cal M}(G),E')$ be its quotient graph and $F \subseteq E$ be a matching of $G$.

The witness matching $F'$ is obtained from $F$ by keeping a representative for every possible type of edge in an augmenting path.
Precisely:
\begin{itemize}
\item Let $M \in {\cal M}(G)$.
If $E(G[M]) \cap F \neq \emptyset$ then there is exactly one edge $\{u_M,v_M\} \in E(G[M]) \cap F$ such that $\{u_M,v_M\} \in F'$.
Furthermore if $E(G[M]) \setminus F \neq \emptyset$ then we pick an edge $\{x_M,y_M\} \in E(G[M]) \setminus F$ and we add in $F'$ every edge in $F$ that is incident to either $x_M$ or $y_M$.
\item Let $M,M' \in {\cal M}(G)$ be adjacent in $G'$.
There are exactly $\min \{ 4, |F \cap (M \times M')| \}$ edges $\{v_M,v_{M'}\}$ added in $F'$ such that $v_M \in M, \ v_{M'} \in M'$ and $\{v_M,v_{M'}\} \in F$.
\end{itemize}
The witness subgraph $G_F'$ is the subgraph induced by $V(F')$ with at most two unmatched vertices added for every strong module.
Formally, let $M \in {\cal M}(G)$.
The submodule $M_F \subseteq M$ contains exactly $\min \{ 2, |M \setminus V(F)| \}$ vertices of $M \setminus V(F)$.
Then, $$G_F' = G\left[ V(F') \cup \left( \bigcup_{M \in {\cal M}(G)} M_F \right) \right]. $$
\end{definition}

As an example, suppose that every edge of $F$ has its two ends in a same module and every module induces a matching.
Then, $G_F'$ is obtained from $G'$ by substituting every $M \in {\cal M}(G)$ with at most one edge (if $F \cap E(G[M]) \neq \emptyset$) and at most two isolated vertices (representing unmatched vertices).

From the algorithmic point of view, we need to upper-bound the size of the witness subgraph, as follows.

\begin{lemma}\label{lem:witness-size}
Let $G=(V,E)$ be a graph, $G'=({\cal M}(G),E')$ be its quotient graph and $F \subseteq E$ be a matching of $G$.
The witness subgraph $G_F'$ has order ${\cal O}(|E(G')|)$.
\end{lemma}

\begin{proof}
By construction for every $M \in {\cal M}(G)$ we have $|M \cap V(G_F')| = {\cal O}(deg_{G'}(M))$.
Therefore, $|V(G_F')| = \sum_{M \in {\cal M}(G)} |M \cap V(G_F')| = {\cal O}(|E(G')|)$.
%
\end{proof}

Our algorithm is based on the correspondance between $F$-augmenting paths in $G$ and $F'$-augmenting paths in $G'_F$, that we prove next.
The following Lemma~\ref{lem:witness-maxmatching} is the key technical step of the algorithm.

\begin{lemma}\label{lem:witness-maxmatching}
Let $G=(V,E)$ be a graph such that every module $M \in {\cal M}(G)$ induces a matching.
Let $F \subseteq E$ be a matching of $G$ such that $\bigcup_{M \in {\cal M}(G)} V(E(G[M])) \subseteq V(F)$.
There exists an $F$-augmenting path in $G$ if and only if there exists an $F'$-augmenting path in $G_F'$.
\end{lemma}

\begin{proof}
In one direction, $G_F'$ is an induced subgraph of $G$.
Furthermore, according to Definition~\ref{def:witness-subgraph}, $F' \subseteq F$ and $V(F') = V(F) \cap V(G_F')$.
Thus, every $F'$-augmenting path in $G_F'$ is also an $F$-augmenting path in $G$. 

Conversely, suppose there exists an $F$-augmenting path in $G$.
Let $P=(v_1,v_2,\ldots,v_{2\ell})$ be an $F$-augmenting path in $G$ that satisfies the conditions of Lemma~\ref{lem:bound-edges-augmenting-path}.
We transform $P$ into an $F'$-augmenting path in $G_F'$ as follows.
For every $1 \leq i \leq 2\ell$ let $M_i \in {\cal M}(G)$ such that $v_i \in M_i$.
\begin{itemize}
\item We choose $u_1 \in M_1 \cap V(G_F'), \ u_{2\ell} \in M_{2\ell} \cap V(G_F')$ unmatched.
Furthermore, if $M_1 = M_{2\ell}$ then we choose $u_1 \neq u_{2\ell}$.
The two of $u_1,u_{2\ell}$ exist according to Definition~\ref{def:witness-subgraph}.
\item Then, for every $1 \leq i \leq \ell-1$, we choose $u_{2i} \in M_{2i} \cap V(G_F'), \ u_{2i+1} \in M_{2i+1} \cap V(G_F')$ such that $\{u_{2i},u_{2i+1}\} \in F'$.
Note that if $M_{2i} = M_{2i+1}$ then $\{u_{2i},u_{2i+1}\}$ is the unique edge of $F' \cap E(G[M_{2i}])$.
By Lemma~\ref{lem:bound-edges-augmenting-path} we also have that $\{v_{2i},v_{2i+1}\}$ is the unique edge of $E(P) \cap F$ such that $v_{2i},v_{2i+1} \in M_{2i}$.
Otherwise, $M_{2i} \neq M_{2i+1}$.
If there are $p$ edges $e \in F$ with one end in $M_{2i}$ and the other end in $M_{2i+1}$ then there are at least $\min\{p,4\}$ such edges in $F'$.
By Lemma~\ref{lem:bound-edges-augmenting-path} there are at most $\min\{p,4\}$ edges $e \in E(P) \cap F$ with one end in $M_{2i}$ and the other end in $M_{2i+1}$.
Hence, we can always ensure the $u_j$'s, $1 \leq j \leq 2\ell$, to be pairwise different.
\end{itemize}
The resulting sequence ${\cal S}_P = (u_1,u_2,\ldots,u_{2\ell})$ is not necessarily a path, since two consecutive vertices $u_{2i-1},u_{2i}$ need not be adjacent in $G_F'$.
Roughly, we insert alternating subpaths in the sequence in order to make it a path.
However, we have to be careful not to use twice a same vertex for otherwise we would only obtain a walk.

Let $I_P = \{ i \mid \{u_{2i-1},u_{2i}\} \notin E \}$.
Observe that for every $i \in I_P$ we have $M_{2i-1} = M_{2i}$.
In particular, since we assume $\bigcup_{M \in {\cal M}(G)} V(E(G[M])) \subseteq V(F)$, it implies $i \notin \{1,\ell\}$.
Furthermore, $M_{2i-2} \neq M_{2i}$ and $M_{2i} \neq M_{2i+1}$ since otherwise $v_{2i-2},v_{2i-1},v_{2i} \in M_{2i}$ or $v_{2i-1},v_{2i},v_{2i+1} \in M_{2i}$ thereby contradicting that $M_{2i}$ induces a matching.
According to Definition~\ref{def:witness-subgraph} there exist $x_i,y_i \in M_{2i}$ such that $\{x_i,y_i\} \in E(G[M_{2i}]) \setminus F$ and every edge of $F$ that is incident to either $x_i$ or $y_i$ is in $F'$.
Such two edges always exist since we assume $\bigcup_{M \in {\cal M}(G)} V(E(G[M])) \subseteq V(F)$, hence there exist $w_i,z_i$ such that $\{w_i,x_i\}, \{y_i,z_i\} \in F'$.
Note that $w_i,z_i \notin M_{2i}$ since $\{x_i,y_i\} \in E(G[M_{2i}])$ and $M_{2i}$ induces a matching.

Since by Lemma~\ref{lem:bound-edges-augmenting-path} $\{v_{2i-1},v_{2i}\}$ is the unique edge $e \in E(P) \setminus F$ such that $e \subseteq M_{2i}$, the vertices $x_i,y_i, \ i \in I_P$ are pairwise different.
Furthermore, we claim that there can be no $i_1, i_2 \in I_P$ such that $s_{i_1} \in \{x_{i_1},y_{i_1}\}$ and $t_{i_2} \in \{x_{i_2},y_{i_2}\}$ are adjacent in $G$.
Indeed otherwise, $M_{2i_1} \in N_{G'}(M_{2i_2})$, there exist $v_{2i_1-1},v_{2i_1} \in M_{2i_1}$ and $v_{2i_2-1},v_{2i_2} \in M_{2i_2}$, thereby contradicting Lemma~\ref{lem:bound-edges-augmenting-path}.
As a result, all the vertices $w_i,x_i,y_i,z_i, \ i \in I_P$ are pairwise different. 
However, we may have $\{w_i,x_i\} = \{u_{2j},u_{2j+1}\}$ or $\{y_i,z_i\} = \{u_{2j},u_{2j+1}\}$ for some $j$.

We consider the indices $i \in I_P$ sequentially, by increasing value.
By Lemma~\ref{lem:bound-edges-augmenting-path}, $v_{2j} \notin M_{2i}, \ v_{2j+1} \in M_{2i}$ for some $j \neq i-1$ implies $j=i-2$.
Similarly (obtained by reverting the indices, from $v_1'=v_{2\ell}$ to $v_{2\ell}'=v_1$), $v_{2j} \in M_{2i}, \ v_{2j+1} \notin M_{2i}$ for some $j \neq i$ implies $j=i+1$.
Therefore, if $(w_i,x_i) \in {\cal S}_p$ then $(w_i,x_i) \in \{ (u_{2i-4},u_{2i-3}), (u_{2i-2},u_{2i-1}), (u_{2i+1},u_{2i}), (u_{2i+3},u_{2i+2}) \}$, and the same holds for $(y_i,z_i)$. 
Note also that the pairs $(w_i,x_i)$ and $(z_i,y_i)$ play a symmetric role.
Thus we can reduce by symmetries (on the sequence and on the two of $(w_i,x_i)$ and $(z_i,y_i)$) to the six following cases:

\begin{itemize}

\item Case $x_i,y_i \notin {\cal S}_P$. 
See Fig~\ref{fig-replacement-1}.
In particular, $w_i,z_i \notin {\cal S}_P$.

\begin{figure}[h!]
\centering
\includegraphics[width=.15\textwidth]{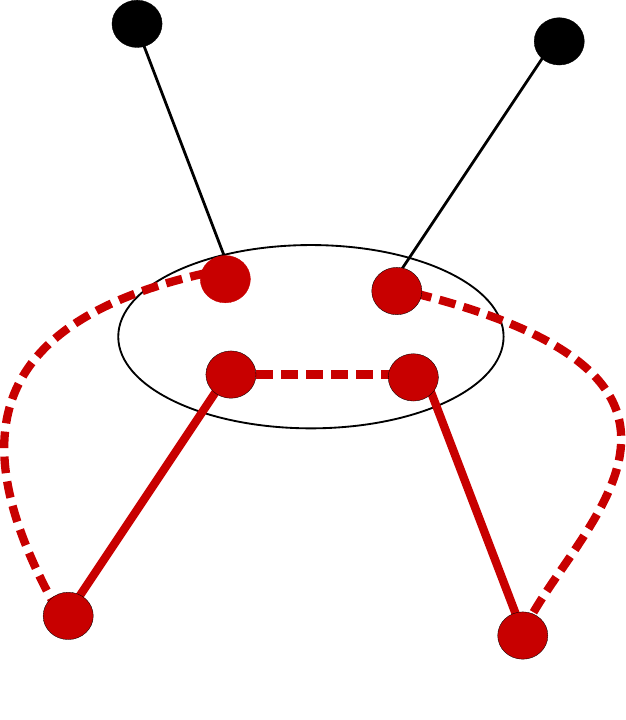}
\caption{Case $x_i,y_i \notin {\cal S}_P$.}
\label{fig-replacement-1}
\end{figure}

We insert the $F'$-alternating subpath $(u_{2i-1},w_i,x_i,y_i,z_i,u_{2i})$.

\item Case $x_i \notin {\cal S}_P$, $y_i \in \{u_{2i-1},u_{2i}\}$.
We assume by symmetry $y_i = u_{2i}$.
See Fig~\ref{fig-replacement-2}.
In particular, we have $w_i \notin {\cal S}_P$.

\begin{figure}[h!]
\centering
\includegraphics[width=.15\textwidth]{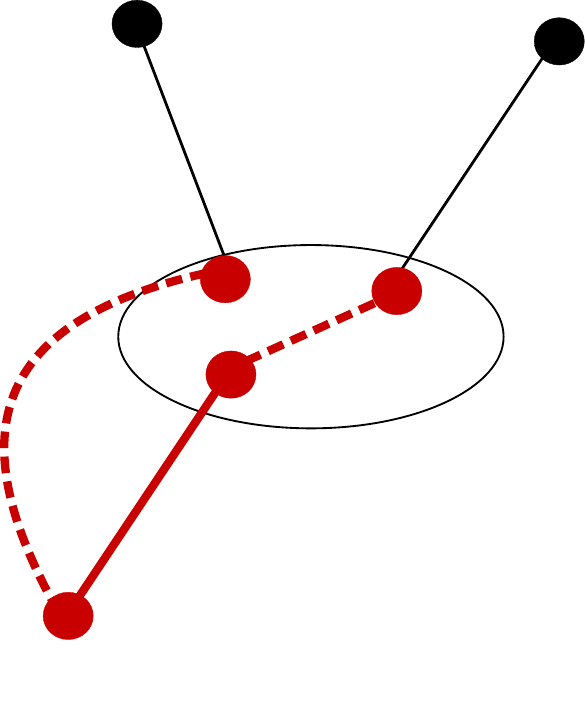}
\caption{Case $x_i \notin {\cal S}_P$, $y_i = u_{2i}$.}
\label{fig-replacement-2}
\end{figure}

We insert the $F'$-alternating subpath $(u_{2i-1},w_i,x_i,y_i = u_{2i})$.
Note that the case $x_i \in \{u_{2i-1},u_{2i}\}$, $y_i \notin {\cal S}_P$ is symmetrical to this one.

%
%

\item Case $x_i \notin {\cal S}_P$, $y_i \in {\cal S}_p \setminus \{u_{2i-1},u_{2i}\}$.
We assume by symmetry $(y_i,z_i) = (u_{2i+2},u_{2i+3})$ (the case $(z_i,y_i) = (u_{2i-4},u_{2i-3})$ is obtained by reverting the indices along the sequence).
See Fig~\ref{fig-replacement-4}.

\begin{figure}[h!]
\centering
\includegraphics[width=.2\textwidth]{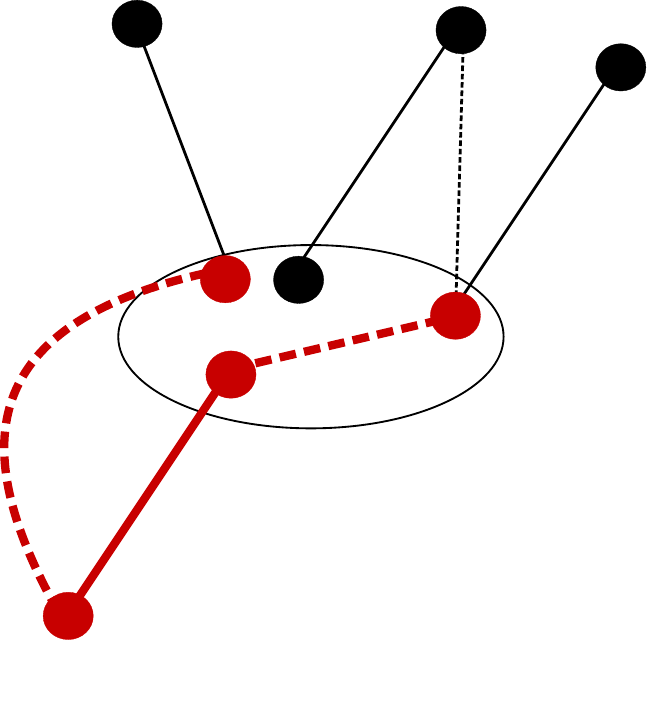}
\caption{Case $x_i \notin {\cal S}_P$, $(y_i,z_i) = (u_{2i+2},u_{2i+3})$.}
\label{fig-replacement-4}
\end{figure}

We replace $(u_{2i-1},u_{2i},u_{2i+1},y_i=u_{2i+2})$ by the $F'$-alternating subpath $(u_{2i-1},w_i,x_i,y_i)$.
Note that the case $x_i \in {\cal S}_p \setminus \{u_{2i-1},u_{2i}\}$, $y_i \notin {\cal S}_P$ is symmetrical to this one.

%
%

\item Case $(w_i,x_i) = (u_{2i-4},u_{2i-3})$, $y_i = u_{2i}$.
See Fig~\ref{fig-replacement-6}.

\begin{figure}[h!]
\centering
\includegraphics[width=.16\textwidth]{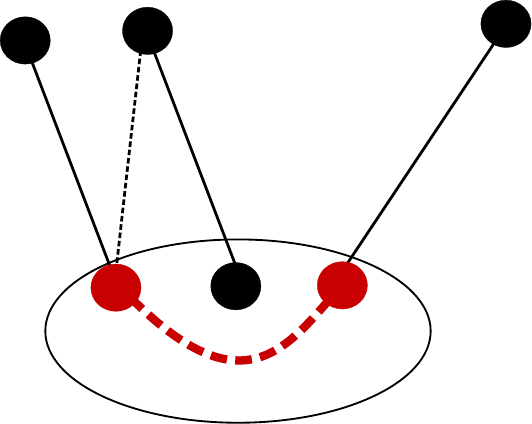}
\caption{Case $(w_i,x_i) = (u_{2i-4},u_{2i-3})$, $y_i = u_{2i}$.}
\label{fig-replacement-6}
\end{figure}

We replace $(u_{2i-3}=x_i,u_{2i-2},u_{2i-1},u_{2i}=y_i)$ by the $F'$-alternating subpath $(x_i,y_i)$.
Note that the case $x_i = u_{2i-1}$, $(y_i,z_i) = (u_{2i+2},u_{2i+3})$, and the two more cases obtained by switching the respective roles of $(x_i,w_i)$ and $(y_i,z_i)$, are symmetrical to this one.

%
%

\item Case $(w_i,x_i) = (u_{2i-4},u_{2i-3})$, $y_i = u_{2i-1}$.
See Fig~\ref{fig-replacement-6-b}.

\begin{figure}[h!]
\centering
\includegraphics[width=.16\textwidth]{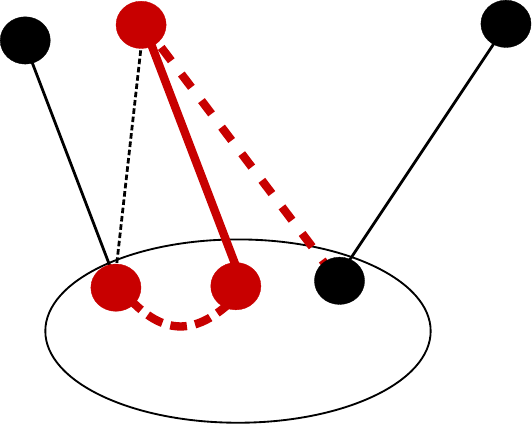}
\caption{Case $(w_i,x_i) = (u_{2i-4},u_{2i-3})$, $y_i = u_{2i-1}$.}
\label{fig-replacement-6-b}
\end{figure}

We replace $(u_{2i-3}=x_i,u_{2i-2}=z_i,u_{2i-1}=y_i,u_{2i})$ by the $F'$-alternating subpath $(x_i,y_i,z_i,u_{2i})$.
Note that the case $x_i = u_{2i}$, $(y_i,z_i) = (u_{2i+2},u_{2i+3})$, and the two more cases obtained by switching the respective roles of $(x_i,w_i)$ and $(y_i,z_i)$, are symmetrical to this one.

\item Case $(w_i,x_i) = (u_{2i-4},u_{2i-3})$, $(y_i,z_i) = (u_{2i+2},u_{2i+3})$.
See Fig~\ref{fig-replacement-8}.

\begin{figure}[h!]
\centering
\includegraphics[width=.16\textwidth]{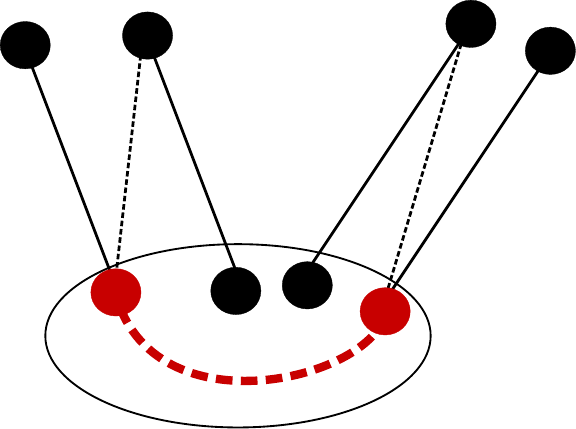}
\caption{Case $(w_i,x_i) = (u_{2i-4},u_{2i-3})$, $(y_i,z_i) = (u_{2i+2},u_{2i+3})$.}
\label{fig-replacement-8}
\end{figure}

Since $x_i = u_{2i-3}, \ y_i = u_{2i+2}$ are adjacent we can remove $(u_{2i-2},u_{2i-1},u_{2i},u_{2i+1})$ from ${\cal S}_P$. 

\end{itemize} 
Overall, in every case the procedure only depends on the subsequence between $u_{2i-4}$ and $u_{2i+3}$.
In order to prove correctness of the procedure, it suffices to prove that this subsequence has not been modified for a smaller $i' \in I_P, \ i' < i$.
Equivalently, we prove that the above procedure does not modify the subsequence between $u_{2j-4}$ and $u_{2j+3}$ for any $j \in I_P, \ j > i$.
First we claim $j \geq i+2$.
Indeed, $M_{2i+1} \in N_{G'}(M_{2i})$.
Hence, by Lemma~\ref{lem:bound-edges-augmenting-path}, $i \in I_P$ implies $i+1 \notin I_P$, that proves the claim.
In this situation, $2j-4 \geq 2i$.
Furthermore, the subsequence $(u_{2i},\ldots,u_{2\ell})$ is modified only if $u_{2i+2} \in \{x_i,y_i\}$.
However in the latter case we have $u_{2i+3} \notin M_{2i}$, hence $M_{2i+3} \in N_{G'}(M_{2i})$, and so, by Lemma~\ref{lem:bound-edges-augmenting-path} $j \geq i+3$.
In particular, $2j-4 \geq 2i+2$ and the subsequence $(u_{2i+2},\ldots,u_{2\ell})$ is not modified by the procedure. 
Altogether, it proves that the above procedure is correct.

Finally, applying the above procedure for all $i \in I_P$ leads to an $F'$-alternating path in $G_F'$. 
\end{proof}

We can now state the main result in this subsection.

\begin{theorem}\label{thm:mw-maxmatching}
For every $G=(V,E)$, {\sc Maximum Matching} can be solved in ${\cal O}(mw(G)^4 \cdot n + m)$-time.
\end{theorem}

\begin{proof}
The algorithm is recursive.
If $G$ is trivial (reduced to a single node) then we output an empty matching.
Otherwise, let $G' = ({\cal M}(G),E')$ be the quotient graph of $G$.
For every module $M \in {\cal M}(G)$, we call the algorithm recursively on $G[M]$ in order to compute a maximum matching $F_M$ of $G[M]$.
Let $F^* = \bigcup_{M \in {\cal M}(G)} F_M$.
By Lemma~\ref{lem:mw-matching-reduction} (applied to every $M \in {\cal M}(G)$ sequentially), we are left to compute a maximum matching for $G^* = (V, (E \setminus \bigcup_{M \in {\cal M}(G)} E(G[M])) \cup F^*)$.
Therefore from now on assume $G = G^*$.

\smallskip
If $G'$ is edgeless then we can output $F^*$.
Otherwise, by Theorem~\ref{thm:modular-dec} $G'$ is either prime for modular decomposition or a complete graph.

Suppose $G'$ to be prime.
We start from $F_0 = F^*$.
Furthermore, we ensure that the two following hold at every step $t \geq 0$:
\begin{itemize}
\item All the vertices that are matched in $F^*$ are also matched in the current matching $F_t$.
For instance, it is the case if $F_t$ is obtained from $F_0$ by only using augmenting paths in order to increase the cardinality of the matching.
\item For every $M \in {\cal M}(G)$ we store $|F_M \cap F_t|$.
For every $M,M' \in {\cal M}(G)$ adjacent in $G'$ we store $|(M \times M') \cap F_t|$.
In particular, $|F_M \cap F_0| = |F_M|$ and $|(M \times M') \cap F_0| = 0$.
So, it takes time ${\cal O}(\sum_{M \in {\cal M}(G)}deg_{G'}(M))$ to initialize this information, that is in ${\cal O}(|E(G')|) = {\cal O}(mw(G)^2)$.
Furthermore, it takes ${\cal O}(\ell)$-time to update this information if we increase the size of the matching with an augmenting path of length $2\ell$.
\end{itemize}

We construct the graph $G'_{F_t}$ according to Definition~\ref{def:witness-subgraph}.
By using the information we store for the algorithm, it can be done in ${\cal O}(|E(G'_{F_t})|)$-time, that is in ${\cal O}(|E(G')|^2) = {\cal O}(mw(G)^4)$ by Lemma~\ref{lem:witness-size}.
Furthermore by Theorem~\ref{thm:berge} there exists an $F_t$-augmenting path if and only if $F_t$ is not maximum.
Since we can assume all the modules in ${\cal M}(G)$ induce a matching, by Lemma~\ref{lem:witness-maxmatching} there exists an $F_t$-augmenting path in $G$ if and only if there exists an $F_t'$-augmenting path in $G'_{F_t}$.
So, we are left to compute an $F_t'$-augmenting path in $G'_{F_t}$ if any.
It can be done in ${\cal O}(|E(G_{F_t}')|)$-time~\cite{GaT83}, that is in ${\cal O}(mw(G)^4)$.
Furthermore, by construction of $G_{F_t}'$, an $F_t'$-augmenting path $P'$ in $G'_{F_t}$ is also an $F_t$-augmenting path in $G$.
Thus, we can obtain a larger matching $F_{t+1}$ from $F_t$ and $P$.
We repeat the procedure above for $F_{t+1}$ until we reach a maximum matching $F_{t_{\max}}$.
The total running time is in ${\cal O}(mw(G)^4 \cdot t_{\max})$.

\smallskip
Finally, assume $G'$ to be complete.
Let ${\cal M}(G) = \{M_1,M_2, \ldots,M_k\}$ be linearly ordered.
For every $1 \leq i \leq k$, write $G_i = G[\bigcup_{j \leq i}M_j]$.
We compute a maximum matching $F^i$ for $G_i$, from a maximum matching $F^{i-1}$ of $G_{i-1}$ and a maximum matching $F_{M_i}$ of $G[M_i]$, sequentially.
For that, we apply the same techniques as for the prime case, to some ``pseudo-quotient graph'' $G_i'$ isomorphic to $K_2$ ({\it i.e.}, the two vertices of $G_i'$ respectively represent $V(G_{i-1})$ and $M_i$).
Since the pseudo-quotient graphs have size two, this step takes total time ${\cal O}(|V(G')| + (|F^k| - |F^*|))$.

\medskip
Overall, summing the order of all the subgraphs in the modular decomposition of $G$ amounts to ${\cal O}(n)$~\cite{Rao08b}.
Furthermore, a maximum matching of $G$ also has cardinality ${\cal O}(n)$.
Therefore, the total running time is in ${\cal O}(mw(G)^4 \cdot n)$ if the modular decomposition of $G$ is given.
The latter decomposition can be precomputed in ${\cal O}(n+m)$-time~\cite{TCHP08}.
\end{proof}

\subsection{More structure: $(q,q-3)$-graphs}
\label{sec:matching-qq3}

The second main result in Section~\ref{sec:maxmatching} is an ${\cal O}(q(G)^4 \cdot n + m)$-time algorithm for {\sc Maximum Matching} (Theorem~\ref{thm:qq3-maxmatching}).
Our algorithm for $(q,q-3)$-graphs reuses the algorithm described in Theorem~\ref{thm:mw-maxmatching} as a subroutine.
However, applying the same techniques to a case where the quotient graph has super-constant size $\Omega(q)$ happens to be more challenging.
Thus we need to introduce new techniques in order to handle with all the cases presented in Lemma~\ref{lem:reduce-qq3}.

Computing a maximum matching for the {\em quotient graph} is easy.
However, we also need to account for the edges present inside the modules.
For that, we need the following stronger variant of Lemma~\ref{lem:reduce-qq3}.
The latter generalizes similar structure theorems that have been obtained for some specific subclasses~\cite{GRT97}.

\begin{theorem}\label{thm:stronger-mdc-qq3}
For an arbitrary $(q,q-3)$-graph $G$, $q \geq 7$, and its quotient graph $G'$, exactly one of the following conditions is satisfied.
\begin{enumerate}
\item $G$ is disconnected;
\item $\overline{G}$ is disconnected;
\item $G$ is a disc (and so, $G=G'$ is prime for modular decomposition);
\item $G$ is a spider (and so, $G'$ is a prime spider);
\item $G'$ is a spiked $p$-chain $P_k$, or a spiked $p$-chain $\overline{P_k}$.
Furthermore, for every $v \in V(G')$, if the corresponding module $M_v \in {\cal M}(G)$ is such that $|M_v| \geq 2$ then we have $v \in \{v_1,v_k,x,y\}$;
\item $G'$ is a spiked $p$-chain $Q_k$, or a spiked $p$-chain $\overline{Q_k}$.
Furthermore, for every $v \in V(G')$, if the corresponding module $M_v \in {\cal M}(G)$ is such that $|M_v| \geq 2$ then we have either $v \in \{v_1,v_k\}$ or $v = z_i$ for some $i$;
\item $|V(G')| \leq q$.
\end{enumerate}
\end{theorem}

The proof of Theorem~\ref{thm:stronger-mdc-qq3} is postponed to the appendix.
It is based on a refinement of modular decomposition called {\em primeval decomposition}.

In what follows, we introduce our techniques for the cases where the quotient graph $G'$ is neither degenerate nor of constant size.

\subsubsection*{Simple cases}

\begin{lemma}\label{lem:disc}
For every disc $G=(V,E)$, a maximum matching can be computed in linear-time.
\end{lemma}

\begin{proof}
If $G=C_n, \ n \geq 5$ is a cycle then the set of edges $\{ \{2i,2i+1\} \mid 0 \leq i \leq \left\lfloor n/2 \right\rfloor - 1 \}$ is a maximum matching.
Otherwise, $G=\overline{C_n}$ is a co-cycle.  
Let $F_n$ contain all the edges $\{4i,4i+2\}, \ \{4i+1,4i+3\}, \ 0 \leq i \leq \left\lfloor n/4 \right\rfloor - 1$.
There are three cases.
If $n = 0 \pmod 4$ or $n = 1 \pmod 4$ then there is at most one vertex unmatched by $F_n$, and so, $F_n$ is a maximum matching of $G$.
Otherwise, if $n = 3 \pmod 4$ then a maximum matching of $G$ is obtained by adding the edge $\{n-3,n-1\}$ to $F_n$.
Finally, assume $n = 2 \pmod 4$.
By construction, $F_n$ leaves unmatched the two of $n-2,n-1$.
We obtain a perfect matching of $G$ from $F_n$ by replacing $\{0,2\}$ with $\{n-2,0\}, \ \{n-1,2\}$.
Note that it is possible to do that since $n \geq 5$.
\end{proof}

\begin{lemma}\label{lem:spider}
If $G=(S \cup K \cup R, E)$ is a spider then there exists a maximum matching of $G$ composed of: a perfect matching between $K$ and $S$; and a maximum matching of $G[R]$.
\end{lemma}

\begin{proof}
We start from a perfect matching $F_0$ between $K$ and $S$.
We increase the size of $F_0$ using augmenting paths until it is no more possible to do so.
By Theorem~\ref{thm:berge}, the obtained matching $F_{\max}$ is maximum.
Furthermore, either there is a perfect matching between $K$ and $S$ or there is at least one vertex of $S$ that is unmatched.
Since $V(F_0) \subseteq V(F_{\max})$ the latter proves the lemma.
\end{proof}

\subsubsection*{The case of prime $p$-trees}

Roughly, when the quotient graph $G'$ is a prime $p$-tree, our strategy consists in applying the following reduction rules until the graph is empty.
\begin{enumerate}
\item Find an isolated module $M$ (with no neighbour).
Compute a maximum matching for $G[M]$ and for $G[V \setminus M]$ separately.
\item Find a pending module $M$ (with one neighbour $v$).
Compute a maximum matching for $G[M]$.
If it is not a perfect matching then add an edge between $v$ and any unmatched vertex in $M$, then discard $M \cup \{v\}$.
Otherwise, discard $M$ (Lemma~\ref{lem:pending-module}).
\item Apply a technique known as ``{\sc SPLIT} and {\sc MATCH}''~\cite{YuY93} to some module $M$ and its neighbourhood $N_G(M)$.
We do so only if $M$ satisfies some properties.
In particular, we apply this rule when $M$ is a universal module (with a complete join between $M$ and $V \setminus M$).
See Definition~\ref{def:split-and-match} and Lemma~\ref{lem:join}.
\end{enumerate} 

We introduce the reduction rules below and we prove their correctness.

\paragraph{Reduction rules.}
The following lemma generalizes a well-known reduction rule for {\sc Maximum Matching}: add a pending vertex and its unique neighbour to the matching then remove this edge~\cite{KaS81}.

\begin{lemma}\label{lem:pending-module}
Let $M$ be a module in a graph $G=(V,E)$ such that $N_G(M) = \{v\}$, $F_M$ is a maximum matching of $G[M]$ and $F_M^*$ is obtained from $F_M$ by adding an edge between $v$ and any unmatched vertex of $M$ (possibly, $F_M^* = F_M$ if it is a perfect matching).
There exists a maximum matching $F$ of $G$ such that $F_M^* \subseteq F$.  
\end{lemma}

\begin{proof}
By Lemma~\ref{lem:mw-matching-reduction}, every maximum matching for $G_M' = (V, (E \setminus E(G[M]) \cup F_M)$ is also a maximum matching for $G$.
There are two cases.

Suppose there exists $u \in M \setminus V(F_M)$.
Then, $u$ is a pending vertex of $G_M'$.
There exists a maximum matching of $G_M'$ that contains the edge $\{u,v\}$~\cite{KaS81}.
Furthermore, removing $u$ and $v$ disconnects the vertices of $M \setminus u$ from $V \setminus N_G[M]$.
It implies that a maximum matching $F'$ of $G \setminus (u,v)$ is the union of any maximum matching of $G[M \setminus u]$ with any maximum matching of $G[V \setminus N_G[M]]$.
In particular, $F_M$ is contained in some maximum matching $F'$ of $G \setminus (u,v)$.
Since $\{u,v\}$ is contained in a maximum matching of $G$, therefore $F = F' \cup \{\{u,v\}\}$ is a maximum matching of $G$.
We are done since $F_M^* = F_M \cup \{\{u,v\}\} \subseteq F$ by construction.

Otherwise, $F_M$ is a perfect matching of $G[M]$.
For every edge $\{x,y\} \in F_M$, we have that $x,y$ have degree two in $G_M'$.
The following reduction rule has been proved to be correct in~\cite{KaS81}: remove any $x$ of degree two, merge its two neighbours and increase the size of the solution by one unit.
In our case, since $N_{G_M'}[y] \subseteq N_{G_M'}[v]$ the latter is equivalent to put the edge $\{x,y\}$ in the matching. 
Overall, applying the reduction rule to all edges $\{x,y\} \in F_M$ indeed proves the existence of some maximum matching $F$ such that $F_M = F_M^* \subseteq F$. 
\end{proof}

Then, we introduce a technique known as ``{\sc SPLIT} and {\sc MATCH}'' in the literature~\cite{YuY93}.

\begin{definition}\label{def:split-and-match}
Let $G=(V,E)$ be a graph, $F \subseteq E$ be a matching of $G$.
Given some module $M \in {\cal M}(G)$ we try to apply the following two operations until none of them is possible:
\begin{itemize}
\item Suppose there exist $u \in M, \ v \in N_G(M)$ unmatched.
We add an edge $\{u,v\}$ to the matching ({\sc MATCH}).
\item Otherwise, suppose there exist $u,u' \in M, \ v,v' \in N_G(M)$ such that $u$ and $u'$ are unmatched, and $\{v,v'\}$ is an edge of the matching.
We replace the edge $\{v,v'\}$ in the matching by the two new edges $\{u,v\}, \ \{u',v'\}$ ({\sc SPLIT}).
\end{itemize}
\end{definition}

The ``{\sc SPLIT} and {\sc MATCH}'' has been applied to compute a maximum matching in linear-time for cographs and some of its generalizations~\cite{FGV99,FPT97,YuY93}.
Our Theorem~\ref{thm:mw-maxmatching} can be seen as a broad generalization of this technique.
In what follows, we introduce more cases where the ``{\sc SPLIT} and {\sc MATCH}'' technique can be used in order to compute a maximum matching directly.

\begin{lemma}\label{lem:join}
Let $G = G_1 \oplus G_2$ be the join of two graphs $G_1,G_2$ and let $F_1,F_2$ be maximum matchings for $G_1,G_2$, respectively.
For $F = F_1 \cup F_2$, applying the `{\sc SPLIT} and {\sc MATCH}'' technique to $V(G_1)$, then to $V(G_2)$ leads to a maximum matching of $G$.
\end{lemma}

\begin{proof}
The lemma is proved in~\cite{YuY93} when $G$ is a cograph.
In particular, let $G^* = (V, (V(G_1) \times V(G_2)) \cup F_1 \cup F_2)$.
Since it ignores the edges from $(E(G_1) \setminus F_1) \cup (E(G_2) \setminus F_2)$, the procedure outputs the same matching for $G$ and $G^*$.
Furthermore, $G^*$ is a cograph, and so, the outputted matching is maximum for $G^*$.
By Lemma~\ref{lem:mw-matching-reduction}, a maximum matching for $G^*$ is a maximum matching for $G$.
\end{proof}

\paragraph{Applications.}

We can now combine our reductions rules as follows.

\begin{proposition}\label{prop:maxmatching-prime-ptree}
Let $G=(V,E)$ be a $(q,q-3)$-graph, $q \geq 7$, such that its quotient graph $G'$ is isomorphic to a prime $p$-tree.
For every $M \in {\cal M}(G)$ let $F_M$ be a maximum matching of $G[M]$ and let $F^* = \bigcup_{M \in {\cal M}(G)}F_M$.

A maximum matching $F_{\max}$ for $G$ can be computed in ${\cal O}(|V(G')|+|E(G')|+|F_{\max}|-|F^*|)$-time if $F^*$ is given as part of the input.
\end{proposition}

\begin{proof}
There are five cases.
If $G'$ has order at most $7$ then we can apply the same techniques as for Theorem~\ref{thm:mw-maxmatching}.
Otherwise, $G'$ is either a spiked $p$-chain $P_k$, a spiked $p$-chain $\overline{P_k}$, a spiked $p$-chain $Q_k$ or a spiked $p$-chain $\overline{Q_k}$.

\paragraph{Case $G$ is a spiked $p$-chain $P_k$.}
By Theorem~\ref{thm:stronger-mdc-qq3} we have that $(v_2,v_3,\ldots,v_{k-1})$ are vertices of $G$.
In this situation, since $N_{G'}(v_1) = v_2$, $M_{v_1}$ is a pending module.
We can apply the reduction rule of Lemma~\ref{lem:pending-module} to $M_{v_1}$.
Doing so, we discard $M_{v_1}$ and possibly $v_2$.
Let $S = M_x$ if $v_2$ has already been discarded and let $S = M_x \cup \{v_2\}$ otherwise.
We have that $S$ is a pending module in the resulting subgraph, with $v_3$ being its unique neighbour.
Furthermore, by Lemma~\ref{lem:pending-module} we can compute a maximum matching of $G[S]$ from $F_{M_{v_x}}$, by adding an edge between $v_2$ (if it is present) and an unmatched vertex in $M_x$ (if any).
So, we again apply the reduction rule of Lemma~\ref{lem:pending-module}, this time to $S$.
Doing so, we discard $S$, and possibly $v_3$.
Then, by a symmetrical argument we can also discard $M_{v_k}, \ M_y, \ v_{k-1}$ and possibly $v_{k-2}$.
We are left with computing a maximum matching for some subpath of $(v_3,v_4,\ldots,v_{k-2})$, that can be done in linear-time by taking half of the edges.

\paragraph{Case $G$ is a spiked $p$-chain $\overline{P_k}$.}
By Theorem~\ref{thm:stronger-mdc-qq3}, the nontrivial modules of ${\cal M}(G)$ can only be $M_{v_1},M_{v_k},M_x,M_y$. 
In particular, $F^* = F_{M_{v_1}} \cup F_{M_{v_k}} \cup F_{M_x} \cup F_{M_y}$.
Let $U = M_{v_1} \cup M_{v_k} \cup M_x \cup M_y$.
The graph $G \setminus U$ is isomorphic to $\overline{P_{k-2}}, \ k \geq 6$.
Furthermore, let $F_{k-2}$ contain the edges $\{v_2,v_{\left\lceil k/2 \right\rceil +1}\}, \ \{v_{\left\lfloor k /2 \right\rfloor}, v_{k-1}\}$ plus all the edges $\{v_i,v_{k+1-i}\}, \ 3 \leq i \leq \left\lfloor k/2 \right\rfloor -1$.
Observe that $F_{k-2}$ is a maximum matching of $\overline{P_{k-2}}$.
In particular it is a perfect matching of $\overline{P_{k-2}}$ if $k$ is even, and if $k$ is odd then it only leaves vertex $v_{\left\lceil k/2 \right\rceil}$ unmatched.
We set $F_0 = F^* \cup F_{k-2}$ to be the initial matching.
Then, we repeat the procedure below until we cannot increase the matching anymore.
We consider the modules $M \in \{ M_{v_1},M_{v_k},M_x,M_y \}$ sequentially.
For every $M$ we try to apply the {\sc SPLIT} and {\sc MATCH} technique of Definition~\ref{def:split-and-match}.

Overall, we claim that the above procedure can be implemented to run in constant-time per loop.
Indeed, assume that the matched vertices (resp., the unmatched vertices) are stored in a list in such a way that all the vertices in a same module $M_v, \ v \in V(G')$ are consecutive.
For every matched vertex $u$, we can access to the vertex that is matched with $u$ in constant-time.
Furthermore for every $v \in V(G')$, we keep a pointer to the first and last vertices of $M_v$ in the list of matched vertices (resp., in the list of unmatched vertices).
For any loop of the procedure, we iterate over four modules $M$, that is a constant.
Furthermore, since $|N_G(M)| \geq |V(G) \setminus M|-2$ then we only need to check three unmatched vertices of $V \setminus M$ in order to decide whether we can perform a {\sc MATCH} operation.
Note that we can skip scanning the unmatched vertices in $M$ using our pointer structure, so, it takes constant-time.
In the same way, we only need to consider three matched vertices of $V \setminus M$ in order to decide whether we can perform a {\sc SPLIT} operation.
Again, it takes constant-time.
Therefore, the claim is proved.

\smallskip
Let $F_{\max}$ be the matching so obtained.
By the above claim it takes ${\cal O}(|F_{\max}| - |F_0|)$-time to compute it with the above procedure.
Furthermore, we claim that $F_{\max}$ is maximum.
Suppose for the sake of contradiction that $F_{\max}$ is not a maximum matching.
By Lemma~\ref{lem:mw-matching-reduction}, $F_{\max}$ cannot be a maximum matching of $G^*$, obtained from $G$ by removing the edges in $(E(G[M_{v_1}]) \cup E(G[M_{v_x}])\cup E(G[M_{v_y}]) \cup E(G[M_{v_k}]))\setminus F^*$.
Let $P=(u_1,u_2,\ldots,u_{2\ell})$ be a {\em shortest} $F_{\max}$-augmenting path in $G^*$, that exists by Theorem~\ref{thm:berge}.

We prove as an intermediate subclaim that both $u_1,u_{2\ell}$ must be part of a same module amongst $M_{v_1},M_{v_k},M_x,M_y$.
Indeed, for every distinct $M,M' \in \{ M_{v_1},M_{v_k},M_x,M_y \}$, every vertex of $M$ is adjacent to every vertex of $M'$.
Furthermore, $V(F_{k-2}) \subseteq V(F_{\max})$ by construction and $v_{\left\lceil k/2 \right\rceil}$ (the only vertex of $\overline{P_{k-2}}$ possibly unmatched) is adjacent to every vertex of $U$.
Therefore, if the subclaim were false then $u_1,u_{2\ell}$ should be adjacent, hence they should have been matched together with a {\sc MATCH} operation.
A contradiction.
So, the subclaim is proved.

Let $M \in \{ M_{v_1},M_{v_k},M_x,M_y \}$ so that $u_1,u_{2\ell} \in M$.
Since $E(G^*[M])=F_M$, and $V(F_M) \subseteq V(F^*) \subseteq V(F_{\max})$ by construction, we have $u_2 \in N_G(M)$.
Furthermore, $u_3 \notin N_G(M)$ since otherwise, by considering $u_1,u_{2\ell} \in M$ and $u_2,u_3 \in N_G(M)$, we should have increased the matching with a {\sc SPLIT} operation.
In this situation, either $u_3 \in M$ or $u_3 \in V \setminus N_G[M]$.
We prove as another subclaim that $u_3,u_4 \in M$.
Indeed, suppose by contradiction $u_4 \in N_G(M)$.
In particular, $(u_1,u_4,u_5,\ldots,u_{2\ell})$ is a shorter augmenting path than $P$, thereby contradicting the minimality of $P$.
Therefore, $u_4 \notin N_G(M)$.
Moreover, if $u_3 \in V \setminus N_G[M]$ then, since the set $V \setminus N_G[M]$ induces a stable, we should have $u_4 \in N_G(M)$.
A contradiction.
So, $u_3 \in M$, and $u_4 \in N_G[M] \setminus N_G(M) = M$, that proves the subclaim.

The above subclaim implies $\{u_3,u_4\} \in F_M$.
Since $\{u_3,u_4\} \notin F_{\max}$, there exists a module $M'$ such that $u_2,u_5 \in M'$, and the edges $\{u_2,u_3\}, \ \{u_4,u_5\}$ have been obtained with a {\sc SPLIT} operation.
However, since $u_1,u_{2\ell} \in M$ are unmatched, and $M \subseteq N_G(M')$, we should have performed two {\sc MATCH} operations intead of performing a {\sc SPLIT} operation.
A contradiction.
Therefore, as claimed, $F_{\max}$ is a maximum matching of $G$.

\paragraph{Case $G$ is a spiked $p$-chain $Q_k$.}
For every $1 \leq i \leq \left\lceil k/2 \right\rceil$, let $V_i = \bigcup_{j \geq i} (M_{v_{2j-1}}\cup M_{v_{2j}} \cup M_{z_{2j-1}} \cup M_{z_{2j}})$ (by convention $M_{v} = \emptyset$ if vertex $v$ is not present).
Roughly, our algorithm tries to compute recursively a maximum matching for $G_i = G[V_i \cup U_{i-1}]$, where $U_{i-1}$ is a union of modules in $\{M_{v_{2i-2}}, \ M_{z_{2i-2}}\}$.  
Initially, we set $i=1$ and $U_0 = \emptyset$.
See Fig.~\ref{fig-p-tree-matching}.

\begin{figure}[h!]
\centering
\includegraphics[width=.5\textwidth]{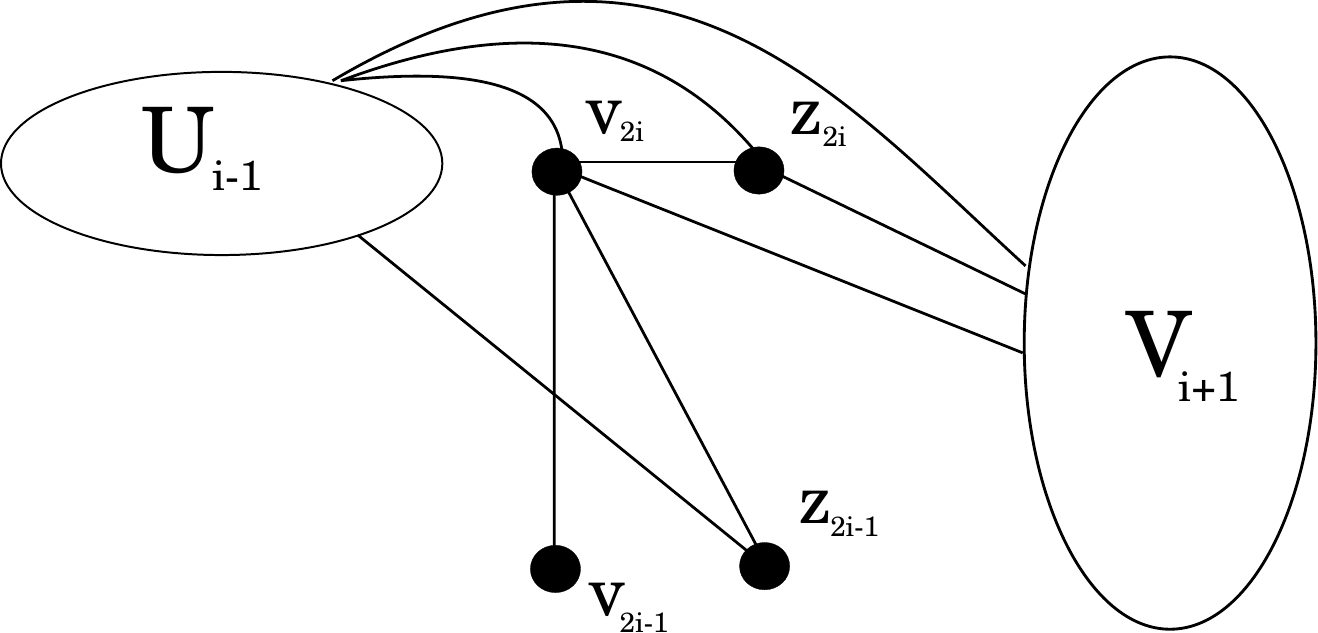}
\caption{Schematic view of graph $G_i$.}
\label{fig-p-tree-matching}
\end{figure}

If $i= \left\lceil k/2 \right\rceil$ then the quotient subgraph $G_i'$ has order at most six.
We can reuse the same techniques as for Theorem~\ref{thm:mw-maxmatching} in order to solve this case.
Thus from now on assume $i < \left\lceil k/2 \right\rceil$.
We need to observe that $v_{2i-1}$ is a pending vertex in the quotient subgraph $G_i'$, with $v_{2i}$ being its unique neighbour.
By Theorem~\ref{thm:stronger-mdc-qq3}, $v_{2i} \in V(G)$, hence $M_{v_{2i-1}}$ is a pending module of $G_i$.
Thus, we can apply the reduction rule of Lemma~\ref{lem:pending-module}.
Doing so, we can discard the set $S_i$, where $S_i = M_{v_{2i-1}} \cup \{v_{2i}\}$ if $F_{M_{v_{2i-1}}}$ is not a perfect matching of $G[M_{v_{2i-1}}]$, and $S_i = M_{v_{2i-1}}$ otherwise. 

Furthermore, in the case where $U_{i-1} \neq \emptyset$, there is now a complete join between $U_{i-1}$ and $V_i \setminus S_i$.
By Lemma~\ref{lem:join} we can compute a maximum matching of $G_i \setminus S_i$ from a maximum matching of $G[U_{i-1}]$ and a maximum matching of $G[V_i \setminus S_i]$.
In particular, since $U_{i-1}$ is a union of modules in $\{M_{v_{2i-2}}, \ M_{z_{2i-2}}\}$ and there is a complete join between $M_{v_{2i-2}}$ and $M_{z_{2i-2}}$, by Lemma~\ref{lem:join} a maximum matching of $G[U_{i-1}]$ can be computed from $F_{M_{v_{2i-2}}}$ and $F_{M_{z_{2i-2}}}$.
So, we are left to compute a maximum matching of $G[V_i \setminus S_i]$.

Then, there are two subcases.
If $v_{2i} \in S_i$ then $M_{z_{2i-1}}$ is disconnected in $G[V_i \setminus S_i]$.
Let $U_i = M_{z_{2i}}$.
The union of $F_{M_{z_{2i-1}}}$ with a maximum matching of $G_{i+1} = G[V_{i+1} \cup U_i]$ is a maximum matching of $G[V_i \setminus S_i]$.
Otherwise, $M_{z_{2i-1}}$ is a pending module of $G[V_i \setminus S_i]$ with $v_{2i}$ being its unique neighbour.
We apply the reduction rule of Lemma~\ref{lem:pending-module}.
Doing so, we can discard the set $T_i$, where $T_i = M_{z_{2i-1}} \cup \{v_{2i}\}$ if $F_{M_{z_{2i-1}}}$ is not a perfect matching of $G[M_{z_{2i-1}}]$, and $T_i = M_{z_{2i-1}}$ otherwise.
Let $U_i = M_{z_{2i}}$ if $v_{2i} \in T_i$ and $U_i = M_{z_{2i}} \cup M_{v_{2i}}$ otherwise.
We are left to compute a maximum matching of $G_{i+1} = G[V_{i+1} \cup U_i]$.
Overall, the procedure stops after we reach an empty subgraph, that takes ${\cal O}(|V(G')|)$ recursive calls. 

\paragraph{Case $G$ is a spiked $p$-chain $\overline{Q_k}$.}
Roughly, the case where $G'$ is isomorphic to a spiked $p$-chain $\overline{Q_k}$ is obtained by reverting the role of vertices with even index and vertices with odd index.
For every $1 \leq i \leq \left\lfloor k/2 \right\rfloor$, let $V_i = \bigcup_{j \geq i} (M_{v_{2j}} \cup M_{z_{2j}} \cup M_{v_{2j+1}} \cup M_{z_{2j+1}})$. 
Our algorithm tries to compute recursively a maximum matching for $G_i = G[V_i \cup U_{i-1}]$, where $U_{i-1}$ is a union of modules in $\{M_{v_{2i-1}}, \ M_{z_{2i-1}}\}$.  
Initially, we set $i=1$ and $U_0 = M_{v_1}$.

If $i= \left\lfloor k/2 \right\rfloor$ then the quotient subgraph $G_i'$ has order at most six.
We can reuse the same techniques as for Theorem~\ref{thm:mw-maxmatching} in order to solve this case.
Thus from now on assume $i < \left\lfloor k/2 \right\rfloor$.
We need to observe that $v_{2i}$ is a pending vertex in the quotient subgraph $G_i'$, with $v_{2i+1}$ being its unique neighbour.
By Theorem~\ref{thm:stronger-mdc-qq3}, $v_{2i+1} \in V(G)$, hence $M_{v_{2i}}$ is a pending module of $G_i$.
Thus, we can apply the reduction rule of Lemma~\ref{lem:pending-module}.
Doing so, we can discard the set $S_i$, where $S_i = M_{v_{2i}} \cup \{v_{2i+1}\}$ if $F_{M_{v_{2i}}}$ is not a perfect matching of $G[M_{v_{2i}}]$, and $S_i = M_{v_{2i}}$ otherwise.

Furthermore, in the case where $U_{i-1} \neq \emptyset$, there is now a complete join between $U_{i-1}$ and $V_i \setminus S_i$.
By Lemma~\ref{lem:join} we can compute a maximum matching of $G_i \setminus S_i$ from a maximum matching of $G[U_{i-1}]$ and a maximum matching of $G[V_i \setminus S_i]$.
In particular, since $U_{i-1}$ is a union of modules in $\{M_{v_{2i-1}}, \ M_{z_{2i-1}}\}$ and there is a complete join between $M_{v_{2i-1}}$ and $M_{z_{2i-1}}$, by Lemma~\ref{lem:join} a maximum matching of $G[U_{i-1}]$ can be computed from $F_{M_{v_{2i-2}}}$ and $F_{M_{z_{2i-2}}}$.
So, we are left to compute a maximum matching of $G[V_i \setminus S_i]$.

Then, there are two subcases.
If $v_{2i+1} \in S_i$ then $M_{z_{2i}}$ is disconnected in $G[V_i \setminus S_i]$.
Let $U_i = M_{z_{2i+1}}$.
The union of $F_{M_{z_{2i}}}$ with a maximum matching of $G_{i+1} = G[V_{i+1} \cup U_i]$ is a maximum matching of $G[V_i \setminus S_i]$.
Otherwise, $M_{z_{2i}}$ is a pending module of $G[V_i \setminus S_i]$ with $v_{2i+1}$ being its unique neighbour.
We apply the reduction rule of Lemma~\ref{lem:pending-module}.
Doing so, we can discard the set $T_i$, where $T_i = M_{z_{2i}} \cup \{v_{2i+1}\}$ if $F_{M_{z_{2i}}}$ is not a perfect matching of $G[M_{z_{2i}}]$, and $T_i = M_{z_{2i}}$ otherwise.
Let $U_i = M_{z_{2i+1}}$ if $v_{2i+1} \in T_i$ and $U_i = M_{z_{2i+1}} \cup M_{v_{2i+1}}$ otherwise.
We are left to compute a maximum matching of $G_{i+1} = G[V_{i+1} \cup U_i]$.
Overall, the procedure stops after ${\cal O}(|V(G')|)$ recursive calls. 
\end{proof}

\subsubsection*{Main result}

\begin{theorem}\label{thm:qq3-maxmatching}
For every $G=(V,E)$, {\sc Maximum Matching} can be solved in ${\cal O}(q(G)^4 \cdot n + m)$-time.
\end{theorem}

\begin{proof}
We generalize the algorithm for Theorem~\ref{thm:mw-maxmatching}.
In particular the algorithm is recursive.
If $G$ is trivial (reduced to a single node) then we output an empty matching.
Otherwise, let $G' = ({\cal M}(G),E')$ be the quotient graph of $G$.
For every module $M \in {\cal M}(G)$, we call the algorithm recursively on $G[M]$ in order to compute a maximum matching $F_M$ of $G[M]$.
Let $F^* = \bigcup_{M \in {\cal M}(G)} F_M$.
If $G'$ is either edgeless, complete or a prime graph with no more than $q(G)$ vertices then we apply the same techniques as for Theorem~\ref{thm:mw-maxmatching} in order to compute a maximum matching $F_{\max}$ for $G$.
It takes constant-time if $G'$ is a stable, ${\cal O}(q(G)^4\cdot(|F_{\max}|-|F^*|))$-time if $G'$ is prime and ${\cal O}(|V(G')|+ (|F_{\max}|-|F^*|))$-time if $G'$ is a complete graph.
Otherwise by Theorem~\ref{thm:stronger-mdc-qq3} the following cases need to be considered.
\begin{itemize}
\item Suppose $G$ is a disc. 
In particular, $G = G'$.
By Lemma~\ref{lem:disc}, we can compute a maximum matching for $G$ in ${\cal O}(|V(G')|+|E(G')|)$-time.
\item Suppose $G=(S \cup K \cup R, E)$ is a spider.
In particular, $G'=(S \cup K \cup R',E')$ is a prime spider.
By Lemma~\ref{lem:spider}, the union of $F_{R} = F^*$ with a perfect matching between $S$ and $K$ is a maximum matching of $G$.
It can be computed in ${\cal O}(|V(G')|+|E(G')|)$-time.

\item Otherwise $G'$ is a prime $p$-tree.
By Proposition~\ref{prop:maxmatching-prime-ptree}, a maximum matching $F_{\max}$ for $G$ can be computed in ${\cal O}(|V(G')|+|E(G')|+|F_{\max}|-|F^*|)$-time.
\end{itemize}
Overall, summing the order of all the subgraphs in the modular decomposition of $G$ amounts to ${\cal O}(n)$~\cite{Rao08b}.
Summing the size of all the subgraphs in the modular decomposition of $G$ amounts to ${\cal O}(n+m)$~\cite{Rao08b}.
Furthermore, a maximum matching of $G$ also has cardinality ${\cal O}(n)$.
Therefore, the total running time is in ${\cal O}(q(G)^4 \cdot n+m)$ if the modular decomposition of $G$ is given.
The latter decomposition can be precomputed in ${\cal O}(n+m)$-time~\cite{TCHP08}.
\end{proof}

\section{Applications to other graph classes}
\label{sec:applications}

Our algorithmic schemes in Sections~\ref{sec:dist} and~\ref{sec:maxmatching} are all based on preprocessing methods with either split decomposition or modular decomposition.
If the prime subgraphs of the decomposition have constant-size then the input graph has bounded clique-width.
However, when the prime subgraphs are ``simple'' enough w.r.t. the problem considered, we may well be able to generalize our techniques in order to apply to some graph classes with unbounded clique-width.
In what follows, we present such examples.

\smallskip
A graph is {\em weak bipolarizable} if every prime subgraph in its modular decomposition is a chordal graph~\cite{Ola89}.
Some cycle problems such as {\sc Girth} (trivially) and {\sc Triangle Counting} (by using a clique-tree) can be easily solved in linear-time for chordal graphs.
The latter extends to the larger class of weak bipolarizable graphs by using our techniques.

\smallskip
Another instructive example is the class of graphs with small prime subgraphs for {\em c-decomposition}.
The c-decomposition consists in successively decomposing a graph by the modular decomposition and the split decomposition until all the subgraphs obtained are either degenerate (complete, edgeless or star) or prime for both the modular decomposition and the split decomposition~\cite{Lan01}.
Let us call c-width the minimum $k \geq 2$ such that any prime subgraph in the c-decomposition has order at most $k$.
The following was proved in~\cite{Rao08}.

\begin{theorem}[~\cite{Rao08}]
The class of graphs with c-width $2$ ({\it i.e.}, completely decomposable by the c-decomposition) has unbounded clique-width.
\end{theorem}

It is not clear how to compute the c-decomposition in linear-time.
However, both the modular decomposition and the split decomposition of graphs with small c-width already have some interesting properties which can be exploited for algorithmic purposes.
Before concluding this section we illustrate this fact with {\sc Eccentricities}.

\begin{lemma}\label{lem:split-dec-cdec}
Let $G=(V,E)$ be a graph with c-width at most $k$ that is prime for modular decomposition.
Every split component of $G$ that is not degenerate either has order at most $k$ or contains a universal vertex.
\end{lemma}

\begin{proof}
Since $G$ has c-width at most $k$, every non degenerate split component of $G$ with order at least $k+1$ can be modularly decomposed.
We show in the proof of Lemma~\ref{lem:mw-to-sw} that if a non degenerate graph can be modularly decomposed and it does not contain a universal vertex then it has a split.
Therefore, every non degenerate split component of size at least $k+1$ contains a universal vertex since it is prime for split decomposition.
\end{proof}

We now revisit the algorithmic scheme of Theorem~\ref{thm:sw-ecc}.

\begin{proposition}\label{prop:diam-unbounded-cw}
For every $G=(V,E)$ with c-width at most $k$, {\sc Eccentricities} can be solved in ${\cal O}(k^2 \cdot n + m)$-time.
In particular, {\sc Diameter} can also be solved in ${\cal O}(k^2 \cdot n + m)$-time.
\end{proposition}

\begin{proof}
Let $G'=(V',E')$ be the quotient graph of $G$.
Note that $G'$ has c-width at most $k$.
Furthermore, by Theorem~\ref{thm:mw-ecc} the problem reduces in linear-time to solve {\sc Eccentricities} for $G'$.
We compute the split-decomposition of $G'$.
It takes linear-time~\cite{CDR12}.
By Lemma~\ref{lem:split-dec-cdec} every split component of $G'$ either has order at most $k$ or it has diameter at most $2$.

Let us consider the following subproblem for every split component $C_i$.
Given a weight function $e : V(C_i) \to \mathbb{N}$, compute $\max_{u \in V(C_i) \setminus \{v\}} dist_{C_i}(u,v) + e(u)$ for every $v \in C_i$.
Indeed, the algorithm for Theorem~\ref{thm:sw-ecc} consists in solving the above subproblem a constant-number of times for every split component, with different weight functions $e$ that are computed by tree traversal on the split decomposition tree.
In particular, if the above subproblem can be solved in ${\cal O}(k^2 \cdot |V(C_i)| + |E(C_i)|)$-time for every split component $C_i$ then we can solve {\sc Eccentricities} for $G'$ in ${\cal O}(k^2 \cdot |V(G')| + |E(G')|)$-time.

There are two cases.
If $C_i$ has order at most $k$ then the above subproblem can be solved in ${\cal O}(|V(C_i)||E(C_i)|)$-time, that is in ${\cal O}(k^2 \cdot |V(C_i)|)$.
Otherwise, by Lemma~\ref{lem:split-dec-cdec} $C_i$ contains a universal vertex, that can be detected in ${\cal O}(|V(C_i)| + |E(C_i)|)$-time.
In particular, $C_i$ has diameter at most two.
Let $V(C_i) = (v_1,v_2,\ldots,v_{|V(C_i)|})$ be totally ordered such that, for every $j < j'$ we have $e(v_j) \geq e(v_{j'})$.
An ordering as above can be computed in ${\cal O}(|V(C_i)|)$-time, for instance using a bucket-sort algorithm.
Then, for every $v \in V(C_i)$ we proceed as follows.
We compute $D_v = 1 + \max_{u \in N_{C_i}(v)} e(u)$.
It takes ${\cal O}(deg_{C_i}(v))$-time.
Then, we compute the smallest $j$ such that $v_j$ and $v$ are nonadjacent (if any).
Starting from $v_1$ and following the ordering, it takes ${\cal O}(deg_{C_i}(v))$-time.
Finally, we are left to compare, in constant-time, $D_v$ with $2 + e(v_i)$.
Overall, the subproblem is solved in ${\cal O}(|V(C_i)|+|E(C_i)|)$-time in this case.

Therefore, {\sc Eccentricities} can be solved in ${\cal O}(k^2 \cdot n + m)$-time for $G$.
\end{proof}

\bibliographystyle{abbrv}
\bibliography{bibliography}

\appendix

\section{Proof of Theorem~\ref{thm:stronger-mdc-qq3}}

Our proof in this section involves a refinement of modules, that is called {\em $p$-connected components}.
The notion of $p$-connectedness also generalizes connectivity in graphs.
A graph $G=(V,E)$ is $p$-connected if and only if, for every bipartition $(V_1,V_2)$ of $V$, there exists a path of length four with vertices in both $V_1$ and $V_2$.
The $p$-connected components of a graph are its maximal induced subgraphs which are $p$-connected. 
Furthermore, a $p$-connected graph is termed {\em separable} if there exists a bipartition $(V_1,V_2)$ of its vertex-set such that, for every crossing $P_4$, its two ends are in $V_2$ and its two internal vertices are in $V_1$. 
The latter bipartition $(V_1,V_2)$ is called a separation, and if it exists then it is unique.

We need a strengthening of Theorem~\ref{thm:modular-dec}:

\begin{theorem}[~\cite{JaO95}]\label{thm:primeval-dec}
For an arbitrary graph $G$ exactly one of the following conditions is satisfied.
\begin{enumerate}
\item $G$ is disconnected;
\item $\overline{G}$ is disconnected;
\item There is a unique proper separable $p$-connected component of $G$, with its separation being $(V_1,V_2)$ such that every vertex not in this component is adjacent to every vertex of $V_1$ and nonadjacent to every vertex of $V_2$;
\item $G$ is $p$-connected.
\end{enumerate}
\end{theorem}

If $G$ or $\overline{G}$ is disconnected then it corresponds to a degenerate node in the modular decomposition tree. So we know how to handle with the two first cases.
It remains to study the $p$-connected components of $(q,q-3)$-graphs.

For that, we need to introduce the class of $p$-trees:

\begin{definition}[~\cite{Bab00}]\label{def:p-trees}
A graph $G=(V,E)$ is a $p$-tree if one of the following conditions hold:
\begin{itemize}
\item the quotient graph $G'$ of $G$ is a $P_4$.
Furthermore, $G$ is obtained from $G'$ by replacing one vertex by a cograph.
\item the quotient graph $G'$ of $G$ is a spiked $p$-chain $P_k$, or its complement.
Furthermore, $G$ is obtained from $G'$ by replacing any of $x,y,v_1,v_k$ by a module inducing a cograph.
\item the quotient graph $G'$ of $G$ is a spiked $p$-chain $Q_k$, or its complement.
Furthermore, $G$ is obtained from $G'$ by replacing any of $v_1,v_k, z_2, z_3, \ldots, z_{k-5}$ by a module inducing a cograph.
\end{itemize}
\end{definition}

We stress that the case where the quotient graph $G'$ is a $P_4$, and so, of order $4 \leq 7 \leq q$ can be ignored in our analysis.
Other characterizations for $p$-trees can be found in~\cite{Bab98}.
The above Definition~\ref{def:p-trees} is more suitable to our needs.

\begin{theorem}[~\cite{BaO99}]\label{thm:p-comp}
A $p$-connected component of a $(q,q-3)$-graph either contains less than $q$ vertices, or is isomorphic to a prime spider, to a disc or to a $p$-tree.
\end{theorem}

Finally, before we can prove Theorem~\ref{thm:stronger-mdc-qq3}, we need to further characterize the {\em separable} $p$-connected components.
We use the following characterization of separable $p$-connected components.

\begin{theorem}[~\cite{JaO95}]\label{thm:separable-p-comp}
A $p$-connected graph $G=(V,E)$ is separable if and only if its quotient graph is a split graph.
Furthermore, its unique separation $(V_1,V_2)$ is given by the union $V_1$ of the strong modules inducing the clique and the union $V_2$ of the strong modules inducing the stable set. 
\end{theorem}

We are now ready to prove Theorem~\ref{thm:stronger-mdc-qq3}.

\begin{proofof}{Theorem~\ref{thm:stronger-mdc-qq3}}
Suppose $G$ and $\overline{G}$ are connected (otherwise we are done).
By Theorem~\ref{thm:primeval-dec} there are two cases.
First we assume $G$ to be $p$-connected.
By Theorem~\ref{thm:p-comp}, $G$ either contains less than $q$ vertices, or is isomorphic to a prime spider, to a disc or to a $p$-tree.
Furthermore, if $G$ is a $p$-tree then according to Definition~\ref{def:p-trees}, the nontrivial modules can be characterized.
So, we are done in this case.
Otherwise, $G$ is not $p$-connected.
Let $V = V_1 \cup V_2 \cup V_3$ such that: $H = G[V_1 \cup V_2]$ is a separable $p$-component with separation $(V_1,V_2)$, every vertex of $V_3$ is adjacent to every vertex of $V_1$ and nonadjacent to every vertex of $V_2$.
Note that $G'$ is obtained from the quotient graph $H'$ of $H$ by possibly adding a vertex adjacent to all the strong modules in $V_1$.
In particular, by Theorem~\ref{thm:separable-p-comp} $H'$ is a split graph, and so, $G'$ is also a split graph.
By Lemma~\ref{lem:reduce-qq3}, it implies that $G'$ is either a prime spider, a spiked $p$-chain $Q_k$, a spiked $p$-chain $\overline{Q_k}$, or a graph with at most $q$ vertices.
Furthermore, if $G'$ is a prime spider then by Theorem~\ref{thm:p-comp} so is $H$, hence $G$ is a spider.
Otherwise, $G'$ is either a spiked $p$-chain $Q_k$ or a spiked $p$-chain $\overline{Q_k}$.
It implies that $H$ is a $p$-tree.
In particular, the nontrivial modules in $H$ can be characterized according to Definition~\ref{def:p-trees}. 
The only nontrivial module of $G$ that is not a nontrivial module of $H$ (if any) contains $V_3$. 
Finally, since the module that contains $V_3$ has no neighbour among the modules in $V_2$, the corresponding vertex in the quotient can only be a $z_i$, for some $i$.
So, we are also done in this case.
\end{proofof}

\end{document}